\documentclass[12pt]{article}

\usepackage{amssymb,amsmath,amsthm,amsfonts,eurosym,geometry,ulem,graphicx,caption,color,setspace,sectsty,comment,footmisc,natbib,pdflscape,array,bbm,hyperref,mathrsfs,mathtools,enumerate,enumitem,indentfirst,todonotes,yhmath,bm,algorithm2e,cleveref,scalerel,stackengine,booktabs,threeparttable,multirow,float,xcolor,soul,tcolorbox,makecell,subcaption,longtable,placeins}

\stackMath
\newcommand\What[1]{%
\savestack{\tmpbox}{\stretchto{%
  \scaleto{%
    \scalerel*[\widthof{\ensuremath{#1}}]{\kern-.6pt\bigwedge\kern-.6pt}%
    {\rule[-\textheight/2]{1ex}{\textheight}}%
  }{\textheight}%
}{0.5ex}}%
\stackon[1pt]{#1}{\tmpbox}%
}

\usepackage[toc,page]{appendix}
\RestyleAlgo{ruled}
\SetKwComment{Comment}{/* }{ */}

\allowdisplaybreaks[1]

\newcounter{question}
\def\sym#1{\ifmmode^{#1}\else\(^{#1}\)\fi}

\newcommand{\eps}{\varepsilon}
\newcommand{\cQ}{\mathcal{Q}} %
\newcommand{\mA}{\mathcal{A}}

\newcommand{\mK}{\mathcal{K}}
\newcommand{\mL}{\mathcal{L}}

\newcommand{\mP}{\mathcal{P}}

\newcommand{\mV}{\mathcal{V}}

\newcommand{\mX}{\mathcal{X}}
\newcommand{\mY}{\mathcal{Y}}
\newcommand{\mZ}{\mathcal{Z}}

\newcommand{\bC}{\mathbb{C}}

\newcommand{\bE}{\mathbb{E}}

\newcommand{\bR}{\mathbb{R}}

\newcommand{\bfP}{\mathbf{P}}

\newcommand{\bfY}{\mathbf{Y}}

\newcommand{\argmin}{\operatornamewithlimits{argmin}}
\newcommand{\argmax}{\operatornamewithlimits{argmax}}

\newtheorem{theorem}{Theorem}
\newtheorem{assumption}{Assumption}
\newtheorem{proposition}{Proposition}
\newtheorem{corollary}{Corollary}
\newtheorem{lemma}{Lemma}
\newtheorem{definition}{Definition}

\theoremstyle{definition}
\newtheorem{example}{Example}

\newcommand{\itemsnewline}{\hfil\allowbreak\newline\vspace{-\dimexpr\baselineskip*2/3\relax}}

\newlist{assumptionitems}{enumerate}{1}
\setlist[assumptionitems]{label=(\roman*), ref=\theassumption(\roman*), align=left, leftmargin=2.2em, topsep=0pt, partopsep=0pt, before=\itemsnewline}
\Crefname{assumption}{Assumption}{Assumptions}
\crefname{assumptionitemsi}{assumption}{assumptions}

\newlist{lemmaitems}{enumerate}{1}
\setlist[lemmaitems]{label=(\roman*), ref=\thelemma(\roman*), align=left, leftmargin=2.2em, topsep=0pt, partopsep=0pt, before=\itemsnewline}
\Crefname{lemma}{Lemma}{Lemmas}
\crefname{lemmaitemsi}{lemma}{lemmas}

\newlist{propositionitems}{enumerate}{1}
\setlist[propositionitems]{label=(\roman*), ref=\theproposition(\roman*), align=left, leftmargin=2.2em, topsep=0pt, partopsep=0pt, before=\itemsnewline}
\Crefname{proposition}{Proposition}{Propositions}
\crefname{propositionitemsi}{proposition}{propositions}

\newlist{theoremitems}{enumerate}{1}
\setlist[theoremitems]{label=(\roman*), ref=\thetheorem(\roman*), align=left, leftmargin=2.2em, topsep=0pt, partopsep=0pt, before=\itemsnewline}
\Crefname{theorem}{Theorem}{Theorems}
\crefname{theoremitemsi}{theorem}{theorems}

\newlist{definitionitems}{enumerate}{1}
\setlist[definitionitems]{label=(\roman*), ref=\thedefinition(\roman*), align=left, leftmargin=2.2em, topsep=0pt, partopsep=0pt, before=\itemsnewline}
\Crefname{definition}{Definition}{Definitions}
\crefname{definitionitemsi}{definition}{definitions}

\theoremstyle{definition}
\newtheorem{remark}{Remark}

\newcolumntype{L}[1]{>{\raggedright\let\newline\\arraybackslash\hspace{0pt}}m{#1}}
\newcolumntype{C}[1]{>{\centering\let\newline\\arraybackslash\hspace{0pt}}m{#1}}
\newcolumntype{R}[1]{>{\raggedleft\let\newline\\arraybackslash\hspace{0pt}}m{#1}}

\newcommand{\continuation}{??}

\hypersetup{colorlinks=true,linkcolor=black,citecolor=blue,urlcolor=blue}

\newcommand{\blueref}[1]{\hyperref[#1]{\textcolor{blue}{\ref*{#1}}}}
\begin{document}

\title{Sequential Estimation of Dynamic Discrete Choice Models with Unobserved Heterogeneity\thanks{The use of Worldpanel by Numerator data does not imply the endorsement of Worldpanel by Numerator in relation to the interpretation or analysis of the data. All errors are our own.}}
\author{Ertian Chen\thanks{School of Economics \& Paula and Gregory Chow Institute for Studies in Economics, Xiamen University and IFS. Email: \href{mailto:ertianchen@xmu.edu.cn}{ertianchen@xmu.edu.cn}.}
\and
Hiroyuki Kasahara\thanks{Vancouver School of Economics, University of British Columbia. Email: \href{mailto:hkasahar@mail.ubc.ca}{hkasahar@mail.ubc.ca}.}
\and
Katsumi Shimotsu\thanks{Faculty of Economics, University of Tokyo. Email: \href{mailto:shimotsu@e.u-tokyo.ac.jp}{shimotsu@e.u-tokyo.ac.jp}.}
}
\date{\today}
\maketitle
\vspace{-1cm}
\begin{center}
\end{center}

\begin{abstract} 

\noindent Unobserved heterogeneity is empirically central in dynamic discrete choice models but computationally costly to incorporate: estimation requires repeatedly solving fixed-point equations for every latent type. We develop EM-NPL($q$), a unified framework combining the sequential pseudo-likelihood (NPL) and finite-mixture Expectation-Maximization (EM) algorithms, truncating the inner solver to $q$ iterations, and accommodating the Bellman, policy valuation, Euler, and Efficient Pseudo-Likelihood (EPL) equations, with step-by-step implementation guidance. For linear-in-parameters models estimated via the policy valuation or EPL equations, we establish truncation invariance: for any $q\geq 1$, EM-NPL($q$) is numerically identical to the fully converged EM-NPL estimator, so q affects computation but not statistical properties. We establish consistency, asymptotic normality, and local convergence. Truncation reduces runtime by up to 26\% for PV\_GMRES in single-agent simulations and 58\% for EPL in dynamic games. In a cola-demand application, ignoring unobserved heterogeneity understates own-price elasticities and soda-tax compensating variation by up to 85\% and 90\%.
\bigskip



\bigskip
\end{abstract}
\pagebreak \newpage

\onehalfspacing
\section{Introduction}

\noindent Dynamic discrete choice (DDC) models map panel data on discrete decisions into the primitives of dynamic optimization. These primitives support counterfactual and welfare analysis of policies whose effects unfold over time, such as a tax on a storable good. Agents and markets differ persistently in dimensions the econometrician does not observe, so persistent choices admit two explanations, true state dependence and unobserved heterogeneity, a distinction \citet{heckman1981heterogeneity} emphasized. The distinction matters for policy: under true state dependence a temporary intervention can have lasting effects, whereas stable unobserved differences imply none; conflating the two distorts estimated dynamics and counterfactuals. Such heterogeneity is quantitatively important across labor, demand estimation, and industrial organization: unobserved endowments account for 90 percent of the variance in lifetime utility in the career decisions of young men \citep{keane1997career}, with similar evidence for fertility and work among married women \citep{francesconi2002joint} and for dynamic demand \citep{osborne2011consumer,gowrisankaran2012dynamics,wang2015impact}. A finite mixture of latent types captures this heterogeneity tractably, and nonparametric results give conditions for identifying finite-mixture DDC models from panel data \citep{kasahara2009nonparametric,hu2012nonparametric}.

At the same time, incorporating unobserved heterogeneity is computationally costly because it requires repeatedly solving fixed-point equations for all unobserved types. The estimation problem has a nested structure: outer iterations update the structural parameters, the mixture weights, and the conditional choice probabilities (CCPs), and each outer iteration requires inner fixed-point solves for every latent type, so the cost of the inner solver multiplies across types and across outer iterations. This burden is substantial in single-agent models with large state spaces and is magnified further in dynamic games \citep{aguirregabiria2010dynamic,arcidiacono2011practical}. A natural question is whether one can reduce the cost of these inner fixed-point solvers without sacrificing statistical efficiency.

Existing approaches reduce this burden in different ways. Full-solution methods solve the Bellman equation at each parameter evaluation \citep{rust1987optimal}, while CCP-based estimators replace the dynamic programming problem with fixed-point problems in CCP spaces \citep{hotz1993conditional,aguirregabiria2002swapping,aguirregabiria2007sequential,arcidiacono2011conditional}. The Mathematical Programming with Equilibrium Constraints approach treats the fixed-point problem as a constraint in a joint optimization problem \citep{su2012constrained,dube2012improving}. Other approaches exploit special structure of the model, such as finite dependence \citep{arcidiacono2011conditional,aguirregabiria2023solving} and index invertibility \citep{bunting2025faster}. \citet{bugni2021iterated} study the truncation of outer iterations, \citet{adusumilli2025temporal} develop temporal-difference estimation, which approximates the value function with a sieve, and \citet{dearing2025efficient} propose the Efficient Pseudo-Likelihood (EPL) estimator for dynamic games. These methods substantially expand the range of feasible applications, yet they leave two gaps for a practitioner facing latent types. The first gap is practical: the tools were developed for particular fixed-point equations, largely in models without unobserved heterogeneity, and there is little guidance on how to combine the equation, the inner algorithm, and the truncation level once latent types enter the model. The inner maps of these equations differ in whether they are contractions and in whether they depend on the current parameter, so guidance developed for one equation does not transfer automatically to another. The second gap is theoretical: it is unknown how truncating the inner fixed-point solver affects the statistical properties of the resulting estimator, so any runtime saved on the inner loop comes with unquantified statistical risk.

This paper develops a unified framework for sequential estimation of DDC models with unobserved heterogeneity. We make two contributions. First, we propose the EM-NPL($q$) estimator, which combines the sequential pseudo-likelihood (NPL) framework of \citet{aguirregabiria2002swapping,aguirregabiria2007sequential} with the finite-mixture Expectation-Maximization (EM) approach of \citet{arcidiacono2011conditional}. At each outer iteration, the estimator applies only $q$ steps of a chosen inner algorithm to approximate the fixed-point object required in the M-step, and then updates the structural parameters and mixture weights. The framework accommodates different fixed-point equations, including the Bellman equation \citep{rust1987optimal}, policy valuation \citep{hotz1993conditional,aguirregabiria2007sequential}, the Euler equation \citep{aguirregabiria2023solving}, and EPL \citep{dearing2025efficient,fukasawa2025epl}, and it can be paired with different inner algorithms such as successive approximation, the Generalized Minimal Residual (GMRES) method, or Newton-type methods. When the model has a single latent type, the E-step is vacuous and the EM-NPL($q$) estimator reduces to the NPL($q$) estimator of \citet{aguirregabiria2007sequential}, so existing sequential estimators arise as special cases. In this framework, the fixed-point equation, the inner algorithm, and the truncation level $q$ are three distinct design choices, and the step-by-step implementation guidance described below shows how to set them jointly.

Second, we establish a truncation-invariance result for the workhorse class of DDC models with linear-in-parameters utility estimated via the policy valuation or EPL equations. In this class, the fixed-point map does not depend on the current structural parameter: it reduces to a parameter-independent linear system solved once per outer iteration rather than at every parameter evaluation. The observation is simple, and its consequence is exact. We show that the EM-NPL($q$) estimator, for any $q \geq 1$, is numerically identical to the EM-NPL estimator that solves the inner fixed-point problem until convergence. The identity holds in any finite sample and involves no asymptotic approximation. Consequently, the number of inner iterations and the choice of inner solver affect only computation, never statistical properties: there is no efficiency loss from inner truncation, and NPL standard errors remain valid without any adjustment. The result gives the question posed above a sharp answer: within this class, truncating the inner solver costs nothing statistically. It also decouples numerical implementation from asymptotic properties, and the practical guidance exploits this decoupling.

The implementation guidance proceeds in three stages. In the first stage, we benchmark candidate inner solvers on the initial fixed-point problem, before any outer iteration is run, and discard those that are dominated. In the second stage, we select the fixed-point equation and the inner algorithm from three model features. The first feature is the functional form of the utility function: linear-in-parameters utility reduces the policy valuation and EPL problems to linear systems suited to GMRES, whereas nonlinear-in-parameters utility calls for Newton-type methods on the Bellman equation. The second is finite dependence, which makes the Euler equation available \citep{aguirregabiria2023solving}; its contraction modulus is strictly below the discount factor $\beta$, so even successive approximation converges relatively quickly. The third is the discount factor itself, which governs the convergence rate of successive approximation: when $\beta$ is as low as 0.8, the choice of algorithm matters little, whereas a $\beta$ close to one makes successive approximation extremely slow. In the third stage, we choose the truncation level $q$, and the appropriate level is equation-specific. In our simulations, with superlinear solvers such as GMRES and Newton, $q = 4$ to $6$ usually suffices; successive approximation may need $q = 8$ to $10$, or more, when $\beta$ is large, since the inner contraction modulus approaches one as $\beta$ approaches one. EPL is a separate case: its inner map is not a contraction, and our simulations favor $q = 8$ to $12$ for computational speed. In practice, we recommend starting with a low $q$ and increasing it only if the outer loop fails to converge. Under linear-in-parameters utility, the truncation-invariance result makes $q$ a pure runtime choice; outside that class, $q$ balances inner accuracy against the number of outer iterations.

We support these contributions with a full asymptotic and convergence theory. We establish consistency and $\sqrt{N}$ asymptotic normality of the EM-NPL($q$) estimator. We characterize local convergence of the algorithm, whose key determinant is the spectral radius of the untruncated algorithm's iteration map. If the EM-NPL algorithm converges locally, then the EM-NPL($q$) algorithm also converges locally for sufficiently large $q$. We derive an iteration-count bound that links the number of outer iterations to the approximation error of the truncated inner solver, making precise how inner accuracy translates into outer-loop effort. Outside the linear-in-parameters case, we show that the asymptotic variance of EM-NPL($q$) differs from that of EM-NPL by a term of the order of the inner-solver approximation error, so the faster the inner algorithm converges, the faster this term vanishes in $q$. Together, these results clarify when inner truncation changes only runtime and when it may also affect convergence behavior.

Two earlier approaches are closest to our use of truncation. \citet{kasahara2012sequential} propose the $q$-NPL algorithm, which replaces the CCP mapping with its $q$-fold composition within each outer iteration. This improves the convergence of sequential estimators, and under a stability condition the efficiency approaches that of maximum likelihood as $q$ grows; see also \citet{kasahara2008pseudo} for sequential pseudo-likelihood estimation of single-agent models and \citet{kasahara2011sequential} for a value-function-mapping variant. In the $q$-NPL algorithm, $q$ counts exact evaluations of the fixed-point mapping itself; in ours, $q$ counts steps of a numerical solver applied to the fixed-point equation, inside an EM loop over latent types. We also cover four classes of fixed-point equations and, for the linear-in-parameters class, establish exact numerical invariance at every finite $q$ rather than gains that accrue as $q$ grows. \citet{bugni2021iterated} study truncation of the outer iteration sequence and derive the asymptotic properties of $k$-step estimators. Our truncation operates on the inner solver while holding the outer structure fixed, so the two forms of truncation are complementary: outer truncation reduces the number of pseudo-likelihood updates, whereas inner truncation reduces the cost of each update.

Two further comparisons sharpen the contribution. Relative to \citet{arcidiacono2011conditional}, whose finite-mixture EM approach our M-step embeds, we add truncated inner solvers, the menu of fixed-point equations, and convergence theory for the resulting algorithm. Relative to \citet{dearing2025efficient}, whose EPL estimator arises as one instance of our framework, we extend EPL to finite-mixture unobserved heterogeneity and quantify the computational effect of its inner truncation. These connections keep the framework close to existing practice: the additions concern how, and how accurately, the inner fixed-point problems are solved.

We apply EM-NPL($q$) to two Monte Carlo simulations: a single-agent entry--exit model with unobserved heterogeneity and a dynamic game. In the single-agent model, truncated PV\_GMRES implementations reduce runtime by about 7--26\% relative to full inner convergence and are often several times faster than standard alternatives. In dynamic games, policy valuation GMRES is the fastest method. The EPL estimator, even with its own inner truncation, is substantially slower than the policy valuation approach but achieves a smaller mean squared error. At well-chosen truncation levels, EPL\_GMRES reduces runtime by about 58\% relative to full-convergence EPL\_GMRES. These simulations underpin the three-stage guidance: they reveal which solvers are dominated, how the preferred equation and algorithm vary with model features, and which truncation levels are sensible for each equation.

We then apply EM-NPL($q$) to household demand for cola using Worldpanel by Numerator Take Home data. We estimate up to ten latent types separately for three income groups. Computational time is approximately linear in the number of types, and estimating a ten-type model takes about 4.8 minutes on average. Ignoring unobserved heterogeneity understates long-run own-price elasticities by up to about 85\%, short-run own-price elasticities by up to about 85\%, and compensating variation from a soda tax by up to about 90\%. These magnitudes show that unobserved heterogeneity is quantitatively important, while the runtime demonstrates that accounting for it is computationally feasible at a realistic scale.

\noindent \textbf{Notation.} All limits below are taken as $N \rightarrow \infty$, unless stated otherwise. Let $\coloneqq$ denote ``equals by definition.'' Let $\mathbbm{1}\{A\}$ denote an indicator function that takes the value of one when $A$ is true and zero otherwise. Let $\Delta(\mA)$ denote the probability simplex over $\mA$. For vectors, $\|\cdot\|$ denotes the Euclidean norm and $\|\cdot\|_{\infty}$ denotes the sup norm. For matrices, $\|\cdot\|_{F}$ denotes the Frobenius norm. For a square matrix $A$, $\mathrm{spec}(A)$ denotes its spectrum (the set of eigenvalues) and $\rho(A) \coloneqq \max_{\lambda \in \mathrm{spec}(A)}|\lambda|$ its spectral radius. With a slight abuse of notation, for (column) vectors $a_1, \ldots, a_k$, we write $a = (a_1, \ldots, a_k)$ to denote the stacked column vector $a = (a_1^{\top}, \ldots, a_k^{\top})^{\top}$.

\noindent \textbf{Outline.} \Cref{sec: model} describes the finite-mixture dynamic discrete choice model. \Cref{sec: estimator} describes the EM-NPL($q$) estimator. \Cref{sec: theory} presents the theoretical properties. \Cref{sec: practical guidance} provides implementation guidance. \Cref{sec: simulation} provides Monte Carlo simulations. \Cref{sec:empirical} presents the cola-demand application. \Cref{sec: conclusion} concludes. Appendix~\ref{sec:proofs} collects auxiliary lemmas, all proofs, and additional implementation details.

\section{Model} \label{sec: model}

\noindent We consider a stationary dynamic discrete choice (DDC) model with unobserved heterogeneity following \citet{kasahara2009nonparametric} and \citet{arcidiacono2011conditional}. Heterogeneity is modeled as a finite mixture with $M$ unobserved types.

A panel of $N$ markets is observed, each with $J$ firms indexed by $j = 1, \ldots, J$ and observed for $T$ periods.  The unobserved heterogeneity operates at the market level: market $i$'s type $m_i \in \{1,\ldots,M\}$ is drawn once at the start and remains fixed throughout the sample. Let $\pi_m = \Pr(m_i = m)$ denote the mixing weight, with $\pi_m > 0$ for all $m$ and $\sum_{m=1}^{M} \pi_m = 1$.  In what follows, we suppress the market index $i$.

Time is discrete and infinite, indexed by $t = 0, 1, 2, \ldots$.  Denote the action space by $a \in \mA = \{0, 1, \ldots, |\mA|-1\}$, where action 0 denotes the outside option.  In each period $t$, given a vector of state variables $x_{t} \in \mX$ observable to all firms, and firm $j$'s private information $\eps_{jt}\coloneqq(\eps_{jt}(a))_{a \in \mA} \in \bR^{|\mA|}$, firm $j$ chooses an action $a_{jt} \in \mA$ simultaneously with other firms.  The state variables follow a first-order controlled Markov process with transition density $f(x_{t+1}, \eps_{t+1} \mid x_{t}, a_{t}, \eps_{t}; \theta_{f})$, where $a_{t} = (a_{1t}, \ldots, a_{Jt})$ collects the actions of all firms in period $t$, $\eps_{t} = (\eps_{1t}, \ldots, \eps_{Jt})$, and $\theta_{f} \in \Theta_{f}$ is finite-dimensional.

Each firm in a market of type $m$ has flow utility $U_{j}(x_{t}, a_{t}, \eps_{jt}; \theta^{m})$, where $\theta^{m} \in \Theta \subseteq \bR^{d}$ is a vector of structural parameters indexed by the type $m \in \{1, \ldots, M\}$. Firm $j$ chooses its actions to maximize expected discounted utility, 
\begin{equation*}
    \bE\left[ \sum_{t=0}^{\infty} \beta^{t} U_{j}(x_{t}, a_{t}, \eps_{jt}; \theta^{m}) \mid x_{0}, \eps_{j0} \right],
\end{equation*}
where $\beta \in (0, 1)$ is the discount factor.  When $M = 1$, this reduces to the standard single-type DDC model.  We impose the following standard assumptions.

\begin{assumption} \label{assump: standard}
\begin{assumptionitems}
    \item (Additive separability) $U_{j}(x_{t}, a_{t}, \eps_{jt}; \theta^{m}) = U_{j}(x_{t}, a_{t}; \theta^{m}) + \eps_{jt}(a_{jt})$.
    \item \label{assump: conditional} (Conditional independence) $f(x_{t+1}, \eps_{t+1}|x_{t}, a_{t}, \eps_{t}; \theta_{f}) = f_{x}(x_{t+1}|x_{t}, a_{t}; \theta_{f})f_{\eps}(\eps_{t+1})$.
    \item (Independent Private Values) Private values are i.i.d.\ across firms, time, markets, and actions. Moreover, they follow a Type I extreme value distribution.
    \item (Finite state space) The observable state space $\mX$ is finite.
\end{assumptionitems}
\end{assumption}
The transition density is common across types, and $\theta_{f}$ can be estimated separately from $\theta^m$.  Thus, we focus on the estimation of $\theta^{m}$ and suppress the dependence on $\theta_{f}$.

We consider Markov Perfect Equilibrium, in which each firm's strategy depends only on the current state variables.  We restrict attention to stationary equilibria and suppress the time subscript $t$.  For each type $m$, let $\sigma^{m} \coloneqq(\sigma_{1}^{m}, \ldots, \sigma_{J}^{m})$ be a strategy profile where $\sigma_{j}^{m}: \mX \times \bR^{|\mA|} \to \mA$ maps the current state variables and private information to an action.  The Conditional Choice Probability (CCP) for firm $j$ in a market of type $m$ is
\begin{equation*}
    P_{j}^{\sigma^{m}}(a|x) = \int \mathbbm{1}\{ \sigma_{j}^{m}(x, \eps) = a \} f_{\eps}(\eps) d\eps .
\end{equation*}

For a market of type $m$ and firm $j$, let $\sigma_{-j}^{m}$ be the strategy profile of the other firms, and let $P_{-j}^{\sigma^{m}}$ be the corresponding CCP profile. The expected utility and transition density for firm $j$ are
\begin{align*}
    & U_{j}(x, a_{j}; \theta^{m}, P_{-j}^{\sigma^{m}}) = \sum_{a_{-j} \in \mA^{J-1}} U_{j}(x, a_{j}, a_{-j}; \theta^{m}) \prod_{k \neq j} P_{k}^{\sigma^{m}}(a_{k}|x), \\
    & f_{x}(x'|x, a_{j}; P_{-j}^{\sigma^{m}}) = \sum_{a_{-j} \in \mA^{J-1}} f_{x}(x'|x, a_{j}, a_{-j}) \prod_{k \neq j} P_{k}^{\sigma^{m}}(a_{k}|x).
\end{align*}

Under the extreme value assumption on private values, the \textit{integrated value function} for firm $j$ given $\sigma^{m}_{-j}$ satisfies the Bellman equation
\begin{equation*}
    V_{j}(x; \theta^{m}, P_{-j}^{\sigma^{m}}) = \log \left( \sum_{a \in \mA} \exp\left( v_{j}(x, a; \theta^{m}, P_{-j}^{\sigma^{m}})\right) \right) + \kappa ,
\end{equation*}
where $\kappa$ is Euler's constant, and $v_{j}(x, a; \theta^{m}, P_{-j}^{\sigma^{m}})$ is the conditional value function for firm $j$ defined as
\begin{equation*}
    v_{j}(x, a; \theta^{m}, P_{-j}^{\sigma^{m}}) \coloneqq U_{j}(x, a; \theta^{m}, P_{-j}^{\sigma^{m}}) + \beta \sum_{x' \in \mX} f_{x}(x'|x, a; P_{-j}^{\sigma^{m}}) V_{j}(x'; \theta^{m}, P_{-j}^{\sigma^{m}}).
\end{equation*}

The following definition characterizes a Markov Perfect Equilibrium, a strategy profile in which each firm's strategy is a best response given the others' strategies.

\begin{definition}[Markov Perfect Equilibrium]
A strategy profile $\sigma^{m*} = (\sigma_{1}^{m*}, \ldots, \sigma_{J}^{m*})$ is a Markov Perfect Equilibrium for type $m$ if for each firm $j$ and for all $(x, \eps) \in \mX \times \bR^{|\mA|}$,
\begin{equation*}
    \sigma_{j}^{m*}(x, \eps) \in \argmax_{a \in \mA} \left\{ v_{j}(x, a; \theta^{m}, P_{-j}^{\sigma^{m*}}) + \eps(a) \right\}.
\end{equation*}
\end{definition}
Under the Type I extreme value assumption (Assumption~\ref{assump: standard}(iii)), the equilibrium CCPs take the logit form
\begin{equation} \label{eq: ccp}
    P_{j}^{\sigma^{m*}}(a|x) = \frac{\exp\left(v_{j}(x, a; \theta^{m}, P_{-j}^{\sigma^{m*}}) \right)}{\sum_{a' \in \mA} \exp\left( v_{j}(x, a'; \theta^{m}, P_{-j}^{\sigma^{m*}}) \right)}.
\end{equation}

\section{Framework} \label{sec: estimator}

\noindent Following the sequential estimation approach of \citet{aguirregabiria2002swapping,aguirregabiria2007sequential} and the Expectation--Maximization (EM) framework of \citet{arcidiacono2011conditional}, we define the EM-NPL($q$) estimator as a fixed point of a system combining pseudo-likelihood maximization with fixed-point equations. The conventional approach solves the fixed-point equations exactly; our estimator replaces the exact solver with $q$ iterations of an inner algorithm. We first introduce the fixed-point equations that arise in estimation, present the inner algorithms used to solve them, and define the EM-NPL($q$) estimator.

\subsection{Fixed-Point Equations for Estimation} \label{sec: em npl m step}

\noindent In this section, we suppress the type index $m$ for notational simplicity. Let $P_j \coloneqq \{P_j(a|x): a \in \mA, x \in \mX\}$, and let $P \coloneqq (P_{1}, \ldots, P_{J}) \in \mP$ stack the CCPs across firms. At each maximization step, given $(P, \tilde{\theta}, \tilde Y)$ from the previous iteration, we solve for the nuisance parameter $Y \in \mY$ (for example, the value function) that satisfies the fixed-point equation
\begin{equation} \label{eq: fixed point general}
    Y = G(\theta, Y, P, \tilde{\theta}, \tilde Y).
\end{equation}
Here $\tilde{\theta}$ and $\tilde Y$ are the structural parameter and nuisance value from the previous iteration, and, for linearization-based mappings, $\tilde Y$ is the center about which the mapping is built. A given mapping need not depend on all arguments: the policy valuation mapping of \Cref{ex: policy valuation} is $G(\theta, Y, P)$ and the single-agent Bellman mapping of \Cref{ex: bellman equation} is $G(\theta, Y)$, whereas the EPL mapping of \Cref{ex: epl} is $G(\theta, Y, \tilde{\theta}, \tilde Y)$ and carries no exogenous $P$. As the examples below illustrate, different choices of $G$ and $Y$ give rise to different estimators.

\begin{example}[Policy Valuation and Value Function] \label{ex: policy valuation}
As shown in \citet{aguirregabiria2002swapping,aguirregabiria2007sequential}, the value function for firm $j$ satisfies
\begin{equation}\label{eq:Bellman}
    V_{j}(x) = \sum_{a \in \mA} P_{j}(a|x) \Bigl[ U_{j}(x, a; \theta, P_{-j}) + \kappa - \log P_{j}(a|x) \Bigr] + \beta \sum_{x' \in \mX} f^P(x'|x) V_{j}(x'),
\end{equation}
where $f^P(x'|x) \coloneqq \sum_{a \in \mA^J} (\prod_{j=1}^J P_{j}(a_j|x)) f_{x}(x'|x, a)$ is the transition probability of $x$ induced by $P$. Let $U_{j}^{P}(x; \theta) \coloneqq \sum_{a \in \mA} P_{j}(a|x) \bigl[ U_{j}(x, a; \theta, P_{-j}) + \kappa - \log P_{j}(a|x) \bigr]$ be the expected flow payoff under $P$ inclusive of the entropy term, and let $U_j^P(\theta) \coloneqq \{U_j^P(x; \theta): x \in \mX\}$. Let $Y = (V_{1}, \ldots, V_{J})$ with $V_j \coloneqq \{V_j(x): x \in \mX\}$, and let $F^P$ denote the matrix with entries $f^P(x'|x)$, so that \eqref{eq:Bellman} can be written as $(I - \beta F^P) V_j = U_j^P(\theta)$. Then, the above can be written as
\begin{equation*}
    Y = G(\theta, Y, P),
\end{equation*}
where $G$ stacks the right-hand side across firms and states.  This is the fixed-point equation used in the NPL estimator of \citet{aguirregabiria2007sequential}.

When the utility function is linear in parameters, $U_{j}(x,a;\theta, P_{-j}) = \phi_{j}(x,a; P_{-j})^{\top}\theta = \sum_{\ell=1}^{d_\theta} \phi_{j,\ell}(x,a; P_{-j})\,\theta_\ell$, we can remove $\theta$ from the fixed-point map entirely, because $\theta$ enters the policy valuation equation only through the flow payoff, additively and linearly, while the other terms depend on $P$ alone. We can write the stacked Bellman equation \eqref{eq:Bellman} as $V_j = \sum_{\ell=1}^{d_\theta}  \bar\phi_{j,\ell}\theta_\ell + \bar\phi_{j,P} + \beta F^P V_j$, where $\bar{\phi}_{j,\ell}$ stacks $\sum_{a \in \mA} P_{j}(a|x)\,\phi_{j,\ell}(x,a; P_{-j})$ over $x \in \mX$ and $\bar{\phi}_{j,P}$ stacks $\sum_{a \in \mA}P_{j}(a|x)(\kappa- \log P_{j}(a|x))$ over $x \in \mX$. We then solve for the value contribution of $\bar\phi_{j,\ell}$ and $\bar\phi_{j,P}$ separately via the linear systems
\begin{align*}
    W_{j,\ell} & = \bar{\phi}_{j,\ell} + \beta F^P W_{j,\ell}, \qquad \ell = 1,\ldots,d_{\theta}, \\
    W_{j,P} & = \bar{\phi}_{j,P} + \beta F^P W_{j,P}.
\end{align*}
All $d_\theta + 1$ systems share the operator $(I - \beta F^P)$ and differ only in their forcing terms. Reassembling these components recovers the value function as an affine function of $\theta$ as
\[
    V_{j} = \sum_{\ell=1}^{d_{\theta}}  W_{j,\ell} \theta_{\ell} + W_{j,P},
\]
with coefficients that depend only on $P$. Stacking $(W_{j,\ell}, W_{j,P})$ across components and firms into $Y$, the fixed-point equation becomes
\begin{equation} \label{eq: linear policy value}
    Y = G(Y, P),
\end{equation}
in which $G$ does not depend on $\theta$. Hence we solve \eqref{eq: linear policy value} only once per outer iteration, and evaluating $V_j$ for any candidate $\theta$ reduces to forming the linear combination above. At each maximization step we solve $MJ(d_{\theta}+1)$ linear systems of size $|\mX|$; since the $d_\theta + 1$ systems for each firm and type share the operator $(I - \beta F^P)$, a single factorization per firm-type is reused across all right-hand sides.
\end{example}

\begin{example}[Efficient Pseudo-Likelihood] \label{ex: epl}
\citet{dearing2025efficient} propose the EPL estimator for dynamic discrete games. The key insight is to use a quasi-Newton step to obtain the zero-Jacobian property that guarantees efficiency and convergence of the sequential estimator. The EPL estimator can work on the CCP space or the conditional value function space. As shown in \citet[Section 3.3]{dearing2025efficient}, working on the conditional value function space can exploit the linearity of the utility function to achieve tremendous computational simplicity. We also focus on the conditional value function space in this example. For firm $j$, define the mapping $\Phi_{j}$ as
\begin{equation} \label{eq:phi_j}
    \Phi_{j}(x, a; \theta, v) \coloneqq U_{j}(x, a; \theta, P_{-j}(v)) + \beta \sum_{x' \in \mX} f_{x}(x' \mid x, a; P_{-j}(v))\, S(v_{j}(x'))
\end{equation}
where $P_{-j}(v)$ denotes the CCP profile obtained from $v$ via the logit formula \eqref{eq: ccp}, and $S(v_{j}(x')) = \log\bigl(\sum_{a' \in \mA} \exp(v_{j}(x', a'))\bigr)+\kappa$ is the social surplus function. Stacking across all $(j, x, a)$, define $\Phi: \Theta \times \bR^{J|\mX||\mA|} \to \bR^{J|\mX||\mA|}$, so that equilibrium conditional value functions satisfy $v = \Phi(\theta, v)$. Let $\nabla_v \Phi(\theta, v)$ denote the Jacobian of $\Phi$ with respect to $v$.

Given the previous iterates $(\tilde{\theta}, \tilde{v})$ and current candidate $\theta$, EPL updates the conditional value function via a quasi-Newton step
\begin{equation} \label{eq:newton_EPL}
    v = \tilde{v} - \left[I - \nabla_v \Phi(\tilde{\theta}, \tilde{v})\right]^{-1} \left(\tilde{v} - \Phi(\theta, \tilde{v})\right).
\end{equation}
Defining $Y=v$ and $\tilde Y=\tilde v$, we can write \eqref{eq:newton_EPL} as $[I - \nabla_v \Phi(\tilde{\theta}, \tilde{Y})] Y =  \Phi(\theta, \tilde{Y}) - \nabla_v \Phi(\tilde{\theta}, \tilde{Y}) \tilde{Y}$.
With $G$ defined by
\begin{equation} \label{eq:fp_EPL}
G(\theta, Y, \tilde{\theta}, \tilde{Y}) = [\Phi(\theta, \tilde{Y})- \nabla_v \Phi(\tilde{\theta}, \tilde{Y}) \tilde{Y}] + \nabla_v \Phi(\tilde{\theta}, \tilde{Y}) Y ,
\end{equation}
\eqref{eq:newton_EPL} is equivalent to the fixed-point equation $Y = G(\theta, Y, \tilde{\theta}, \tilde{Y})$.

If the utility function is linear in parameters, $U_{j}(x,a;\theta, P_{-j}(v)) = \sum_{\ell=1}^{d_\theta} \phi_{j,\ell}(x,a; P_{-j}(v))\,\theta_\ell$, the mapping $\Phi(\theta,v)$ becomes linear in $\theta$: $\Phi(\theta,v) = \sum_{\ell=1}^{d_\theta} \Phi_{U,\ell}\theta_\ell +  \Phi_{S}(v)$, where $\Phi_{U,\ell}$ stacks $\phi_{j,\ell}(x,a; P_{-j}(v))$ across all $(j, x, a)$, and $\Phi_{S}(v)$ stacks the second term on the right-hand side of \eqref{eq:phi_j} across all $(j, x, a)$. Then, similarly to \Cref{ex: policy valuation}, we can define the fixed-point equation \eqref{eq:fp_EPL} in terms of a system of equations that does not depend on $\theta$. At each maximization step, the EPL needs to solve $M (d_{\theta}+1)$ linear systems of equations of size $J|\mX||\mA|$, which involves larger linear systems than those in \Cref{ex: policy valuation}.

\end{example}

\begin{example}[Bellman Equation and Value Function] \label{ex: bellman equation}
For the single-agent ($J=1$) DDC model, the Bellman equation gives
\begin{equation} \label{eq: bellman}
    V(x) = \log \biggl(\sum_{a \in \mA} \exp\Bigl(U(x,a; \theta) + \beta \sum_{x' \in \mX} f_{x}(x'|x, a)\, V(x') \Bigr) \biggr) + \kappa,
\end{equation}
which can be written as a fixed-point equation $Y = G(\theta, Y)$, where $Y$ stacks $V(x)$ across all $x \in \mX$. For dynamic games, given the CCP profile $P_{-j}$ of other firms, the Bellman equation can also be used to solve for firm $j$'s value function, giving $Y = G(\theta, Y, P)$.
\end{example}

\begin{example}[Euler Equation and Conditional Value Function Difference] \label{ex: euler}
    For single-agent DDC, \citet{aguirregabiria2023solving} propose the Euler fixed-point mapping based on the finite-dependence property, whose Lipschitz constant is strictly less than the discount factor $\beta$.  Let the state variable be $x_{t} = (a_{t-1},z_{t})$ where $a_{t-1}$ is the action taken in the previous period and $z_{t} \in \mZ$ is an exogenous state variable.  As the only endogenous state variable is the lagged action, the model has the two-period finite-dependence property.  That is, conditional on any pair of choice paths $(a_{t}, a_{t+1})$ and $(a'_{t}, a'_{t+1})$, the distribution of $x_{t+2}$ is the same as long as $a_{t+1} = a'_{t+1}$.  \citet{aguirregabiria2023solving} show that the conditional value function difference, $\tilde{v}(x,a) \coloneqq v(x,a) - v(x,0)$, is the solution to the Euler equation
    \begin{equation} \label{eq: euler}
        \tilde{v}(x,a) = c(x,a;\theta) + \beta \sum_{z' \in \mZ} \bigl(S(\tilde{v}(a,z')) - S(\tilde{v}(0,z')) \bigr)\, f_{z}(z'|z),
    \end{equation}
    where $\tilde{v}(x)$ is the vector of $\tilde{v}(x,a)$ for all $a \in \mA$, $S(\tilde{v}(x)) = \log \bigl(\sum_{a \in \mA} \exp(\tilde{v}(x,a))\bigr)+\kappa$ is the social surplus function, and
    \begin{equation*}
        c(x,a;\theta) \coloneqq U(x,a;\theta) - U(x,0;\theta) + \beta \sum_{z' \in \mZ} \bigl(U(a,z',0;\theta) - U(0,z',0;\theta) \bigr)\, f_{z}(z'|z).
    \end{equation*}
Let $Y$ stack $\tilde{v}(x,a)$ across all $(x,a)$; then the Euler equation \eqref{eq: euler} can be written as $Y = G(\theta, Y)$.
\end{example}

\subsection{The Fixed-Point Solver $\Gamma^q$} \label{sec: inner algorithms}

\noindent Many estimators for DDC models estimate $\theta$ by maximizing an objective function in which $Y$ is determined by the fixed-point equation \eqref{eq: fixed point general}. Solving \eqref{eq: fixed point general} to full convergence is expensive because it must be repeated at every evaluation of the objective function as $\theta$ varies. The EM-NPL($q$) estimator substitutes a $q$-step iterative algorithm 
\begin{equation}\label{eq:gamma_q_defn}
\Gamma^{q}(\theta, Y^{(0)}, P, \tilde{\theta}) 
\end{equation}
into the objective function in place of the exact solution, where $Y^{(0)}$ is an initial guess for $Y$. Each solver sets the center equal to the initial guess, $\tilde Y = Y^{(0)}$; hence $\Gamma^q$ takes no separate $\tilde Y$ argument. A given algorithm $\Gamma^{q}$ need not depend on all arguments. We present three such algorithms below; for a comprehensive reference, see \citet{judd1998numerical}.
\begin{example}[Successive Approximation]
Successive approximation (SA) is universally applicable to all fixed-point equations in \Cref{sec: em npl m step} and requires only matrix--vector multiplications per iteration. It iterates the mapping $G$ in its second argument, starting from $Y^{(0)}$ and holding the center $\tilde Y = Y^{(0)}$ fixed throughout
\begin{equation*}
    \Gamma^{q}_{\mathrm{SA}}(\theta, Y^{(0)}, P, \tilde{\theta}) = \underbrace{G(\theta, G(\theta, \cdots G(\theta, Y^{(0)}, P, \tilde{\theta}, Y^{(0)}) \cdots , P, \tilde{\theta}, Y^{(0)}), P, \tilde{\theta}, Y^{(0)})}_{q \text{ times}}.
\end{equation*}
For the policy valuation, Bellman, and Euler mappings (Examples~\ref{ex: policy valuation}, \ref{ex: bellman equation}, and~\ref{ex: euler}), $G$ does not depend on $\tilde Y$, so the fifth argument is inert and $\Gamma^q_{\mathrm{SA}}$ reduces to ordinary SA. The convergence speed of SA depends on the contraction modulus of $G$. When $\beta$ is close to 1, SA may require a large $q$ to converge.
\end{example}

\begin{example}[Generalized Minimal Residual Method (GMRES)] \label{ex: linear equation solver}
In Examples \ref{ex: policy valuation} and \ref{ex: epl}, the fixed-point equation reduces to a linear system $AY = b$, where $A = I - \beta F^P$ and $b$ is $U_j^P(\theta)$ for policy valuation, or $A = I - \nabla_v \Phi(\tilde{\theta}, \tilde{Y})$ and $b = \Phi(\theta, \tilde{Y}) - \nabla_v \Phi(\tilde{\theta}, \tilde{Y}) \tilde{Y}$ for the EPL mapping. In the latter case, $A$ and $b$ depend on $\tilde{Y}$, and the solver sets $\tilde{Y} = Y^{(0)}$.

The Generalized Minimal Residual (GMRES) method is a Krylov subspace method for large linear systems. Its convergence is fast when the eigenvalues of $A$ are clustered, as they are for $A = I - \beta F^P$ unless $\beta$ is very close to one, and it is typically much faster than SA. The cost per iteration is comparable to SA, as it requires only one matrix--vector product with $A$ and a few vector operations.
Given $Y^{(0)}$ and residual $r_0 = b - AY^{(0)}$, GMRES builds an orthonormal basis for the Krylov subspace $\mK_q(A, r_0) = \mathrm{span}\{r_0, Ar_0, \ldots, A^{q-1}r_0\}$ via the Arnoldi process and returns the iterate $Y^{(q)}$ that minimizes $\|AY^{(q)} - b\|$ over the affine subspace $Y^{(0)} + \mK_q$ as
\begin{equation*}
    \Gamma^q_{\mathrm{GMRES}}(\theta, Y^{(0)}, P, \tilde{\theta}) = \argmin_{Y \,\in\, Y^{(0)} + \mK_q(A,\, r_0)} \|AY - b\|, \qquad r_0 = b - AY^{(0)}.
\end{equation*}
Since $\mK_q(A,r_0) \subseteq \mK_{q+1}(A,r_0)$, the residual $\|AY^{(q)} - b\|$ is nonincreasing in $q$, and $\Gamma^{q}_{\mathrm{GMRES}}$ with $q \geq \dim(Y)$ recovers the exact solution $A^{-1}b$ in exact arithmetic. See \citet{saad2003iterative} for a comprehensive review of other iterative solvers. Standard GMRES is not differentiable when its initial residual is zero because the first Arnoldi normalization is then undefined. The implementation in \Cref{sec: gmres algorithm} therefore uses a locally smoothed GMRES map that replaces the Arnoldi recursion by a differentiable convergent iteration only near a zero residual. This modification leaves every exact solution unchanged and makes the solver map three times continuously differentiable on a neighborhood of the fixed point, so that $\Gamma^q_{\mathrm{GMRES}}$ satisfies \Cref{assump: LS smoothness} and \Cref{assump: Gamma smoothness} below.
\end{example}

\begin{example}[Newton--Kantorovich Method]
Newton's method achieves local quadratic convergence for nonlinear problems such as the Bellman and Euler equations (Examples~\ref{ex: bellman equation} and~\ref{ex: euler}). Applying $q$ Newton steps from $Y^{(0)}$ defines $\Gamma^q_{\mathrm{Newton}}(\theta, Y^{(0)}, P, \tilde{\theta})= Y^{(q)}$, where
\begin{equation*}
    Y^{(s+1)} = Y^{(s)} - \left[I - \nabla_Y G(\theta, Y^{(s)}, P, \tilde{\theta})\right]^{-1}(Y^{(s)} - G(\theta, Y^{(s)}, P, \tilde{\theta})),
\end{equation*}
for $s = 0, \ldots, q-1$. Due to local quadratic convergence, a much smaller $q$ suffices compared to SA. However, each Newton step requires inverting $I - \nabla_Y G$, which is expensive for large state spaces; GMRES can be used to solve the Newton linear system instead.
\end{example}

\subsection{The EM-NPL($q$) Estimator} \label{sec: em npl estimator}

\noindent Following \citet{aguirregabiria2007sequential} and \citet{arcidiacono2011conditional}, we define the EM-NPL($q$) estimator as a fixed point of a system of equations that maximizes the pseudo-likelihood.

For each firm $j$, the CCP mapping $\Lambda_j: \Theta \times \mY \times \mP \to \Delta(\mA)^{|\mX|}$ maps structural parameters, the nuisance parameter, and the CCP profile to firm $j$'s CCP. We write $\Lambda = (\Lambda_1, \ldots, \Lambda_J)$ for the stacked mapping across all firms.  Let $\Gamma^{q}(\theta, Y, P, \tilde{\theta})$ denote the inner algorithm \eqref{eq:gamma_q_defn}, which approximates the solution of \eqref{eq: fixed point general}, with the initial guess set to the estimate from the previous iteration, $Y^{(0)} = Y$. Let $\boldsymbol{\theta} = (\theta^1, \ldots, \theta^M)$, $\bfY = (Y^1, \ldots, Y^M)$, $\bfP = (P^1, \ldots, P^M)$, and $\boldsymbol{\pi} = (\pi_1, \ldots, \pi_M)$. We write $\boldsymbol{\Gamma}^q(\boldsymbol{\theta}, \bfY, \bfP, \tilde{\boldsymbol{\theta}})$ and $\boldsymbol{\Lambda}(\boldsymbol{\theta}, \bfY, \bfP)$ for the type-by-type stacked mappings, i.e., $\boldsymbol{\Gamma}^q(\boldsymbol{\theta}, \bfY, \bfP, \tilde{\boldsymbol{\theta}}) \coloneqq\bigl(\Gamma^q(\theta^1, Y^1, P^1, \tilde{\theta}^1), \ldots, \Gamma^q(\theta^M, Y^M, P^M, \tilde{\theta}^M)\bigr)$ and analogously for $\boldsymbol{\Lambda}$. Because the panel's first observed state $x_{i1}$ is generally correlated with market $i$'s latent type (the initial-conditions problem of \citealp{heckman1981heterogeneity}), the pseudo-likelihood below accounts for the type-specific distribution of $x_{i1}$. For a CCP profile $P^m$, define the type-$m$ transition kernel over $x$ induced by $P^m$ as
\begin{equation*}
    f_x^{P^m}(x'\mid x) \coloneqq\sum_{a \in \mA^J} f_x(x'\mid x, a;\theta_f) \prod_{k=1}^J P_k^m(a_k \mid x), \qquad x, x' \in \mX.
\end{equation*}
We impose the following assumption to resolve this problem, following \citet[Section 3.5]{aguirregabiria2007sequential}. \Cref{assump: stationary initial existence} holds automatically when $f_x^{P^m}(x'\mid x)>0$ for all $x,x'\in\mX$. \Cref{assump: stationary initial draw} reflects a panel observed after equilibrium play has settled into its steady state.
\begin{assumption}\label{assump: stationary initial}
\begin{assumptionitems}
    \item \label{assump: stationary initial existence} For every type $m$ and every CCP profile $P^m$ in a neighborhood of $P_0^m$, the Markov kernel $f_x^{P^m}$ has a unique stationary distribution $p_*^m(\cdot\,;P^m)$ on $\mX$.
    \item \label{assump: stationary initial draw} Conditional on market $i$'s type $m_i = m$, the initial state $x_{i1}$ is drawn from $p_*^m(\cdot\,;P_0^m)$, independently across markets.
\end{assumptionitems}
\end{assumption}

Define the pseudo-likelihood of market $i$ for type $m$ with $(Y^m, P^m, \tilde{\theta}^m)$ from the previous iteration and with $\Gamma^q$ in place of the exact solution to the fixed-point equation as
\begin{equation}\label{eq: Ell_i}
\mL_i^{q}(\theta^m;Y^m, P^m, \tilde{\theta}^m) \coloneqq p_*^m(x_{i1};P^m) \prod_{t=1}^{T} \prod_{j=1}^{J} \Lambda_j\bigl(\theta^m,\, \Gamma^q(\theta^m, Y^m, P^m, \tilde{\theta}^m),\, P^m\bigr)(a_{jit}|x_{it}).
\end{equation}
Define the pseudo-likelihood function with $(\bfY, \bfP, \tilde{\boldsymbol{\theta}})$ from the previous iteration as
\begin{equation*}
    \cQ_N^q(\boldsymbol{\theta}, \boldsymbol{\pi};\, \bfY, \bfP,\tilde{\boldsymbol{\theta}}) \coloneqq\frac{1}{N} \sum_{i=1}^N \log \left( \sum_{m=1}^M \pi_m \, \mL_i^{q}(\theta^m;Y^m, P^m, \tilde{\theta}^m) \right).
\end{equation*}

\begin{definition}[EM-NPL($q$) Estimator] \label{def: em npl estimator}
    The EM-NPL($q$) fixed points are defined as
    \begin{equation} \label{eq: em npl fp set}
        \mathcal{E}_N^q \coloneqq\left\{(\boldsymbol{\theta}, \boldsymbol{\pi}, \bfY, \bfP) \;\middle|\;
        \begin{aligned}
            & (\boldsymbol{\theta}, \boldsymbol{\pi}) = \argmax_{\boldsymbol{\vartheta},\boldsymbol{\varpi}}\; \cQ_N^q(\boldsymbol{\vartheta}, \boldsymbol{\varpi};\, \bfY, \bfP, \boldsymbol{\theta}), \\
            & \bfY = \boldsymbol{\Gamma}^q(\boldsymbol{\theta}, \bfY, \bfP, \boldsymbol{\theta}),\quad \bfP = \boldsymbol{\Lambda}(\boldsymbol{\theta}, \bfY, \bfP)
        \end{aligned}
        \right\}.
    \end{equation}
    The EM-NPL($q$) estimator is the element of $\mathcal{E}_N^q$ that maximizes the pseudo-likelihood.
\end{definition}
The fixed-point conditions in \eqref{eq: em npl fp set} require: (i)~(\textit{argmax condition}) $(\boldsymbol{\theta}, \boldsymbol{\pi})$ maximize the pseudo-likelihood given $(\boldsymbol{\theta}, \bfY, \bfP)$ where $\boldsymbol{\theta}$ on the right-hand side is the previous iterate; (ii)~(\textit{$Y$-condition}) $\bfY$ is a fixed point of the inner algorithm $\Gamma^q$; and (iii)~(\textit{$P$-condition}) $\bfP$ is a fixed point of the mapping $\Lambda$. 

We next define the estimator based on the exact solution of \eqref{eq: fixed point general}, for which we impose the following assumption.
\begin{assumption}[Unique Fixed Point] \label{assumption: unique fixed point}
For any $(\theta, P, \tilde{\theta}, \tilde{Y})$, the fixed-point equation \eqref{eq: fixed point general} has a unique fixed point in $Y$ denoted by $Y_{\mathrm{fp}}(\theta, P, \tilde{\theta}, \tilde{Y})$.
\end{assumption}
\Cref{assumption: unique fixed point} ensures that the nuisance parameter is pinned down once the structural parameters, CCPs, and the previous-iteration values $(\tilde{\theta}, \tilde{Y})$ are known. The assumption is mild: for Examples \ref{ex: policy valuation}, \ref{ex: bellman equation}, and \ref{ex: euler}, the operator $G$ is a contraction, so the fixed point is unique, and for the EPL mapping, the fixed point is unique whenever $I - \nabla_v \Phi(\tilde{\theta}, \tilde{Y})$ is nonsingular. We maintain \Cref{assumption: unique fixed point} throughout the remainder of the paper and do not repeat it in formal statements.

\begin{definition}[EM-NPL Estimator] \label{def: em npl estimator inf}
    The EM-NPL estimator is defined analogously to \Cref{def: em npl estimator}, with the exact solution $Y_{\mathrm{fp}}$ of \eqref{eq: fixed point general} in place of $\Gamma^q$. The EM-NPL fixed-point set $\mathcal{E}_N$ is defined as $\mathcal{E}_N^q$ in \eqref{eq: em npl fp set} with $Y_{\mathrm{fp}}(\theta^m, P^m, \tilde{\theta}^m, Y^m)$ substituted for $\Gamma^q(\theta^m, Y^m, P^m, \tilde{\theta}^m)$ throughout. The EM-NPL estimator is the element of $\mathcal{E}_N$ that maximizes the pseudo-likelihood.
\end{definition}

To compute the EM-NPL($q$) estimator, we define the EM-NPL($q$) algorithm following \citet{aguirregabiria2007sequential} and \citet{arcidiacono2011conditional}. We first impose two primitive regularity conditions used repeatedly below.
\begin{assumption} \label{assump: LS primitives}
\begin{assumptionitems}
    \item \label{assump: LS positive ccp} $\Lambda_j(\theta, \Gamma^q(\theta, Y, P, \tilde{\theta}), P)(a|x) > 0$ for all $(\theta, Y, P, \tilde{\theta}) \in \Theta \times \mY \times \mP\times \Theta $ and all $(x, a, j)$.
    \item \label{assump: LS smoothness} The mappings $G$, $\Lambda$, and $\Gamma^q$ are twice continuously differentiable in all arguments.
\end{assumptionitems}
\end{assumption}

In the EM-NPL($q$) algorithm defined below, we use $\theta^{m,(k-1)}$ as the input $\tilde{\theta}$ in $\mL_i^{q}$ in the $k$-th iteration. This is because, in $G(\theta,Y,P,\tilde{\theta},\tilde{Y})$, $\tilde{\theta}$ and $\tilde{Y}$ must be evaluated at the same iterate, and $\Gamma^{q}(\theta,Y,P,\tilde{\theta})$ initializes its solver at $Y^{(0)}=Y$, so $Y$ itself serves as the center $\tilde{Y}$. 

\begin{definition}[EM-NPL($q$) Algorithm] \label{alg: em npl}
    Given inner algorithm $\Gamma^q$, convergence tolerance $\varepsilon > 0$, and initial values $( \boldsymbol{\pi}^{(0)}, \boldsymbol{\theta}^{(0)}, \bfY^{(0)}, \bfP^{(0)})$,\footnote{See \Cref{sec: initialization} on how to choose the initial values $( \boldsymbol{\pi}^{(0)}, \boldsymbol{\theta}^{(0)},\bfY^{(0)}, \bfP^{(0)})$.} the EM-NPL($q$) algorithm iterates the following steps for $k = 1, 2, \ldots$. 

    \medskip
    \textbf{E-Step.}  Compute the posterior type probability for each market $i$ and type $m$ as
    \begin{equation*}
        w_{im}^{(k)} = \frac{\pi_m^{(k-1)}\,\mL_i^{q}\bigl(\theta^{m,(k-1)};Y^{m,(k-1)},P^{m,(k-1)},\theta^{m,(k-1)}\bigr)}{\sum_{m'=1}^M \pi_{m'}^{(k-1)}\,\mL_i^{q}\bigl(\theta^{m',(k-1)};Y^{m',(k-1)},P^{m',(k-1)},\theta^{m',(k-1)}\bigr)}.
    \end{equation*}

    \textbf{M-Step.}  For each type $m = 1,\ldots,M$, update the structural parameters as
    \begin{equation} \label{eq: M-step_obj}
        \theta^{m,(k)} = \argmax_{\theta^m \in \Theta} \sum_{i=1}^N w_{im}^{(k)} 
        \log \bigl( \mL_i^{q}(\theta^m;Y^{m,(k-1)}, P^{m,(k-1)}, \theta^{m,(k-1)})\bigr),
    \end{equation}
    update the mixing weights by
    \begin{equation*}
        \pi_m^{(k)} = \frac{1}{N} \sum_{i=1}^N w_{im}^{(k)},
    \end{equation*}
    and update the nuisance parameter and CCPs by 
    \begin{align*}
        Y^{m,(k)} = \Gamma^q\bigl(\theta^{m,(k)}, Y^{m,(k-1)}, P^{m,(k-1)}, \theta^{m,(k-1)}\bigr), \quad P^{m,(k)} = \Lambda\bigl(\theta^{m,(k)}, Y^{m,(k)}, P^{m,(k-1)}\bigr).
    \end{align*}

    \textbf{Stop} if $\max \bigl\{ \|\boldsymbol{\theta}^{(k)} - \boldsymbol{\theta}^{(k-1)}\|_{\infty},\, \|\bfP^{(k)} - \bfP^{(k-1)}\|_{\infty},\, \|\boldsymbol{\pi}^{(k)} - \boldsymbol{\pi}^{(k-1)}\|_{\infty} \bigr\} \leq \varepsilon$.
\end{definition}
We call $(\boldsymbol{\theta}^{(k)}, \boldsymbol{\pi}^{(k)}, \bfY^{(k)}, \bfP^{(k)})$, $k = 1, 2, \ldots$, the \textit{EM-NPL($q$) sequence}. The following lemma shows that if the EM-NPL($q$) sequence converges to a limit, then the limit is a stationary point of $\cQ_N^q$ that satisfies the $Y$- and $P$-conditions in \eqref{eq: em npl fp set}. This is a generalization of the standard convergence result of \citet{dempster1977maximum} and \citet{wu1983convergence}. 

\begin{lemma}\label{lemma: algorithm limit is fixed point}
    Suppose that \Cref{assump: LS primitives} holds and the EM-NPL($q$) sequence converges to $(\hat{\boldsymbol{\theta}}, \hat{\boldsymbol{\pi}}, \hat{\bfY}, \hat{\bfP})$ as $k \to \infty$, with $(\hat{\boldsymbol{\theta}}, \hat{\boldsymbol{\pi}})$ in the interior of $\Theta^M \times \Delta^{M-1}$. Then  $(\hat{\boldsymbol{\theta}}, \hat{\boldsymbol{\pi}}, \hat{\bfY}, \hat{\bfP})$ satisfies $\hat{\bfY} = \boldsymbol{\Gamma}^q(\hat{\boldsymbol{\theta}}, \hat{\bfY}, \hat{\bfP}, \hat{\boldsymbol{\theta}})$, $\hat{\bfP} = \boldsymbol{\Lambda}(\hat{\boldsymbol{\theta}}, \hat{\bfY}, \hat{\bfP})$, and $\nabla_{(\vartheta,\varpi)} \cQ_N^q(\hat{\boldsymbol{\theta}},\hat{\boldsymbol{\pi}};\hat{\bfY},\hat{\bfP},\hat{\boldsymbol{\theta}})= 0$, where the gradient with respect to $\boldsymbol{\varpi}$ is taken with respect to the free coordinates $(\varpi_1,\ldots,\varpi_{M-1})$ after substituting $\varpi_M = 1 - \sum_{m=1}^{M-1}\varpi_m$.
\end{lemma}

\Cref{sec: convergence} establishes local convergence of the algorithm to the EM-NPL($q$) estimator given a consistent initial estimator.
\begin{remark}[Reduction to single-type case] \label{remark: single type reduction}
    When $M = 1$, the E-step is vacuous ($w_{i1}^{(k)} = 1$ for all $i$) and the EM-NPL($q$) estimator reduces to the NPL($q$) estimator of \citet{aguirregabiria2007sequential} with inner algorithm $\Gamma^q$.
\end{remark}

\begin{remark}[Alternative CCP updates] \label{remark: spectral}
The CCP update step can be replaced by an alternative mapping that improves the convergence properties of the outer iterations.  For example, the relaxation method \citep{kasahara2012sequential} or spectral methods \citep{aguirregabiria2021imposing} can be applied. 
\end{remark}

\section{Theoretical Properties} \label{sec: theory}

\noindent This section develops the theoretical properties of the EM-NPL($q$) estimator. We first establish the truncation invariance property, present consistency and asymptotic normality, and finally characterize local convergence of the EM-NPL($q$) algorithm.

\subsection{Truncation Invariance} \label{sec: equivalence}

\noindent Recall that $q$ is the number of inner iterations per outer step. We show that for the workhorse class of DDC models with linear-in-parameters utility, the EM-NPL($q$) estimator for any $q \geq 1$ is numerically identical to the EM-NPL estimator. Therefore, truncating the inner solver at any finite $q$ does not distort the estimator. The key observation is that for linear-in-parameters utility, the policy valuation (\Cref{ex: policy valuation}) and EPL (\Cref{ex: epl}) fixed-point equations do not depend on the current $\theta$. We formalize this as $\theta$-separability.
\begin{assumption}[$\theta$-separability] \label{assumption: theta separability}
The fixed-point equation \eqref{eq: fixed point general} takes the form $Y = G(Y, P, \tilde \theta, \tilde{Y})$ with $G$ independent of $\theta$.
\end{assumption}
Under $\theta$-separability, $Y$ can be solved once for given $(P, \tilde{\theta}, \tilde{Y})$ from the previous iteration, rather than re-solved for each candidate $\theta$ during the M-step. For \Cref{ex: policy valuation,ex: epl} under linear-in-parameters utility, $G$ satisfies Assumption \ref{assumption: theta separability} after the reformulation described in those examples. Since $\Gamma^q$ is applied to that reformulated system, it likewise does not depend on $\theta$. We further impose the following condition on the inner algorithm. 
\begin{assumption} \label{assumption: same fixed point}
For any $q$, $Y = \Gamma^q(\theta, Y, P, \tilde{\theta})$ holds if and only if $Y = Y_{\mathrm{fp}}(\theta, P, \tilde{\theta}, Y)$.
\end{assumption}
\Cref{assumption: same fixed point} is mild and can be verified for common algorithms. For example, when $G$ is a contraction, SA satisfies the condition. For linear systems, GMRES also typically satisfies the condition (see \citealp{saad2003iterative}). The following theorem shows that the EM-NPL($q$) estimator for any $q \geq 1$ is numerically identical to the EM-NPL estimator.
\begin{theorem}[Truncation Invariance] \label{theorem: truncation invariance}
    Suppose Assumptions \ref{assumption: theta separability} and \ref{assumption: same fixed point} hold. Then $\mathcal{E}_N^q = \mathcal{E}_N$ for all $q \geq 1$, and the EM-NPL($q$) estimator is numerically identical to the EM-NPL estimator.
\end{theorem}

\Cref{theorem: truncation invariance} has two important consequences for inference.  First, there is no efficiency loss from inner truncation. NPL standard errors are valid without any adjustment for the inner truncation. Second, the choices of $q$ and the inner algorithm $\Gamma^q$ become purely computational decisions: they affect the speed of convergence of the EM-NPL($q$) algorithm, but not the statistical properties of the estimator.

\subsection{Large-Sample Properties} \label{sec: large sample}

\noindent We establish consistency and asymptotic normality of the EM-NPL($q$) estimator by adapting \citet[Proposition 2]{aguirregabiria2007sequential} and \citet[Theorem 2]{arcidiacono2011conditional}. Define the population pseudo-likelihood as
\begin{equation*}
    \cQ^q(\boldsymbol{\theta}, \boldsymbol{\pi};\, \bfY, \bfP, \tilde{\boldsymbol{\theta}}) \coloneqq\bE\biggl[ \log \biggl( \sum_{m=1}^M \pi_m\, p_*^m(x_1;P^m) \prod_{t=1}^{T} \prod_{j=1}^{J} \Lambda_j\bigl(\theta^m,\, \Gamma^q(\theta^m, Y^m, P^m, \tilde{\theta}^m),\, P^m\bigr)(a_{jt}|x_{t}) \biggr)\biggr].
\end{equation*}
Define the population EM-NPL($q$) fixed-point set as
\begin{equation*}
    \mathcal{E}^q \coloneqq\left\{(\boldsymbol{\theta}, \boldsymbol{\pi}, \bfY, \bfP) \;\middle|\;
    \begin{aligned}
        & (\boldsymbol{\theta}, \boldsymbol{\pi}) = \argmax_{\boldsymbol{\vartheta},\boldsymbol{\varpi}}\; \cQ^q(\boldsymbol{\vartheta}, \boldsymbol{\varpi};\, \bfY, \bfP, \boldsymbol{\theta}), \\
        & \bfY = \boldsymbol{\Gamma}^q(\boldsymbol{\theta}, \bfY, \bfP, \boldsymbol{\theta}),\quad \bfP = \boldsymbol{\Lambda}(\boldsymbol{\theta}, \bfY, \bfP)
    \end{aligned}
    \right\}.
\end{equation*}

Let $(\boldsymbol{\theta}_0, \boldsymbol{\pi}_0, \bfY_0, \bfP_0) \in \mathcal{E}^q$ denote the true parameter value. We impose the following identification conditions.
\begin{assumption}\label{assump: nplq identification}
    \begin{assumptionitems}
        \item \label{assump: nplq identification1} Among all $(\boldsymbol{\theta}, \boldsymbol{\pi}, \bfY, \bfP) \in \mathcal{E}^q$, $(\boldsymbol{\theta}_0, \boldsymbol{\pi}_0, \bfY_0, \bfP_0)$ uniquely maximizes $\cQ^q(\boldsymbol{\theta}, \boldsymbol{\pi};\,  \bfY, \bfP, \boldsymbol{\theta})$ up to label permutation.
        \item \label{assump: nplq identification2} $(\boldsymbol{\theta}_0, \boldsymbol{\pi}_0, \bfY_0, \bfP_0)$ is an isolated population EM-NPL($q$) fixed point, i.e., it is unique or there exists an open ball around it that does not contain any other population EM-NPL($q$) fixed point.
    \end{assumptionitems}
\end{assumption}
\Cref{assump: nplq identification1} is the identification assumption, which can be established along the lines of \citet{kasahara2009nonparametric} and \citet{Higgins23joe}. We assume throughout that the number of types $M$ is known; in practice it can be selected by the Bayesian Information Criterion (BIC) or the sequential testing procedure of \citet{kasahara2014non}. \Cref{assump: nplq identification2} parallels the isolated-fixed-point condition in \citet[Proposition 2]{aguirregabiria2007sequential}.

\begin{remark}[Relationship to standard identification] \label{remark: id relationship}
    Under Assumptions \ref{assumption: theta separability} and \ref{assumption: same fixed point}, $G$ and hence $\Gamma^q$ do not depend on $\theta$, so the EM-NPL($q$) fixed-point system reduces to the standard NPL fixed-point system.  In that case, the population pseudo-likelihood simplifies to
    \begin{equation*}
      \bE\biggl[ \log \biggl( \sum_{m=1}^M \pi_m\, p_*^m(x_1;P^m) \prod_{t=1}^{T} \prod_{j=1}^{J} \Lambda_j\bigl(\theta^m,\, \Gamma^q(Y^m, P^m, \tilde{\theta}^m),\, P^m\bigr)(a_{jt}|x_{t}) \biggr)\biggr].
    \end{equation*}
    At the EM-NPL($q$) fixed point, $\Gamma^q(Y^m, P^m, \tilde{\theta}^m) = G(Y^m, P^m, \tilde{\theta}^m, Y^m)$, and the EM-NPL($q$) fixed points coincide with the NPL fixed points. Therefore, \Cref{assump: nplq identification} reduces to the standard identification condition for NPL estimators.
\end{remark}

Next, we impose regularity conditions to establish consistency and asymptotic normality.  These conditions are similar to those in \citet{aguirregabiria2007sequential}. Let $\boldsymbol{\alpha} \coloneqq (\boldsymbol{\theta}, \boldsymbol{\pi})$ with $\dim(\boldsymbol{\alpha} ) = d_\alpha = M \cdot d_\theta + M - 1$ collect all structural parameters and mixing probabilities, and let $\boldsymbol{\alpha}_0 \coloneqq (\boldsymbol{\theta}_0, \boldsymbol{\pi}_0)$. Let $\Gamma^q_Y$ and $\Gamma^q_P$ denote the Jacobians of $\Gamma^q(\theta, Y, P, \tilde{\theta})$ with respect to $Y$ and $P$ evaluated at $(\theta_0^m, Y_0^m, P_0^m, \theta_0^m)$, respectively, where we suppress the type superscript $m$, and define $\Lambda_Y$, $\Lambda_P$, and $G_Y$ similarly. Let $(Y_{\mathrm{fp}})_{\tilde{Y}}$ and $(Y_{\mathrm{fp}})_{P}$ denote the Jacobians of $Y_{\mathrm{fp}}(\theta, P, \tilde{\theta}, \tilde{Y})$ with respect to $\tilde{Y}$ and $P$ evaluated at $(\theta_0^m, P_0^m, \theta_0^m, Y_0^m)$. Because the solver sets the center equal to the initial guess, $\Gamma^q_Y$ is the total derivative with respect to $Y$, including the dependence through the center $\tilde{Y} = Y$.

\begin{assumption} \label{assump: LS regularity}
\begin{assumptionitems}
    \item \label{assump: LS iid} The observations $\{(x_{it}, (a_{jit})_{j=1}^{J})_{t=1}^{T}\}_{i=1}^N$ are i.i.d.\ across $i$ with common panel length $T$.  The process $\{(x_{it}, (a_{jit})_{j=1}^{J})\}_{t=1}^{T}$ is stationary Markov and $x_{it}$ has full support conditional on type $m_i = m$.  Types are drawn i.i.d.\ from $(\pi_{1,0},\ldots,\pi_{M,0})$.
    \item \label{assump: LS compactness} $\Theta$ is compact, $\theta_0^m \in \mathrm{int}(\Theta)$ for all $m$, and $\boldsymbol{\pi}_0 \in \mathrm{int}(\Delta^{M-1})$.
    \item \label{assump: LS info matrix} The Hessian $\nabla_{\boldsymbol{\alpha} \boldsymbol{\alpha}}^2 \cQ^q(\boldsymbol{\alpha}_0;\, \bfY_0, \bfP_0, \boldsymbol{\theta}_0)$ is negative definite, and the adjusted score Jacobian $\Omega^q_{\boldsymbol{\alpha}\boldsymbol{\alpha}}+S^q$ defined below is nonsingular. The same two conditions hold with $Y_{\mathrm{fp}}(\boldsymbol{\theta}, \bfP, \tilde{\boldsymbol{\theta}}, \bfY)$ in place of $\Gamma^q(\boldsymbol{\theta}, \bfY, \bfP, \tilde{\boldsymbol{\theta}})$ throughout, i.e., for the Hessian of $\cQ^{\mathrm{NPL}}$ and for $\Omega^{\mathrm{NPL}}_{\boldsymbol{\alpha}\boldsymbol{\alpha}}+S^{\mathrm{NPL}}$.
    \item \label{assump: LS local equiv} There exists a neighborhood of $(\bfY_0, \bfP_0,\boldsymbol{\theta}_0)$ such that the global maximizer $\boldsymbol{\alpha}(\bfY, \bfP,\tilde{\boldsymbol{\theta}})$ of $\cQ^q(\boldsymbol{\alpha};\, \bfY, \bfP, \tilde{\boldsymbol{\theta}})$ is unique up to label permutation for all $( \bfY, \bfP,\tilde{\boldsymbol{\theta}})$ in that neighborhood.
    \item \label{assump: LS isolated fp} Let $\phi_q(\bfY, \bfP,\tilde{\boldsymbol{\theta}}) \coloneqq \bigl(\boldsymbol{\Gamma}^q(\boldsymbol{\theta}_\alpha, \bfY, \bfP, \tilde{\boldsymbol{\theta}}),\; \boldsymbol{\Lambda}(\boldsymbol{\theta}_\alpha, \boldsymbol{\Gamma}^q(\boldsymbol{\theta}_\alpha, \bfY, \bfP, \tilde{\boldsymbol{\theta}}), \bfP) \bigr)$, where $\boldsymbol{\theta}_\alpha$ is the $\boldsymbol{\theta}$-component of $\boldsymbol{\alpha}(\bfY, \bfP,\tilde{\boldsymbol{\theta}})$, and define $\phi_{\mathrm{NPL}}(\bfY, \bfP,\tilde{\boldsymbol{\theta}})$ analogously with $Y_{\mathrm{fp}}(\boldsymbol{\theta}_\alpha, \bfP, \tilde{\boldsymbol{\theta}}, \bfY)$ in place of $\boldsymbol{\Gamma}^q(\boldsymbol{\theta}_\alpha, \bfY, \bfP, \tilde{\boldsymbol{\theta}})$. Both $\bigl(\phi_q(\bfY, \bfP, \tilde{\boldsymbol{\theta}}), \boldsymbol{\theta}_\alpha\bigr) - (\bfY, \bfP,\tilde{\boldsymbol{\theta}})$ and $\bigl(\phi_{\mathrm{NPL}}(\bfY, \bfP, \tilde{\boldsymbol{\theta}}),\boldsymbol{\theta}_\alpha \bigr) - (\bfY, \bfP,\tilde{\boldsymbol{\theta}})$ have a nonsingular Jacobian matrix at $(\bfY_0, \bfP_0, \boldsymbol{\theta}_0)$.
    \item \label{assump: implicit function theorem} For each type $m = 1,\ldots,M$, the matrices $\begin{pmatrix} I - \Gamma^q_{Y} & -\Gamma^q_{P} \\ -\Lambda_{Y} & I - \Lambda_{P} \end{pmatrix}$ and $\begin{pmatrix} I - (Y_{\mathrm{fp}})_{\tilde{Y}} & -(Y_{\mathrm{fp}})_{P} \\ -\Lambda_{Y} & I - \Lambda_{P} \end{pmatrix}$ are invertible.
    \item \label{assump: fp jacobian} For each type $m = 1,\ldots,M$, the matrix $I - G_Y$ is nonsingular, where $G_Y$ denotes the Jacobian of $G$ with respect to its second argument $Y$. Consequently $Y_{\mathrm{fp}}(\theta,P,\tilde{\theta},\tilde{Y})$ is continuously differentiable in a neighborhood of $(\theta_0^m, P_0^m, \theta_0^m, Y_0^m)$.
\end{assumptionitems}
\end{assumption}
\Cref{assump: LS regularity} is similar to the assumptions in Proposition 2 of \citet{aguirregabiria2007sequential}. Write the per-observation log-likelihood contribution, with $\mL_i^{q}$ defined in \eqref{eq: Ell_i}, as
\begin{equation*}
    \ell_i^q(\boldsymbol{\alpha};\, \bfY, \bfP, \tilde{\boldsymbol{\theta}}) \coloneqq\log \left(\sum_{m=1}^M \pi_m\, \mL_i^{q}(\theta^m;Y^m, P^m, \tilde{\theta}^m) \right),
\end{equation*}
and the score $s_i^q(\boldsymbol{\alpha};\, \bfY, \bfP, \tilde{\boldsymbol{\theta}}) \coloneqq\nabla_{\boldsymbol{\alpha}} \ell_i^q(\boldsymbol{\alpha};\, \bfY, \bfP, \tilde{\boldsymbol{\theta}})$. Define the information matrices $\Omega_{\boldsymbol{\alpha}\boldsymbol{\alpha}}^{q} \coloneqq -\nabla_{\boldsymbol{\alpha}\boldsymbol{\alpha}}^2 \cQ^q$,  $\Omega^{q}_{\boldsymbol{\alpha} \bfY} \coloneqq-\nabla_{\boldsymbol{\alpha} \bfY}^2 \cQ^q$, $\Omega^{q}_{\boldsymbol{\alpha} \bfP} \coloneqq-\nabla_{\boldsymbol{\alpha} \bfP}^2 \cQ^q$, and $\Omega^{q}_{\boldsymbol{\alpha}\tilde{\boldsymbol{\theta}}} \coloneqq-\nabla_{\boldsymbol{\alpha}\tilde{\boldsymbol{\theta}}}^2 \cQ^q$, where  the derivatives are evaluated at $(\boldsymbol{\alpha}_0; \bfY_0, \bfP_0,\boldsymbol{\theta}_0)$.

Under Assumptions~\ref{assump: LS smoothness} and \ref{assump: implicit function theorem}, the implicit function theorem applies. Let $\bfY^q(\boldsymbol{\theta})$ and $\bfP^q(\boldsymbol{\theta})$ denote the solutions to $\bfY = \boldsymbol{\Gamma}^q(\boldsymbol{\theta}, \bfY, \bfP, \boldsymbol{\theta})$ and $\bfP = \boldsymbol{\Lambda}(\boldsymbol{\theta}, \bfY, \bfP)$ in a neighborhood of $\boldsymbol{\theta}_0$, and let $\bfY^q_{\boldsymbol{\theta}}$ and $\bfP^q_{\boldsymbol{\theta}}$ denote the derivatives of $\bfY^q(\boldsymbol{\theta})$ and $\bfP^q(\boldsymbol{\theta})$ with respect to $\boldsymbol{\theta}$ at $\boldsymbol{\theta}_0$. Define the correction term $ S^{q} \coloneqq (\Omega^{q}_{\boldsymbol{\alpha}\tilde{\boldsymbol{\theta}}}+ \Omega^{q}_{\boldsymbol{\alpha} \bfY}\, \bfY^q_{\boldsymbol{\theta}} + \Omega^{q}_{\boldsymbol{\alpha} \bfP}\, \bfP^q_{\boldsymbol{\theta}})(\mathbf{I}_{Md_\theta},\, 0)$ where $(\mathbf{I}_{Md_\theta},\, 0)$ is the $Md_\theta \times d_\alpha$ selector extracting the $\boldsymbol{\theta}$-block from $\boldsymbol{\alpha}$. The correction term $S^{q}$ contains $\Omega^{q}_{\boldsymbol{\alpha}\tilde{\boldsymbol{\theta}}}$ because the fixed-point mapping $G$ may depend on the previous-iteration parameter $\tilde{\boldsymbol{\theta}}$ as in the EPL estimator (\Cref{ex: epl}).  When $G$ depends only on the current $\boldsymbol{\theta}$, the correction simplifies to $S^q = (\Omega^q_{\boldsymbol{\alpha} \bfY}\, \bfY^q_{\boldsymbol{\theta}} + \Omega^q_{\boldsymbol{\alpha} \bfP}\, \bfP^q_{\boldsymbol{\theta}})(\mathbf{I}_{Md_\theta},\, 0)$. We suppress the superscript $q$ on these derivative matrices when no confusion arises.

\begin{theorem}[Consistency and Asymptotic Normality] \label{theorem: LS mixture asymptotic}
    Under Assumptions~\ref{assump: standard}, \ref{assump: stationary initial}, \ref{assump: LS primitives}, \ref{assumption: same fixed point}, \ref{assump: nplq identification}, and \ref{assump: LS regularity}, the EM-NPL($q$) estimator is consistent and asymptotically normal
    \begin{equation*}
        \sqrt{N} \begin{pmatrix} \hat{\boldsymbol{\theta}}_{\mathrm{NPL}(q)} - \boldsymbol{\theta}_0 \\ \hat{\boldsymbol{\pi}}_{\mathrm{NPL}(q)} - \boldsymbol{\pi}_0 \end{pmatrix} \xrightarrow{d} \mathcal{N}\left(0,\, \Sigma^q\right),
    \end{equation*}
    where $\Sigma^q = \bigl(\Omega^{q}_{\boldsymbol{\alpha}\boldsymbol{\alpha}} + S^{q}\bigr)^{-1}\, \Omega^{q}_{\boldsymbol{\alpha}\boldsymbol{\alpha}}\, \bigl\{\bigl(\Omega^{q}_{\boldsymbol{\alpha}\boldsymbol{\alpha}} + S^{q}\bigr)^{-1}\bigr\}^{\top}$.
\end{theorem}

\begin{corollary}[Asymptotic Normality of EM-NPL] \label{corollary: npl asymptotic normality}
    Under Assumptions~\ref{assump: standard}, \ref{assump: stationary initial}, \ref{assump: LS primitives}, \ref{assumption: same fixed point}, \ref{assump: nplq identification}, and \ref{assump: LS regularity}, the EM-NPL estimator defined in \Cref{def: em npl estimator inf} is consistent and asymptotically normal
    \begin{equation*}
        \sqrt{N} \begin{pmatrix} \hat{\boldsymbol{\theta}}_{\mathrm{NPL}} - \boldsymbol{\theta}_0 \\ \hat{\boldsymbol{\pi}}_{\mathrm{NPL}} - \boldsymbol{\pi}_0 \end{pmatrix} \xrightarrow{d} \mathcal{N}\left(0,\, \Sigma_{\mathrm{NPL}}\right),
    \end{equation*}
    where $\Sigma_{\mathrm{NPL}}$ is the expression for $\Sigma^q$ in \Cref{theorem: LS mixture asymptotic} with $Y_{\mathrm{fp}}(\theta, P, \tilde{\theta}, Y)$ in place of $\Gamma^q(\theta, Y, P, \tilde{\theta})$.
\end{corollary}
\begin{proof}
    Apply \Cref{theorem: LS mixture asymptotic} with $\Gamma^q(\theta, Y, P, \tilde{\theta}) = Y_{\mathrm{fp}}(\theta, P, \tilde{\theta}, Y)$.
\end{proof}

Under \Cref{assumption: theta separability}, \Cref{theorem: truncation invariance} implies that the EM-NPL($q$) estimator is numerically identical to the EM-NPL estimator for any $q \geq 1$, so $\Sigma^q = \Sigma_{\mathrm{NPL}}$ and standard errors can be computed using the NPL formulas. When \Cref{assumption: theta separability} does not hold, $\Sigma^q$ depends on the cross-derivative matrices $\Omega^q_{\boldsymbol{\alpha}\tilde{\boldsymbol{\theta}}}$, $\Omega^q_{\boldsymbol{\alpha}\bfY}$, and $\Omega^q_{\boldsymbol{\alpha}\bfP}$, which can be cumbersome to compute in practice. In that case, a nonparametric bootstrap may be applied.

To compare the asymptotic variance of EM-NPL($q$) with that of EM-NPL, we impose the following condition that bounds the approximation error of the inner algorithm in terms of a function $f(q)$ and the quality of the initial guess $\tilde{Y}$.
\begin{assumption}[Approximation Error Bound] \label{assumption: approximation error upper bound}
    For a given algorithm $\Gamma$, there exists a function $f(q)$ with $f(q) \to 0$ as $q \to \infty$ such that, uniformly over types,
    \[
        \| Y_{\mathrm{fp}}(\theta, P, \tilde{\theta}, \tilde{Y}) - \Gamma^{q}(\theta, \tilde{Y}, P, \tilde{\theta}) \|
        \leq
        \| Y_{\mathrm{fp}}(\theta, P, \tilde{\theta},\tilde{Y}) - \tilde{Y} \|\, f(q),
    \]
    for any $(\theta, P, \tilde{\theta},\tilde{Y})$ in a neighborhood of $(\theta_0^m, P_0^m,\theta_0^m,Y_0^m)$.
\end{assumption}

The rate $f(q)$ is algorithm-specific. For example, if SA is used for the inner algorithm, then $f(q) = \rho_G^q$, where $\rho_G$ is the Lipschitz constant of the operator $G$. For superlinear algorithms such as GMRES, $f(q)$ can decay much faster than for SA.

The following proposition shows that the asymptotic variance of the EM-NPL($q$) estimator converges to that of the EM-NPL estimator at the rate $f(q)$.

\begin{proposition}[Asymptotic Variance Approximation] \label{proposition: variance approximation}
    Under Assumptions~\ref{assump: standard}, \ref{assump: stationary initial}, \ref{assump: LS primitives}, \ref{assumption: same fixed point}, \ref{assump: nplq identification}, \ref{assump: LS regularity}, and \ref{assumption: approximation error upper bound}, we have
    \begin{equation*}
        \| \Sigma^{q} - \Sigma_{\mathrm{NPL}} \|_{F} = O(f(q)).
    \end{equation*}
\end{proposition}
Thus, for superlinear solvers such as GMRES, the asymptotic variance of the EM-NPL($q$) estimator is close to that of the fully converged EM-NPL estimator even for small $q$.

\subsection{Local Convergence} \label{sec: convergence}

\noindent We now characterize the local convergence properties of the EM-NPL($q$) algorithm. Define $\boldsymbol{\zeta} \coloneqq (\boldsymbol{\theta}, \bfY, \bfP, \boldsymbol{\pi})$ and $\boldsymbol{\zeta}_0 \coloneqq (\boldsymbol{\theta}_0, \bfY_0, \bfP_0, \boldsymbol{\pi}_0)$. Define the population posterior-weighted M-step criterion for $\boldsymbol{\vartheta}$ as
\begin{equation*}
    \ell_E^q(\boldsymbol{\vartheta};\,\boldsymbol{\zeta}) \coloneqq\sum_{m=1}^{M} \bE\left[ w_{im}^q(\boldsymbol{\zeta}) \left(\log p_*^m(x_1;P^m) + \sum_{t=1}^{T} \sum_{j=1}^{J} \log \Lambda_j\bigl(\vartheta^m,\, \Gamma^q(\vartheta^m, Y^m, P^m, \theta^m),\, P^m\bigr)(a_{jt}|x_{t})\right) \right],
\end{equation*}
where
\[
w_{im}^q(\boldsymbol{\zeta}) \coloneqq\frac{\pi_m\, \mL_i^{q}({\theta}^{m}; Y^m, P^m, \theta^m)}{\sum_{m'=1}^M \pi_{m'}\, \mL_i^{q}({\theta}^{m'}; Y^{m'}, P^{m'}, \theta^{m'})},
\]
 is the posterior type-membership weight, built from the pseudo-likelihood $\mL_i^{q}$ defined in \eqref{eq: Ell_i} evaluated at the previous-iteration parameter $\boldsymbol{\zeta}=(\boldsymbol{\theta},\bfY, \bfP, \boldsymbol{\pi})$. Define the population M-step map for $\boldsymbol{\pi}$ as $\Pi^q(\boldsymbol{\zeta}) = \bE[w_{i1}^q(\boldsymbol{\zeta}), \ldots, w_{iM}^q(\boldsymbol{\zeta})]^{\top}$. The state $\boldsymbol{\theta}$ enters $\ell_E^q$ both through the posterior weights $w_{im}^q$ and through the fourth argument of $\Gamma^q$ in the criterion.

\begin{assumption} \label{assumption: mixture convergence regularity}
    \begin{assumptionitems}
        \item The initial values satisfy $(\boldsymbol{\theta}^{(0)}, \bfY^{(0)}, \bfP^{(0)}, \boldsymbol{\pi}^{(0)}) - \boldsymbol{\zeta}_0 = o_p(1)$.
        \item\label{assump: Gamma smoothness} $\Gamma^q$ and $\Lambda$ are three times continuously differentiable in a neighborhood of the true parameter value.
        \item The M-step Hessians $\nabla_{\boldsymbol{\vartheta}\boldsymbol{\vartheta}}^2 \ell_E^{q}(\boldsymbol{\theta}_0;\, \boldsymbol{\zeta}_0)$ and $\nabla_{\boldsymbol{\vartheta}\boldsymbol{\vartheta}}^2 \ell_E^{\mathrm{NPL}}(\boldsymbol{\theta}_0;\, \boldsymbol{\zeta}_0)$ are nonsingular, where $\ell_E^{\mathrm{NPL}}$ denotes $\ell_E^q$ with $Y_{\mathrm{fp}}$ in place of $\Gamma^q$.
    \end{assumptionitems}
\end{assumption}

Let $\Lambda^{q}(\theta, Y, P, \tilde{\theta}) \coloneqq\Lambda(\theta, \Gamma^q(\theta, Y, P, \tilde{\theta}), P)$ denote the composition of $\Gamma^q$ and $\Lambda$, and write $\boldsymbol{\Lambda}^q(\boldsymbol{\theta}, \bfY, \bfP, \tilde{\boldsymbol{\theta}})$ for $\Lambda^q$ stacked across types, with $\boldsymbol{\Gamma}^q$ as in \Cref{sec: em npl estimator}. Let $(\hat{\boldsymbol{\pi}}, \hat{\boldsymbol{\theta}}, \hat{\bfY}, \hat{\bfP})$ denote the EM-NPL($q$) estimator. Define the M-step sensitivities of $\boldsymbol{\vartheta}$ for $\boldsymbol{\eta} \in \{ \boldsymbol{\theta}, \bfY, \bfP, \boldsymbol{\pi}\}$ and the M-step sensitivities of $\boldsymbol{\pi}$ as
\[
    H_{\boldsymbol{\eta}}^q
    \coloneqq
    - \{\nabla_{\boldsymbol{\vartheta}\boldsymbol{\vartheta}}^2
    \ell_E^q(\boldsymbol{\theta}_0;\boldsymbol{\zeta}_0)\}^{-1}
    \nabla_{\boldsymbol{\vartheta} \boldsymbol{\eta} }^2
    \ell_E^q(\boldsymbol{\theta}_0;\boldsymbol{\zeta}_0),
    \quad \Pi^q_{\boldsymbol{\eta}} \coloneqq \nabla_{\boldsymbol{\eta}}\Pi^q(\boldsymbol{\zeta}_0).
\]
For $\boldsymbol{\eta} \in \{\boldsymbol{\theta},  \bfY, \bfP, \tilde{\boldsymbol{\theta}}\}$, let $\boldsymbol{\Gamma}_{\boldsymbol{\eta}}^q$ and $\boldsymbol{\Lambda}_{\boldsymbol{\eta}}^q$ denote the Jacobians of $\boldsymbol{\Gamma}^q$ and $\boldsymbol{\Lambda}^q$ evaluated at the true parameter value.  $\boldsymbol{\Gamma}_{\boldsymbol{\eta}}^q$ and $\boldsymbol{\Lambda}_{\boldsymbol{\eta}}^q$ are block-diagonal across types because each type's component depends only on its own type-specific arguments, whereas the M-step sensitivities $H_{\boldsymbol{\eta}}^q$ are full stacked matrices because the posterior weights couple the latent types.

 For $X\in\{\boldsymbol{\pi},\boldsymbol{\theta},\bfY,\bfP\}$ let $\Delta X^{(k)} \coloneqq X^{(k)} - \hat{X}$ denote the deviation of the $k$-th iterate of the EM-NPL($q$) sequence; write $\delta^{(k)} \coloneqq(\Delta\boldsymbol{\pi}^{(k)}, \Delta\boldsymbol{\theta}^{(k)}, \Delta\bfY^{(k)}, \Delta\bfP^{(k)})$.
 
\begin{proposition}[Local Convergence of EM-NPL($q$)] \label{proposition: mixture convergence matrix}
    Under Assumptions~\ref{assump: LS primitives}, \ref{assumption: same fixed point}, \ref{assump: nplq identification}, \ref{assump: LS regularity}, and \ref{assumption: mixture convergence regularity}, the EM-NPL($q$) algorithm satisfies
    \begin{equation*}
        \delta^{(k)} = R^q\, \delta^{(k-1)} + O_p(N^{-1/2}\|\delta^{(k-1)}\| + \|\delta^{(k-1)}\|^2),
    \end{equation*}
    where, in block order $(\boldsymbol{\pi},\boldsymbol{\theta},\bfY,\bfP)$,
    \begin{equation*}
        R^q = \begin{pmatrix}
            \Pi^q_{\boldsymbol{\pi}} & \Pi^q_{\boldsymbol{\theta}} & \Pi^q_{\bfY} & \Pi^q_{\bfP} \\[4pt]
            H_{\boldsymbol{\pi}}^q & H_{\boldsymbol{\theta}}^q & H_{\bfY}^q & H_{\bfP}^q \\[4pt]
            \boldsymbol{\Gamma}_{\boldsymbol{\theta}}^q H_{\boldsymbol{\pi}}^q & \boldsymbol{\Gamma}_{\boldsymbol{\theta}}^q H_{\boldsymbol{\theta}}^q + \boldsymbol{\Gamma}_{\tilde{\boldsymbol{\theta}}}^q  & \boldsymbol{\Gamma}_{\boldsymbol{\theta}}^q H_{\bfY}^q + \boldsymbol{\Gamma}_{\boldsymbol{Y}}^q & \boldsymbol{\Gamma}_{\boldsymbol{\theta}}^q H_{\bfP}^q + \boldsymbol{\Gamma}_{\boldsymbol{P}}^q \\[4pt]
            \boldsymbol{\Lambda}^{q}_{\boldsymbol{\theta}} H_{\boldsymbol{\pi}}^q & \boldsymbol{\Lambda}^{q}_{\boldsymbol{\theta}} H_{\boldsymbol{\theta}}^q  + \boldsymbol{\Lambda}^{q}_{\boldsymbol{\tilde{\theta}}} & \boldsymbol{\Lambda}^{q}_{\boldsymbol{\theta}} H_{\bfY}^q + \boldsymbol{\Lambda}^{q}_{\boldsymbol{Y}} & \boldsymbol{\Lambda}^{q}_{\boldsymbol{\theta}} H_{\bfP}^q + \boldsymbol{\Lambda}^{q}_{\boldsymbol{P}}
        \end{pmatrix}.
    \end{equation*}
\end{proposition}

The rate matrix $R^q$ has a sequential block form that reflects the sequential nature of the EM-NPL($q$) algorithm. The first row corresponds to the M-step's update of the mixing weights. The second row corresponds to the M-step's update of the structural parameter. The sensitivity of $\boldsymbol{\theta}^{(k)}$ to the previous iterates is governed by the implicit function $H$ blocks. The third and fourth rows compose the M-step output with the inner algorithm $\Gamma^q$ and the CCP mapping $\Lambda^q$.

The inner truncation level $q$ enters the rate matrix through $H^q_{\boldsymbol{\eta}}$ and the Jacobians of $(\Pi^q, \boldsymbol{\Gamma}^q, \boldsymbol{\Lambda}^q)$. The difference between these matrices and the corresponding counterpart for the EM-NPL algorithm is governed by the approximation error of the inner algorithm. Consequently, inner truncation affects the convergence speed of the outer loop. Under \Cref{assumption: theta separability}, $\boldsymbol{\Gamma}_{\boldsymbol{\theta}}^q = 0$ because the fixed-point mapping does not depend on the current structural parameter. When, in addition, $G$ does not depend on $\tilde{\boldsymbol{\theta}}$ (as in the policy valuation mapping of \Cref{ex: policy valuation}), $\boldsymbol{\Gamma}_{\tilde{\boldsymbol{\theta}}}^q = 0$ as well.

Let $R^{\mathrm{NPL}}$ denote the counterpart of $R^q$ in \Cref{proposition: mixture convergence matrix} with $Y_{\mathrm{fp}}(\theta, P, \tilde{\theta}, Y)$ in place of $\Gamma^q(\theta, Y, P, \tilde{\theta})$, i.e., the local convergence rate matrix of the EM-NPL algorithm. The following theorem shows that the EM-NPL($q$) algorithm converges locally whenever $\rho(R^{\mathrm{NPL}}) < 1$.
\begin{theorem}[Local Convergence for Sufficiently Large $q$] \label{theorem: npl(q) local convergence}
Suppose Assumptions \ref{assump: stationary initial} and \ref{assumption: approximation error upper bound} and the assumptions of \Cref{proposition: mixture convergence matrix} hold. If $\rho(R^{\mathrm{NPL}}) < 1$, then, for all sufficiently large $q$, the EM-NPL($q$) algorithm converges locally with probability approaching one.
\end{theorem}
The approximation error of the inner algorithm shrinks as $q$ grows, so $\rho(R^q)$ approaches $\rho(R^{\mathrm{NPL}})$. To translate this perturbation bound into a sufficient outer-iteration count, we specialize to the $\theta$-separable case (\Cref{assumption: theta separability}), in which \Cref{theorem: truncation invariance} guarantees that the EM-NPL($q$) and EM-NPL algorithms share a common fixed point. Let $\delta_q^{(k)}$ and $\delta_{\mathrm{NPL}}^{(k)}$ denote the deviations from this common fixed point generated by EM-NPL($q$) and EM-NPL, respectively, starting from the same initial deviation $\delta^{(0)}$. By \Cref{theorem: npl(q) local convergence}, the linearized recursions (dropping the $O_p(\cdot)$ remainder term) satisfy
\begin{equation*}
        \|\delta_q^{(k)}\| \leq \| (R^q)^k \|_F \|\delta^{(0)}\|, \qquad \|\delta_{\mathrm{NPL}}^{(k)}\| \leq \| (R^{\mathrm{NPL}})^k \|_F \|\delta^{(0)}\| .
\end{equation*}

The spectral radius $\rho(R^q)$ gives an upper bound on the linear convergence rate of the linearized recursion. The following result uses this to bound the number of outer iterations needed to reach a given tolerance. Let $\nu$ denote the size of the largest Jordan block of $R^{\mathrm{NPL}}$.
\begin{theorem}[Iteration Count Bound] \label{theorem: number of iterations bound}
    Suppose Assumptions \ref{assump: stationary initial}, \ref{assumption: theta separability} and \ref{assumption: approximation error upper bound}, and the assumptions of \Cref{proposition: mixture convergence matrix} hold, and $\rho(R^{\mathrm{NPL}}) < 1$. Then there exists a sequence $g(q)=O(f(q)^{1/\nu})$ such that $|\rho(R^{q}) - \rho(R^{\mathrm{NPL}})|\le g(q)$ for all $q$. Fix $c>0$ such that $\rho(R^{\mathrm{NPL}})+c<1$. Then there exists a constant $C_c$, not depending on $q$, such that $\|(R^q)^k\|_F \le C_c(\rho(R^q)+c)^k$ for all $k$ and all sufficiently large $q$. If $\varepsilon<C_c\|\delta^{(0)}\|$, then the linearized EM-NPL($q$) recursion satisfies $\|\delta_q^{(k)}\|\le \varepsilon$ for all
    \begin{equation*}
        k
        \ge
        K_q(c)
        \coloneqq
        \left\lceil
        \frac{\log\varepsilon-\log C_c-\log\|\delta^{(0)}\|}
        {\log\{\rho(R^q)+c\}}
        \right\rceil .
    \end{equation*}
    Choose $q$ sufficiently large so that $\rho(R^{\mathrm{NPL}})+g(q)+c<1$. Then,
    \begin{equation*}
        K_q(c)
        \le
        \left\lceil
        \frac{\log\varepsilon-\log C_c-\log\|\delta^{(0)}\|}
        {\log\{\rho(R^{\mathrm{NPL}})+g(q)+c\}}
        \right\rceil .
    \end{equation*}
\end{theorem}
\Cref{theorem: number of iterations bound} implies that the sufficient iteration count approaches the exact-solver benchmark as the perturbation $g(q)=O(f(q)^{1/\nu})$ shrinks. For superlinear inner solvers (GMRES, Newton's method), $g(q)$ can decay quickly in $q$, whereas for successive approximation it decays at rate $\rho_G^{q/\nu}$, where $\rho_G < 1$ is the Lipschitz constant of $G$. When $\rho(R^{\mathrm{NPL}})$ is close to one, the bound is most sensitive, and small perturbations in the outer rate can require noticeably more outer iterations.

\section{Practical Guidance} \label{sec: practical guidance}

\noindent This section provides practical guidance on choosing the fixed-point equation $G$, inner algorithm $\Gamma$, and truncation level $q$, which jointly determine the statistical and computational performance of the EM-NPL($q$) estimator. We organize the guidance around three stages: benchmarking inner solvers, selecting the algorithm--equation pair, and choosing $q$.

\subsubsection*{Stage 1: Benchmark Inner Solvers and Eliminate Dominated Algorithms}

Similar to the EM-NPL algorithm, the EM-NPL($q$) algorithm begins with an initial CCP estimate $P^{(0)}$. Natural candidates are: (i) random initialization; (ii) initialization from a model without unobserved heterogeneity, with the estimates perturbed \citep{arcidiacono2011conditional}; and (iii) initialization from a mixture model with unobserved heterogeneity (see \Cref{sec: initialization}). Given $P^{(0)}$, the first outer iteration solves the fixed-point equation to obtain the initial consistent estimator $Y^{(0)}$. This step is common to all EM-NPL($q$) variants and provides a natural opportunity to benchmark candidate inner solvers at negligible additional cost.

The practitioner should solve this initial system using each candidate algorithm, e.g., successive approximation (SA), GMRES, etc., and record the wall-clock time. Algorithms that are slower on the practitioner's specific problem can be immediately discarded as dominated. This benchmark is informative about the numerical structure of the problem. For example, for the policy valuation (PV) equation, properties such as the condition number of $(I - \beta F^{P})$ and its sparsity pattern, which govern the relative performance of different algorithms, are often preserved across subsequent EM-NPL iterations.

\subsubsection*{Stage 2: Select the Fixed-Point Equation and Algorithm}

Among the non-dominated algorithms from Stage~1, the practitioner next selects the combination of fixed-point equation $G$ and inner algorithm $\Gamma$ best suited to the model. This choice depends on three model features: the functional form of the utility function, the finite-dependence structure, and the discount factor $\beta$. \Cref{tab:algorithm_selection} summarizes the recommended algorithm--equation pairs based on these model features.

\paragraph{Functional form.} When the utility function is linear in parameters, $U(x,a) = \phi(x,a)^{\top}\theta$, the policy valuation problem (\Cref{ex: policy valuation}) and the EPL problem (\Cref{ex: epl}) both reduce to $\theta$-independent linear systems, solved once per outer iteration rather than at every parameter evaluation. By \Cref{theorem: truncation invariance}, EM-NPL($q$) is then numerically identical to the EM-NPL estimator for any $q \geq 1$, so it is efficient both for single-agent DDCs via policy valuation \citep{aguirregabiria2002swapping} and for dynamic games via EPL, and the choice of algorithm and $q$ can be guided solely by computational considerations.

When the utility function is nonlinear in parameters (\Cref{ex: bellman equation}), the inner problem is nonlinear and must be solved at every parameter evaluation. In this case, we recommend using Newton's method, as it typically converges much faster than SA. The choice between Bellman and Euler depends on the model structure and ease of implementation.

\paragraph{Finite dependence.} If the model exhibits finite dependence (\Cref{ex: euler}), the Euler equation representation is available. Because its contraction modulus is strictly less than $\beta$, even SA converges relatively quickly.

\paragraph{Discount factor.} The discount factor governs the difficulty of the inner fixed-point problem. When $\beta$ is small (e.g., $\beta = 0.8$), the choice of algorithm matters little. When $\beta$ is close to unity (e.g., $\beta = 0.9999$), SA becomes extremely slow.

\begin{table}[h!]
\centering
\caption{Algorithm selection based on model features}
\label{tab:algorithm_selection}
\footnotesize
\begin{tabular}{@{}l@{\hskip 8pt}l@{\hskip 8pt}p{5.8cm}@{}}
\toprule
Model features & Recommended $(\Gamma, G)$ & Rationale \\
\midrule
Linear utility & GMRES, PV/EPL & Superlinear convergence; no efficiency loss (\Cref{theorem: truncation invariance}) \\[6pt]
Finite dependence & SA, Euler & Lipschitz constant $< \beta$ \\[6pt]
Nonlinear, non-finite dependence, moderate $\beta$ & Newton, Bellman & Quadratic inner convergence \\
\bottomrule
\end{tabular}
\end{table}

\subsubsection*{Stage 3: Choose $q$}

The final choice is the number of inner iterations $q$. With superlinear solvers (GMRES, Newton), simulations show $q = 4$ to $6$ usually suffices for the outer loop to converge; we recommend starting low and increasing $q$ only if the outer loop fails to converge. SA may need larger values ($q = 8$ to $10$, or more) when $\beta$ is large, since the inner contraction modulus approaches one as $\beta \to 1$. EPL is a separate case. The inner map $\nabla_v \Phi(\tilde{\theta}, \tilde{v})$ in \Cref{ex: epl} is not a contraction; our simulations favor $q = 8$ to $12$ for computational speed. As in Stage 2, this affects only the speed of outer-loop convergence under linear utility, not statistical efficiency (\Cref{theorem: truncation invariance}).

\paragraph{Convergence check.} If the EM-NPL($q$) algorithm fails to converge at the initial $q$, increase $q$ and re-run. The convergence results in \Cref{sec: theory} guarantee that convergence is eventually achieved for sufficiently large $q$ whenever the standard EM-NPL algorithm converges.

\paragraph{Computational cost considerations.} The total computational cost is $\text{CT}(q) \approx h(q) \times K_q(c)$, where $h(q)$ is the cost of $q$ inner iterations and $K_q(c)$ is the sufficient outer-iteration bound in \Cref{theorem: number of iterations bound}. It is generally U-shaped in $q$. Too small a $q$ requires many outer iterations (large $K_q(c)$), while too large a $q$ wastes computation on unnecessary inner precision (large $h(q)$).

\section{Monte Carlo Simulations} \label{sec: simulation}

\noindent This section presents two Monte Carlo simulations to evaluate the performance of the proposed EM-NPL($q$) estimator. The first simulation considers a single-agent entry-and-exit model with unobserved heterogeneity based on \citet{aguirregabiria2023solving}. The second applies the estimator to a dynamic game based on \citet{aguirregabiria2007sequential}.

\subsection{Single-Agent Entry and Exit with Unobserved Heterogeneity} \label{sec: mixture mc}

\subsubsection{Design of Simulation}

\noindent We consider a single-firm entry and exit problem with unobserved heterogeneity that extends the model of \citet{aguirregabiria2023solving} to a three-type mixture model ($M = 3$).  At each period $t$, a firm decides whether to exit ($a_{t}=0$) or enter ($a_{t}=1$) the market.  An active firm ($a_t=1$) earns a flow payoff determined by variable profits ($VP_t$), fixed operating costs ($FC_t$), and entry cost ($EC_t$); an inactive firm receives a normalized payoff of zero
\begin{equation*}
    U(a_{t}, x_{t}, \eps_{t}) = \begin{cases}
        VP_t - FC_t - (1 - a_{t-1})\, EC_t + \eps_{t}(1) & \text{if } a_{t} = 1 \\
        \eps_{t}(0) & \text{if } a_{t} = 0
    \end{cases}
\end{equation*}
where $VP_t = (\theta_0^{VP} + \theta_1^{VP} z_{1t} + \theta_2^{VP} z_{2t}) \exp(w_t)$, $FC_t = \theta_0^{FC} + \theta_1^{FC} z_{3t}$, and $EC_t = \theta_0^{EC} + \theta_1^{EC} z_{4t}$.  The state vector is $x_t = (w_t, z_{1t}, z_{2t}, z_{3t}, z_{4t}, a_{t-1})$, where $z_{jt}$ follows an AR(1) process $z_{jt} = 0.6 z_{jt-1} + \eta_{jt}$ for $j=1,\dots,4$ with independent standard normal innovations.  We consider two scenarios for the productivity process $w_{t}$: (i)~finite dependence (FD): $w_t = 0.2 + 0.6 w_{t-1} + \eta_t$; (ii)~non-finite dependence (NFD): $w_t = 0.2 + 0.3 a_{t-1} + 0.6 w_{t-1} + \eta_t$.  The state space is discretized using Tauchen's method \citep{tauchen1986finite} with 6 grid points per variable, giving $|\mX| = 2 \times 6^5 = 15{,}552$.\footnote{For the NFD case, write $w_t = \gamma_0 + \gamma_1 a_{t-1} + \gamma_2 w_{t-1} + \eta_t$, where $(\gamma_0,\gamma_1,\gamma_2)=(0.2,0.3,0.6)$, $\eta_t \sim \mathcal{N}(0,\sigma_w^2)$, $\sigma_w=1$, and $n_{\sigma}=3$. The grid for $w$ spans $[\min(\mu_{d=0}, \mu_{d=1}) - n_{\sigma}\sigma_{LR},\, \max(\mu_{d=0}, \mu_{d=1}) + n_{\sigma}\sigma_{LR}]$, where $\mu_d=(\gamma_0+\gamma_1d)/(1-\gamma_2)$ and $\sigma_{LR}=\sigma_w/\sqrt{1-\gamma_2^2}$.} Because the discretized state components evolve independently conditional on the lagged action, we apply the transition operator using the Kronecker factorization described in \Cref{sec: kronecker}, without constructing the full transition matrix.

Firms are drawn from one of three unobserved types with type-specific parameter vectors $\theta_{0}^1 = (1.5, 1.5, -0.3, -0.3, -0.2, -0.3, -1)$, $\theta_{0}^2 = (0.2, 0.2, -0.2, -3.5, -2.0, -0.5, -3)$, and $\theta_{0}^3 = (0.8, 0.8, -1.0, -1.5, -0.8, -3.0, -1)$, with mixing weights $\pi_0 = (0.5, 0.3, 0.2)$. Type~1 has high variable profits and low costs, Type~2 has low profits and high entry costs, and Type~3 is intermediate. We simulate $N = 5{,}000$ firms, each observed for $T = 20$ periods, with discount factor $\beta \in \{0.8, 0.95, 0.9999\}$. We perform 500 Monte Carlo replications.

We compare six fixed-point--algorithm combinations: (i)~PV\_GMRES solves the policy valuation linear system $(I - \beta F^P)V=U_P$ using GMRES; (ii)~PV\_SA uses successive approximation on the same system; (iii)~BM\_SA applies successive approximation to the Bellman equation; (iv)~BM\_AD uses Anderson acceleration on the Bellman equation; (v)~BM\_NT applies Newton's method to the Bellman equation; and (vi)~EE\_SA uses successive approximation on the Euler equation (available in the FD case only).  For each method, we consider $q \in \{4, 6, 8, 10\}$ and report full-convergence baselines for PV\_GMRES, BM\_NT, and EE\_SA. Both policy valuation methods initialize $V^{(0)}$ by fully solving the policy valuation system with GMRES. They differ in the solver used for the subsequent $q$-step policy valuation updates.\footnote{We find that using SA for the initialization is at least ten times slower than using GMRES. The reported computational times include this common GMRES initialization.} Convergence is declared when $\max\{\|P^{(k)} - P^{(k-1)}\|_\infty,\, \|\theta^{(k)} - \theta^{(k-1)}\|_\infty,\, \|V^{(k)} - V^{(k-1)}\|_\infty / (1 + \|V^{(k-1)}\|_\infty)\} < 10^{-3}$. The inner solver tolerance is $10^{-8}$.

\subsubsection{Results}

\begin{table}[htbp]
\centering
\begin{threeparttable}
\caption{Mixture Model Results ($M=3$, $N=5000$, $\beta=0.95$)}
\small
\setlength{\tabcolsep}{4pt}
\label{tab:mixture_beta_0_95}
\begin{tabular}{ll|cccc|cccc}
\toprule
 & & \multicolumn{4}{c|}{Finite Dependence} & \multicolumn{4}{c}{Non-Finite Dependence} \\
\cmidrule(lr){3-6} \cmidrule(lr){7-10}
Method & $q$ & CT (s) & MSE & Avg Iter & Conv (\%) & CT (s) & MSE & Avg Iter & Conv (\%) \\
\midrule
PV\_GMRES & 4 & \phantom{0}3.14 (0.42) & 0.024 & \phantom{0}8.6 & 100.0 & \phantom{0}3.69 (0.30) & 0.050 & \phantom{0}5.8 & 100.0 \\
PV\_GMRES & 6 & \phantom{0}3.23 (0.57) & 0.024 & \phantom{0}8.6 & 100.0 & \phantom{0}3.63 (0.30) & 0.050 & \phantom{0}5.6 & 100.0 \\
PV\_GMRES & 8 & \phantom{0}3.38 (0.49) & 0.024 & \phantom{0}8.5 & 100.0 & \phantom{0}3.68 (0.33) & 0.050 & \phantom{0}5.7 & 100.0 \\
PV\_GMRES & 10 & \phantom{0}3.61 (0.61) & 0.024 & \phantom{0}8.5 & 100.0 & \phantom{0}3.81 (0.36) & 0.050 & \phantom{0}5.7 & 100.0 \\
\midrule
PV\_SA & 4 & \phantom{0}3.27 (0.46) & 0.024 & \phantom{0}9.1 & 100.0 & \phantom{0}3.75 (0.30) & 0.050 & \phantom{0}5.8 & 100.0 \\
PV\_SA & 6 & \phantom{0}3.50 (0.58) & 0.024 & \phantom{0}9.1 & 100.0 & \phantom{0}3.78 (0.34) & 0.050 & \phantom{0}5.8 & 100.0 \\
PV\_SA & 8 & \phantom{0}3.69 (0.62) & 0.024 & \phantom{0}8.9 & 100.0 & \phantom{0}3.87 (0.35) & 0.050 & \phantom{0}5.7 & 100.0 \\
PV\_SA & 10 & \phantom{0}3.90 (0.73) & 0.024 & \phantom{0}8.8 & 100.0 & \phantom{0}3.98 (0.37) & 0.050 & \phantom{0}5.7 & 100.0 \\
\midrule
BM\_SA & 4 & \phantom{0}8.82 (0.99) & 0.024 & 18.1 & 100.0 & 15.38 (1.37) & 0.050 & 18.0 & 100.0 \\
BM\_SA & 6 & 10.35 (1.21) & 0.024 & 14.1 & 100.0 & 15.81 (1.53) & 0.050 & 14.1 & 100.0 \\
BM\_SA & 8 & 11.72 (1.43) & 0.024 & 11.2 & 100.0 & 16.01 (1.52) & 0.050 & 11.0 & 100.0 \\
BM\_SA & 10 & 13.60 (1.67) & 0.024 & 10.3 & 100.0 & 18.31 (1.76) & 0.050 & 10.0 & 100.0 \\
\midrule
BM\_AD & 4 & 14.87 (3.30) & 0.024 & 10.9 & 100.0 & 24.79 (3.89) & 0.050 & 10.5 & 100.0 \\
BM\_AD & 6 & 13.50 (2.88) & 0.024 & \phantom{0}8.3 & 100.0 & 17.99 (2.30) & 0.050 & \phantom{0}6.1 & 100.0 \\
BM\_AD & 8 & 15.74 (2.15) & 0.024 & \phantom{0}7.3 & 100.0 & 20.35 (2.04) & 0.050 & \phantom{0}5.5 & 100.0 \\
BM\_AD & 10 & 19.35 (2.30) & 0.024 & \phantom{0}6.3 & 100.0 & 25.53 (2.49) & 0.050 & \phantom{0}5.5 & 100.0 \\
\midrule
BM\_NT & 4 & 21.90 (2.99) & 0.024 & \phantom{0}7.3 & 100.0 & 36.68 (4.21) & 0.050 & \phantom{0}5.5 & 100.0 \\
BM\_NT & 6 & 21.98 (3.02) & 0.024 & \phantom{0}7.3 & 100.0 & 36.89 (4.11) & 0.050 & \phantom{0}5.5 & 100.0 \\
BM\_NT & 8 & 21.95 (2.97) & 0.024 & \phantom{0}7.3 & 100.0 & 36.96 (4.23) & 0.050 & \phantom{0}5.5 & 100.0 \\
BM\_NT & 10 & 21.92 (2.89) & 0.024 & \phantom{0}7.3 & 100.0 & 36.92 (4.20) & 0.050 & \phantom{0}5.5 & 100.0 \\
\midrule
EE\_SA & 4 & \phantom{0}4.37 (0.55) & 0.024 & \phantom{0}6.8 & 100.0 & --- & --- & --- & --- \\
EE\_SA & 6 & \phantom{0}5.92 (0.80) & 0.024 & \phantom{0}7.2 & 100.0 & --- & --- & --- & --- \\
EE\_SA & 8 & \phantom{0}7.43 (1.02) & 0.024 & \phantom{0}7.3 & 100.0 & --- & --- & --- & --- \\
EE\_SA & 10 & \phantom{0}8.84 (1.21) & 0.024 & \phantom{0}7.2 & 100.0 & --- & --- & --- & --- \\
\midrule
PV\_GMRES & $\infty$ & \phantom{0}4.11 (0.70) & 0.024 & \phantom{0}8.5 & 100.0 & \phantom{0}4.32 (0.44) & 0.050 & \phantom{0}5.7 & 100.0 \\
BM\_NT & $\infty$ & 31.07 (4.70) & 0.024 & \phantom{0}7.3 & 100.0 & 51.29 (5.89) & 0.050 & \phantom{0}5.5 & 100.0 \\
EE\_SA & $\infty$ & 15.65 (2.23) & 0.024 & \phantom{0}7.3 & 100.0 & --- & --- & --- & --- \\
\bottomrule
\end{tabular}
\vspace{0.5em}
{\footnotesize \textit{Notes:} Results are averaged across 500 simulations ($M=3$ types, $N=5000$, $\beta=0.95$). $q$ denotes the number of inner iterations. CT reports the average computational time in seconds with standard deviation in parentheses. MSE is $\sum_{m=1}^{M}(\|\hat{\theta}_m - \theta_{m,0}\|^2 + \|\hat{\pi}_m - \pi_{m,0}\|^2)$ averaged across simulations. EM convergence is declared when $\max\{\|P^{(k)} - P^{(k-1)}\|_\infty,\, \|\theta^{(k)} - \theta^{(k-1)}\|_\infty,\, \|Y^{(k)} - Y^{(k-1)}\|_\infty / (1 + \|Y^{(k-1)}\|_\infty)\} < 10^{-3}$. The rows with $q=\infty$ show the results when the fixed-point equation is solved to full convergence.}
\end{threeparttable}
\end{table}

\noindent \Cref{tab:mixture_beta_0_95} reports the results for $\beta=0.95$. \Cref{tab:mixture_beta_0_80,tab:mixture_beta_0_9999} in \Cref{sec: additional simulation results} report the results for $\beta=0.80$ and $0.9999$. Three main findings emerge. First, truncating GMRES generally reduces computation time without affecting the estimates. Relative to full-convergence PV\_GMRES, $q=4$ reduces runtime by approximately 22\% and 18\% at $\beta=0.80$, 24\% and 15\% at $\beta=0.95$, and 26\% and 7\% at $\beta=0.9999$, for the finite- and non-finite-dependence cases, respectively. Because PV\_GMRES and PV\_SA share the same GMRES initialization, their total computational times are close. Second, the MSE of PV\_GMRES and PV\_SA is essentially unchanged across truncation levels in every case, consistent with \Cref{theorem: truncation invariance}. All methods converge in almost all replications at $\beta=0.80$ and $0.95$. At $\beta=0.9999$, however, BM\_SA reaches the 100-iteration limit in every replication. BM\_AD also fails at $q=4$, converges only partially at $q=6$, and achieves full convergence at $q\geq8$. Third, the policy valuation methods are substantially faster than the Bellman-based alternatives. At $\beta=0.95$ and $q=4$, PV\_GMRES takes 3.14 seconds in the finite-dependence case and 3.69 seconds in the non-finite-dependence case. The corresponding times are 3.27 and 3.75 seconds for PV\_SA, 8.82 and 15.38 seconds for BM\_SA, and 21.90 and 36.68 seconds for BM\_NT. The advantage remains large at $\beta=0.9999$: full-convergence PV\_GMRES takes 4.58 seconds in the finite-dependence case and 10.98 seconds in the non-finite-dependence case, compared with 38.01 and 76.33 seconds for full-convergence BM\_NT.

\subsection{Dynamic Games} \label{sec: dynamic games}

\subsubsection{Design of Simulation}

\noindent There are $J = 7$ firms competing across independent markets. In each period $t$, firms simultaneously choose whether to be active ($a_{jt} = 1$) or inactive ($a_{jt} = 0$). The flow payoff for firm $j$ is 
\begin{equation*}
    U_{j}(a_{jt}, x_{t}, \eps_{jt}) = \begin{cases}
        \begin{aligned}
            &\theta_{\text{RS}} s_{t} + \theta_{\text{Z}_1} z_{1t} - \theta_{\text{Z}_2} z_{2t} - \theta_{\text{RC}} \log\bigl(1 + \textstyle\sum_{j' \neq j} a_{j't}\bigr) \\
            &\quad {}- \theta_{\text{EC}}(1 - a_{j,t-1}) - \theta_{\text{FC},j} + \eps_{jt}(1)
        \end{aligned} & \text{if } a_{jt} = 1 \\[8pt]
        \eps_{jt}(0) & \text{if } a_{jt} = 0
    \end{cases}
\end{equation*}
where $s_{t} \in \{1,2,3,4,5\}$ is the market size.\footnote{$s_{t}$ follows a Markov chain with transition matrix $\begin{pmatrix}
    0.8 & 0.2 & 0   & 0   & 0   \\
    0.2 & 0.6 & 0.2 & 0   & 0   \\
    0   & 0.2 & 0.6 & 0.2 & 0   \\
    0   & 0   & 0.2 & 0.6 & 0.2 \\
    0   & 0   & 0   & 0.2 & 0.8
\end{pmatrix}
$} The demand shifter $z_{1t}$ and cost shifter $z_{2t}$ evolve according to $z_{\ell t}=0.6z_{\ell,t-1}+\nu_{\ell t}$ for $\ell\in\{1,2\}$, where the innovations $\nu_{\ell t}\sim N(0,1)$ are mutually independent, and $a_{j,t-1}$ is firm $j$'s action in the previous period. We discretize each of $z_{1t}$ and $z_{2t}$ using Tauchen's method \citep{tauchen1986finite} with five grid points. The state vector is $x_t=(s_t,z_{1t},z_{2t},a_{1,t-1},\ldots,a_{J,t-1})$, so the state space has cardinality $|\mX|=5\times5\times5\times2^7=16{,}000$. $\theta_{\text{RS}}$, $\theta_{\text{Z}_1}$, and $\theta_{\text{Z}_2}$ govern the effects of market size, the demand shifter, and the cost shifter, respectively. $\theta_{\text{RC}}$ captures the strategic interaction, $\theta_{\text{EC}}$ is the entry cost parameter, and $\theta_{\text{FC},j}$ are firm-specific fixed costs.

The true parameter values are $\theta_{\text{RS}}=1$, $\theta_{\text{Z}_1}=0.5$, $\theta_{\text{Z}_2}=0.3$, $\theta_{\text{EC}}=1$, $\beta=0.95$, and $\theta_{\text{FC}}=(1.9,1.8,1.7,1.6,1.5,1.4,1.3)$. We consider two values of $\theta_{\text{RC}}\in\{2.4,4\}$, referred to as the \textit{mildly unstable} and \textit{very unstable} cases following the terminology of \citet{aguirregabiria2021imposing}. Because $s_t$, $z_{1t}$, and $z_{2t}$ evolve independently, their joint transition matrix is $P_{\mathrm{exog}}=P_S\otimes P_{\text{Z}_1}\otimes P_{\text{Z}_2}$. As in the single-agent simulation, we apply the exogenous transition operator using the Kronecker factorization described in \Cref{sec: kronecker} and combine it with the conditional transition probabilities over lagged action profiles without constructing the full game transition matrix.

\subsubsection{Results}

\noindent We compare five estimator--algorithm combinations. All methods initialize CCPs via a pooled logit estimator with firm fixed effects. (i)~PV\_GMRES solves the policy valuation linear system $(I-\beta F^P)V=U_P$ for each firm $j$ using GMRES, where $F^P$ is the $16{,}000\times16{,}000$ transition operator induced by the firms' CCPs and the exogenous state process. (ii)~PV\_SA replaces GMRES with successive approximation on the same linear system. (iii)~BM\_SA replaces the policy valuation step with the Bellman equation and applies $q$ SA iterations to the Bellman operator while fixing other firms' CCPs. (iv)~EPL\_GMRES uses the EPL mapping (\Cref{ex: epl}) and solves the resulting linear system using matrix-free GMRES. EPL\_GMRES($\infty$) corresponds to the estimator of \citet{fukasawa2025epl}. The EPL system has a substantially larger dimension, $J|\mA||\mX|=7\times2\times16{,}000=224{,}000$, compared to the policy valuation system's $16{,}000$. (v)~PV\_EPL\_GMRES first runs PV\_GMRES until the NPL outer loop converges and then uses the resulting parameter, CCP, and value-function estimates to initialize EPL\_GMRES($q$). This benchmark examines whether the computationally cheaper policy valuation iterations provide a good warm start for the computationally demanding EPL iterations. For each method, we consider $q\in\{4,8,12,16\}$ inner iterations. Full-convergence baselines are reported for PV\_GMRES, EPL\_GMRES, and PV\_EPL\_GMRES.

Following the spectral algorithm of \citet{aguirregabiria2021imposing}, we update the CCPs using a Barzilai--Borwein step to accelerate outer-loop convergence. Convergence is declared when $\max\{\|P^{(k)} - P^{(k-1)}\|_\infty,\, \|\theta^{(k)} - \theta^{(k-1)}\|_\infty,\, \|Y^{(k)} - Y^{(k-1)}\|_\infty / (1 + \|Y^{(k-1)}\|_\infty)\} < 10^{-3}$.  The inner solver tolerance is $10^{-8}$.  All results are averaged over 500 Monte Carlo simulations with $N = 1{,}600$ independent markets.

\Cref{tab:game_combined} reports the results. Three patterns emerge. First, PV\_GMRES is the fastest method by a wide margin. In the mildly unstable design, PV\_GMRES($q=4$) takes 0.41 minutes, compared with 0.51 minutes for PV\_SA, 1.43 minutes for BM\_SA, and 10.84 minutes for EPL\_GMRES. Under stronger strategic interaction, the corresponding times are 0.80, 0.93, 3.15, and 17.97 minutes. PV\_GMRES converges in every replication for all values of $q$. BM\_SA is less reliable in the very unstable design: its convergence rate is 95.6\% at $q=4$ and 100\% at $q\geq8$. EPL\_GMRES and PV\_EPL\_GMRES each converge in 99.8\% of replications at $q=4$ and in every replication at $q\geq8$ in this design.

Second, the EPL-based estimators achieve lower MSE than the NPL-based estimators. Under full inner convergence, EPL\_GMRES reduces MSE from 0.558 to 0.535 in the mildly unstable design, a reduction of approximately 4\%, and from 0.463 to 0.343 in the very unstable design, a reduction of approximately 26\%. In the latter design, EPL\_GMRES also reduces aggregate absolute bias from 0.139 to 0.016. PV\_EPL\_GMRES produces essentially the same MSE as EPL\_GMRES.

Third, inner truncation substantially reduces the computational cost of EPL when $q$ is sufficiently large. Relative to full-convergence EPL\_GMRES, $q=8$ in the mildly unstable design and $q=12$ in the very unstable design reduce runtime by about 58\% and achieve essentially the same MSE. In contrast, EPL\_GMRES with $q=4$ requires many more outer iterations and is therefore slower than with $q=8$ or $q=12$. The PV\_EPL\_GMRES warm start yields comparable MSE and can provide a further modest reduction in runtime. Thus, PV\_GMRES with $q=4$ remains the preferred default when computational speed is the priority, whereas EPL\_GMRES with $q=8$--$12$ is attractive when lower MSE justifies substantially greater computation.

\begin{table}[htbp]
\centering
\begin{threeparttable}
\caption{Dynamic Game Simulation Results ($\beta=0.95$, $N=1{,}600$)}
\footnotesize
\setlength{\tabcolsep}{3pt}
\label{tab:game_combined}
\begin{tabular}{ll|ccccc|ccccc}
\toprule
 &  & \multicolumn{5}{c|}{Mildly Unstable: $\theta_{\text{RC}}=2.4$} & \multicolumn{5}{c}{Very Unstable: $\theta_{\text{RC}}=4.0$} \\
\cmidrule(lr){3-7}
\cmidrule(lr){8-12}
Method & $q$ & CT (min) & MSE & Bias & Avg Iter & Conv & CT (min) & MSE & Bias & Avg Iter & Conv \\
\midrule
PV\_GMRES & 4 & 0.41 (0.10) & 0.558 & 0.096 & 8.6 & 100.0 & 0.80 (0.19) & 0.463 & 0.139 & 14.7 & 100.0 \\
PV\_GMRES & 8 & 0.43 (0.11) & 0.558 & 0.096 & 8.6 & 100.0 & 0.84 (0.20) & 0.463 & 0.139 & 14.6 & 100.0 \\
PV\_GMRES & 12 & 0.46 (0.11) & 0.558 & 0.096 & 8.6 & 100.0 & 0.88 (0.21) & 0.463 & 0.139 & 14.6 & 100.0 \\
PV\_GMRES & 16 & 0.49 (0.12) & 0.558 & 0.096 & 8.6 & 100.0 & 0.94 (0.23) & 0.463 & 0.139 & 14.6 & 100.0 \\
\midrule
PV\_SA & 4 & 0.51 (0.10) & 0.558 & 0.096 & 8.6 & 100.0 & 0.93 (0.18) & 0.463 & 0.140 & 15.7 & 100.0 \\
PV\_SA & 8 & 0.52 (0.10) & 0.558 & 0.096 & 8.6 & 100.0 & 0.92 (0.19) & 0.463 & 0.139 & 14.8 & 100.0 \\
PV\_SA & 12 & 0.53 (0.11) & 0.558 & 0.096 & 8.6 & 100.0 & 0.93 (0.20) & 0.463 & 0.139 & 14.6 & 100.0 \\
PV\_SA & 16 & 0.54 (0.11) & 0.558 & 0.096 & 8.6 & 100.0 & 0.94 (0.21) & 0.463 & 0.139 & 14.6 & 100.0 \\
\midrule
BM\_SA & 4 & 1.43 (0.52) & 0.580 & 0.023 & 14.4 & 100.0 & 3.15 (1.67) & 0.513 & 0.194 & 38.4 & 95.6 \\
BM\_SA & 8 & 2.13 (0.93) & 0.595 & 0.028 & 9.5 & 100.0 & 4.38 (1.66) & 0.510 & 0.189 & 21.9 & 100.0 \\
BM\_SA & 12 & 2.52 (1.01) & 0.594 & 0.025 & 8.1 & 100.0 & 4.93 (1.62) & 0.509 & 0.193 & 16.2 & 100.0 \\
BM\_SA & 16 & 2.60 (0.99) & 0.593 & 0.021 & 8.1 & 100.0 & 5.73 (1.93) & 0.509 & 0.197 & 15.5 & 100.0 \\
\midrule
EPL\_GMRES & 4 & 10.84 (2.38) & 0.535 & 0.112 & 16.2 & 100.0 & 17.97 (3.86) & 0.344 & 0.016 & 28.7 & 99.8 \\
EPL\_GMRES & 8 & 7.99 (1.60) & 0.535 & 0.111 & 6.4 & 100.0 & 14.78 (4.27) & 0.343 & 0.016 & 13.4 & 100.0 \\
EPL\_GMRES & 12 & 8.78 (2.08) & 0.535 & 0.111 & 5.6 & 100.0 & 12.92 (3.60) & 0.343 & 0.016 & 8.0 & 100.0 \\
EPL\_GMRES & 16 & 9.96 (2.60) & 0.535 & 0.111 & 5.4 & 100.0 & 13.14 (2.81) & 0.343 & 0.016 & 6.6 & 100.0 \\
\midrule
PV\_EPL\_GMRES & 4 & 9.95 (3.22) & 0.535 & 0.111 & 21.6 & 100.0 & 16.28 (3.24) & 0.343 & 0.016 & 36.7 & 99.8 \\
PV\_EPL\_GMRES & 8 & 7.92 (1.78) & 0.535 & 0.111 & 14.2 & 100.0 & 13.63 (3.36) & 0.343 & 0.016 & 24.8 & 100.0 \\
PV\_EPL\_GMRES & 12 & 8.93 (2.35) & 0.535 & 0.111 & 13.7 & 100.0 & 12.80 (2.70) & 0.343 & 0.016 & 21.3 & 100.0 \\
PV\_EPL\_GMRES & 16 & 10.12 (2.88) & 0.535 & 0.111 & 13.6 & 100.0 & 13.26 (2.39) & 0.343 & 0.016 & 20.4 & 100.0 \\
\midrule
PV\_GMRES & $\infty$ & 0.51 (0.13) & 0.558 & 0.096 & 8.6 & 100.0 & 0.97 (0.24) & 0.463 & 0.139 & 14.6 & 100.0 \\
EPL\_GMRES & $\infty$ & 18.46 (7.23) & 0.535 & 0.111 & 5.3 & 100.0 & 30.68 (8.26) & 0.343 & 0.016 & 6.1 & 100.0 \\
PV\_EPL\_GMRES & $\infty$ & 17.92 (7.77) & 0.535 & 0.111 & 13.6 & 100.0 & 28.66 (7.49) & 0.343 & 0.016 & 20.2 & 100.0 \\
\bottomrule
\end{tabular}
\vspace{0.5em}
{\footnotesize \textit{Notes:} $J=7$ firms, $|\mX|=16{,}000$, and the EPL system has dimension $224{,}000$. Results are averaged across 500 simulations. $q$ denotes the number of inner iterations. CT reports average computational time in minutes, with standard deviations in parentheses. MSE is $S^{-1}\sum_{s=1}^{S}\|\hat{\theta}_s-\theta^*\|^2$. Bias is $\|\bar b\|_1$, where $\bar b=S^{-1}\sum_{s=1}^{S}(\hat{\theta}_s-\theta^*)$. Avg Iter is the average number of outer iterations, and Conv is the convergence rate (\%). The inner solver tolerance is $10^{-8}$. Outer-loop convergence is declared when $\max\{\|P^{(k)}-P^{(k-1)}\|_\infty,\|\theta^{(k)}-\theta^{(k-1)}\|_\infty,\|Y^{(k)}-Y^{(k-1)}\|_\infty/(1+\|Y^{(k-1)}\|_\infty)\}<10^{-3}$. The rows with $q=\infty$ solve the inner linear system to full convergence.}
\end{threeparttable}
\end{table}

\paragraph{Summary.} Across both Monte Carlo designs, PV\_GMRES with a small truncation level ($q=4$) is among the fastest methods and achieves full convergence with MSE essentially unchanged across truncation levels. Consistent with \Cref{theorem: truncation invariance}, truncation invariance holds empirically for the policy valuation-based estimators. EPL-based estimators offer lower MSE at substantially higher computational cost. At well-chosen truncation levels, EPL\_GMRES reduces runtime by about 58\% relative to full-convergence EPL\_GMRES without materially changing MSE. These findings validate the practical recommendations in \Cref{sec: practical guidance}: start with PV\_GMRES at small $q$ for speed, and switch to EPL\_GMRES at moderate $q$ when efficiency is the priority.

\FloatBarrier

\section{Empirical Application}\label{sec:empirical}

\noindent We illustrate the EM-NPL($q$) estimator with an application to household demand for cola in London. The estimation results show that ignoring unobserved heterogeneity substantially biases estimates of price elasticities and counterfactual policy effects.

\subsection{Data}\label{subsec:data_processing}

\noindent We construct household-level purchase histories from the Worldpanel by Numerator Take Home data for London covering the period 2015--2019.  The panel records every grocery purchase made by a representative sample of households, including price, quantity, date, and product characteristics.  The Cola category combines bottled and canned colas.  We classify Coca-Cola and Pepsi products by diet type (Regular, Diet, or Zero/Max, labeled Zero for Coca-Cola and Max for Pepsi) and group all remaining brands into a single Other category. This yields seven choice alternatives: Coca-Cola Regular, Coca-Cola Diet, Coca-Cola Zero, Pepsi Regular, Pepsi Diet, Pepsi Max, and Other.  The six named groups jointly account for 87.1\% of category trips.

We aggregate multiple purchases by the same household on the same date to the trip level. When a household makes multiple cola purchases on a single day, we retain only the brand--type combination with the highest expenditure on that trip.  We restrict the sample to London households and to those with more than 10 and fewer than 500 category trips over the five-year window.

After these filters, the sample contains 75{,}559 trips. The overall repeat rate is 70.9\%, and the median inter-purchase spell is 9 days. Table~\ref{tab:cola_groups} reports the choice shares and repeat-purchase rates.

\begin{table}[htbp]
\centering
\caption{Cola: Brand $\times$ Diet Type Choice Groups}
\label{tab:cola_groups}
\small
\begin{tabular}{lcccc}
\toprule
Brand $\times$ Attribute & Share (\%) & Cumul.\ (\%) & Repeat (\%) \\
\midrule
Coca-Cola $\mid$ Regular & 24.1 & 24.1 & 75.6 \\
Coca-Cola $\mid$ Diet & 20.3 & 44.4 & 74.1 \\
Pepsi $\mid$ Max & 20.1 & 64.4 & 70.4 \\
Coca-Cola $\mid$ Zero & 10.2 & 74.6 & 61.7 \\
Pepsi $\mid$ Regular & 6.7 & 81.3 & 61.4 \\
Pepsi $\mid$ Diet & 5.7 & 87.1 & 55.8 \\
Other & 12.9 & 100.0 & 76.5 \\
\bottomrule
\end{tabular}
\begin{flushleft}
\footnotesize
\textit{Notes:} Trip-level choice shares and repeat-purchase rates for London households. Data from Worldpanel by Numerator, GB Take Home data, 2015--2019. Repeat~(\%) is the fraction of purchases where the household chose the same category as the previous trip.
\end{flushleft}
\end{table}

\subsection{Model}\label{subsec:model_estimation}
\noindent We assume that each household makes a purchase every 9 days, the median inter-purchase spell. At each purchase occasion~$t$, household~$i$ of type~$m$ chooses action $a_{it}\in\mA=\{0,\ldots,6\}$ to maximize the expected discounted utility, where the seven alternatives correspond to the Cola choice groups defined in Table~\ref{tab:cola_groups}: Coca-Cola Regular, Coca-Cola Diet, Coca-Cola Zero, Pepsi Regular, Pepsi Diet, Pepsi Max, and Other.  The per-period payoff for type~$m$ is
\begin{equation*}
  u(a,\, a_{-1}, p;\, \theta^m) = \gamma_a^m + \alpha^m p_a + \eta^m\, \mathbbm{1}\{a = a_{-1}\} + \eps(a),
\end{equation*}
where $p$ denotes the vector of current prices, $a_{-1}$ is the previous purchase, $p_a$ is the unit price of product~$a$, $\gamma_a^m$ captures type-specific product fixed effects, $\alpha^m$ is the type-specific price coefficient, and $\eta^m$ measures type-specific state dependence. The error $\eps(a)$ follows an i.i.d.\ Type I extreme value distribution. We collect the type-specific parameters into $\theta^m = (\gamma_1^m,\ldots,\gamma_6^m, \alpha^m, \eta^m) \in \mathbb{R}^8$. We set the discount factor to $\beta=0.99$.

Prices follow a first-order vector autoregression. We discretize the estimated VAR(1) with 5 grid points per product, yielding $5^7 = 78{,}125$ price states. Combined with seven previous-action states, the augmented state space has $7 \times 78{,}125 = 546{,}875$ elements.

We estimate a finite mixture model with $M$ latent types using the EM-NPL($q$) algorithm with $q=10$ GMRES iterations to solve the policy valuation problem. We first estimate a static logit model within the income group and initialize each type's parameter vector by adding an independent \(N(0,0.05^2I)\) perturbation to the static-logit estimate. For EM-NPL($q$), the convergence is declared when $\max\{\|\boldsymbol{\theta}^{(k)} - \boldsymbol{\theta}^{(k-1)}\|_\infty, \|\bfP^{(k)} - \bfP^{(k-1)}\|_\infty, \|\boldsymbol{\pi}^{(k)} - \boldsymbol{\pi}^{(k-1)}\|_\infty\} < 10^{-4}$. We estimate models for $M \in \{1, 2, \ldots, 10\}$.

As our utility function is linear in parameters and we solve the policy valuation problem, $\theta$-separability (\Cref{assumption: theta separability}) holds, so by \Cref{theorem: truncation invariance} our estimator is numerically identical to the EM-NPL estimator. To obtain standard errors for the counterfactual analysis, we draw 100 draws from the asymptotic distribution of the estimated parameters and report the standard deviation across draws.

\subsection{Estimation Results and Counterfactual Analysis}\label{subsec:estimation_results}

\noindent We estimate the model separately for three income groups defined by: low (under \pounds30{,}000), middle (\pounds30{,}000--\pounds49{,}999), and high (\pounds50{,}000 and above). Each income group contains approximately 300-500 households. Estimating separately by income group allows the distribution of latent types, and hence the pattern of brand preferences and price sensitivity, to differ across income segments.

\Cref{fig:computation_time} reports the average computational time of the EM-NPL($q$) algorithm for each $M$ and across the three income groups and 150 random initializations.\footnote{We also implemented the sieve-MLE initialization discussed in \Cref{sec: initialization}. However, it takes around 10 hours to compute the sieve-MLE estimates for $M=10$, which is prohibitively long for our application. Also, the sieve-MLE initialization does not yield higher likelihood values than the random initialization.} Computation time scales approximately linearly in $M$. For $M=10$, the average runtime is about 4.8 minutes, which is feasible for empirical applications.
\begin{figure}[htbp]
  \centering
  \includegraphics[width=0.6\textwidth]{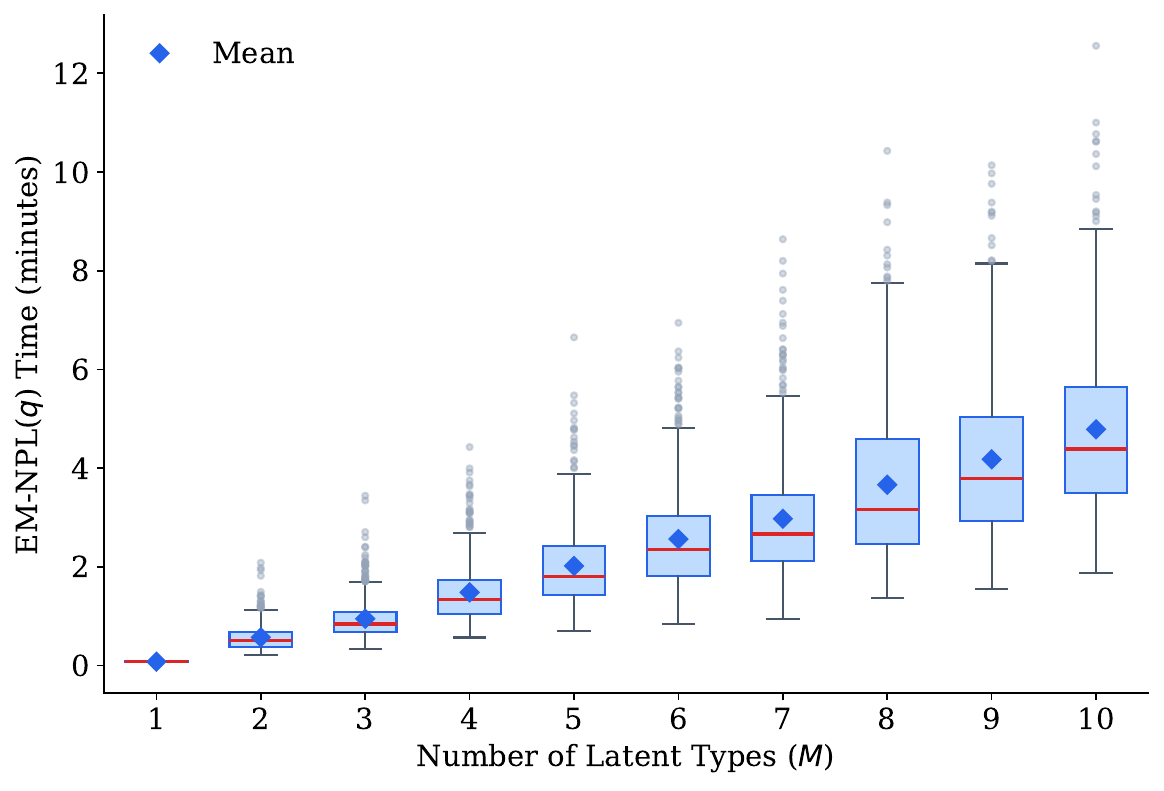}
  \caption{Computational time by number of latent types}
  \label{fig:computation_time}
\end{figure}

\subsubsection{Short-Run and Long-Run Price Elasticities}\label{sssec:sr_lr_elast}

\noindent We compute short-run and long-run price elasticities. The \emph{short-run} elasticity measures the contemporaneous response to an unanticipated, temporary 1\% price increase, holding continuation values fixed. The \emph{long-run} elasticity permanently shifts product~$j$'s price by 1\% and re-solves the dynamic programming problem. Figures~\ref{fig:lr_own_elast_low}--\ref{fig:sr_own_elast_high} plot both elasticities as a function of~$M$ for each income group, and Table~\ref{tab:cross_elas} reports the full long-run cross-price elasticity matrices at the benchmark specification.

In general, own-price elasticities become more negative across all three income groups as additional latent types are introduced. This pattern reflects stronger substitution once the model absorbs persistent taste heterogeneity. The elasticity paths generally stabilize around $M=5$, so we use $M=5$ as the benchmark specification.

For the low-income group (Figures~\ref{fig:lr_own_elast_low}--\ref{fig:sr_own_elast_low}), the long-run own-price elasticity for Coca-Cola Regular shifts from about $-0.91$ under homogeneity ($M=1$) to $-1.93$ (s.e.\ 0.42) at the benchmark $M=5$, while the short-run elasticity shifts from about $-0.10$ to $-0.44$. The largest movement is for Pepsi Regular, whose long-run own-price elasticity goes from about $-0.59$ to $-2.85$ (s.e.\ 0.65) and whose short-run elasticity goes from about $-0.13$ to $-0.80$. Coca-Cola Diet, in contrast, remains close to its homogeneous estimate, at $-0.75$ (s.e.\ 0.54). The cross-price estimates (Panel~A of Table~\ref{tab:cross_elas}) indicate particularly strong substitution from Coca-Cola Regular toward Pepsi Regular. A permanent 1\% increase in the price of Coca-Cola Regular increases Pepsi Regular's long-run share by 2.99\%, compared with 0.50\% for Coca-Cola Zero and 0.27\% for Coca-Cola Diet.

The middle-income group (Figures~\ref{fig:lr_own_elast_medium}--\ref{fig:sr_own_elast_medium} and Panel~B of Table~\ref{tab:cross_elas}) displays a similar pattern of own-price responses. The long-run own-price elasticities are $-1.70$ for Coca-Cola Regular, $-1.36$ for Coca-Cola Diet, $-1.12$ for Coca-Cola Zero, and $-2.69$ for Pepsi Regular. The cross-price matrix also indicates substantial cross-brand substitution: a permanent 1\% increase in the price of Coca-Cola Regular increases Pepsi Regular's long-run share by 2.24\%.

The high-income group (Figures~\ref{fig:lr_own_elast_high}--\ref{fig:sr_own_elast_high} and Panel~C of Table~\ref{tab:cross_elas}) also exhibits substantial price sensitivity. The long-run own-price elasticity is $-1.54$ (s.e.\ 0.66) for Coca-Cola Regular and $-1.34$ (s.e.\ 0.64) for Coca-Cola Zero, while Pepsi Max reaches $-0.95$ (s.e.\ 0.46). Pepsi Regular has the largest long-run own-price elasticity in magnitude, at $-2.80$ (s.e.\ 0.51), followed by Coca-Cola Regular. A permanent 1\% increase in the price of Coca-Cola Regular increases Pepsi Regular's long-run share by 2.79\%. The short-run elasticities are especially sensitive to~$M$ in all three groups. Under homogeneity, consumers appear locked into their habitual brand, and their immediate price response is suppressed; as heterogeneity is introduced, this spurious persistence is reduced and the short-run response becomes substantially larger.

\begin{figure}[htbp]
  \centering
  \begin{subfigure}[t]{0.48\textwidth}
    \centering
    \caption{Long-run --- Low income}
    \label{fig:lr_own_elast_low}
    \includegraphics[width=\textwidth]{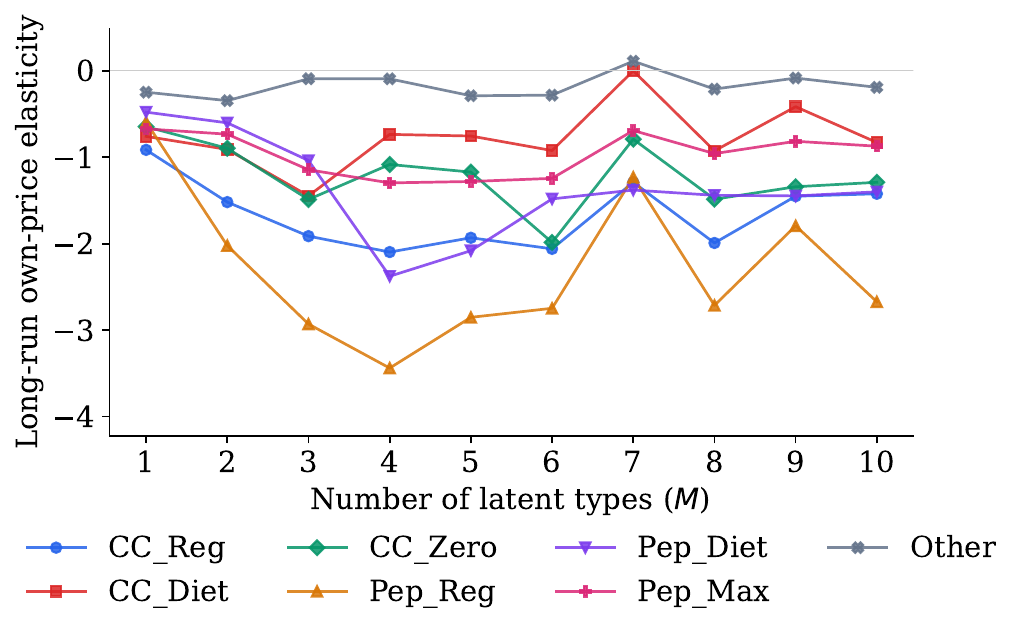}
  \end{subfigure}
  \hfill
  \begin{subfigure}[t]{0.48\textwidth}
    \centering
    \caption{Short-run --- Low income}
    \label{fig:sr_own_elast_low}
    \includegraphics[width=\textwidth]{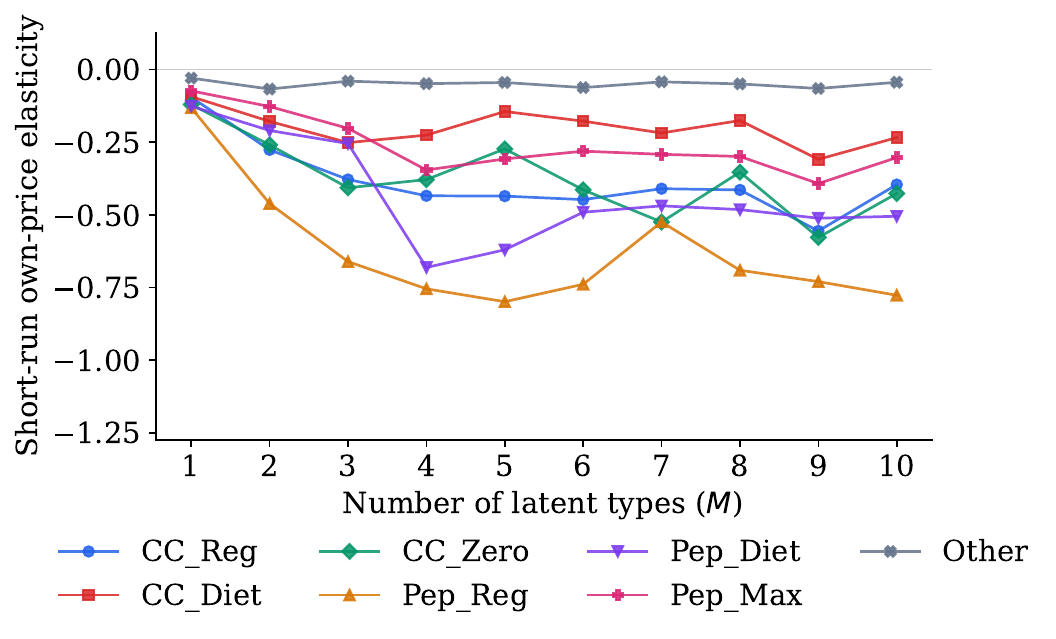}
  \end{subfigure}

  \vspace{0.5em}

  \begin{subfigure}[t]{0.48\textwidth}
    \centering
    \caption{Long-run --- Middle income}
    \label{fig:lr_own_elast_medium}
    \includegraphics[width=\textwidth]{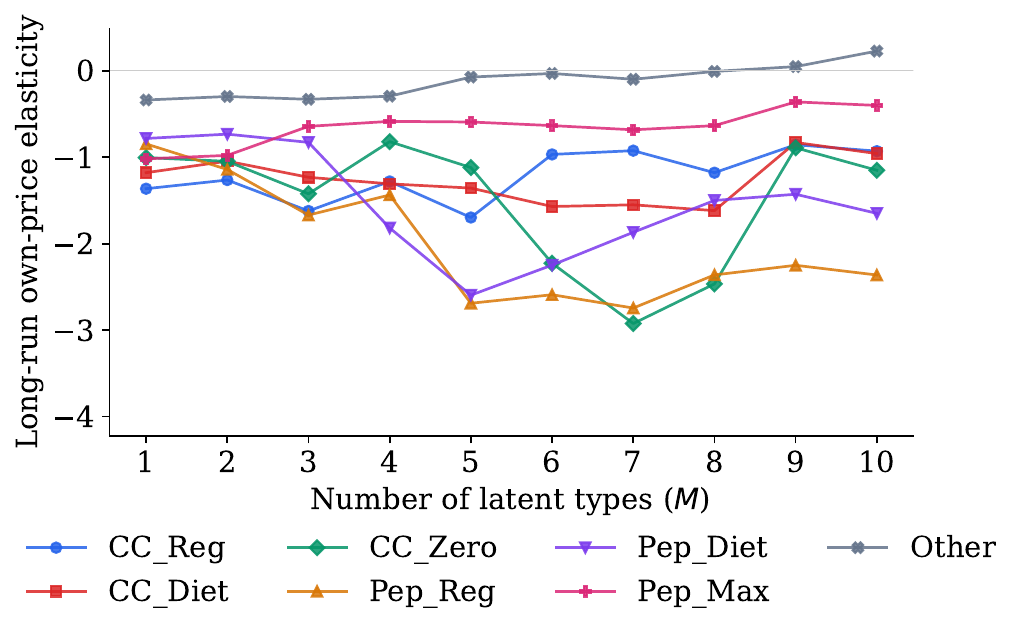}
  \end{subfigure}
  \hfill
  \begin{subfigure}[t]{0.48\textwidth}
    \centering
    \caption{Short-run --- Middle income}
    \label{fig:sr_own_elast_medium}
    \includegraphics[width=\textwidth]{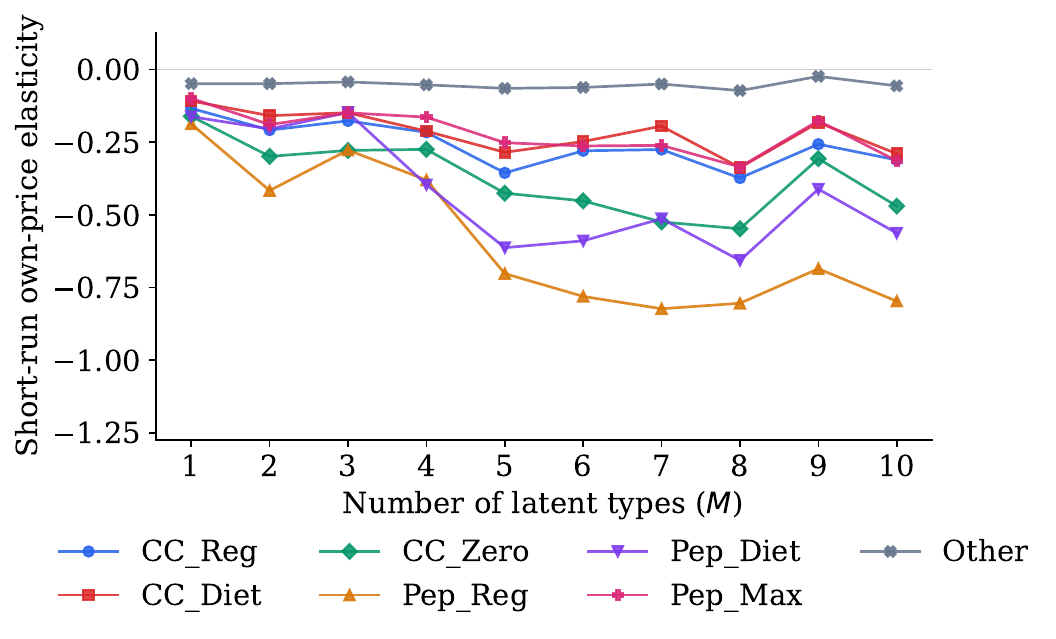}
  \end{subfigure}

  \vspace{0.5em}

  \begin{subfigure}[t]{0.48\textwidth}
    \centering
    \caption{Long-run --- High income}
    \label{fig:lr_own_elast_high}
    \includegraphics[width=\textwidth]{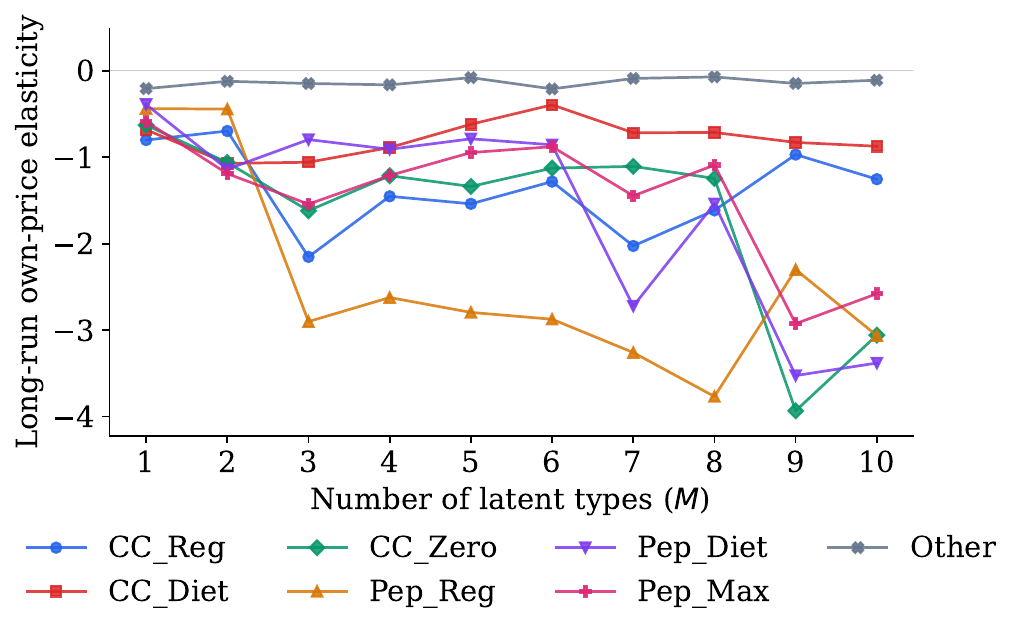}
  \end{subfigure}
  \hfill
  \begin{subfigure}[t]{0.48\textwidth}
    \centering
    \caption{Short-run --- High income}
    \label{fig:sr_own_elast_high}
    \includegraphics[width=\textwidth]{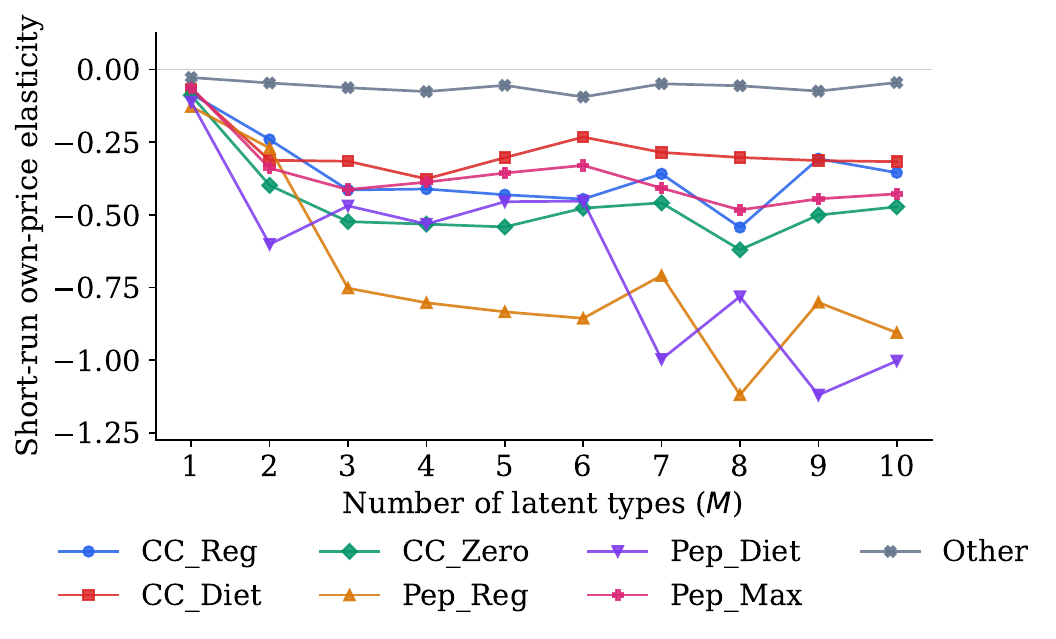}
  \end{subfigure}
  \caption{Own-price elasticities by number of latent types~$M$}
  \label{fig:own_elast}
\end{figure}

\begin{table}[htbp]
\centering
\caption{Long-run cross-price elasticities  ($M = 5$)}
\label{tab:cross_elas}
\footnotesize
\setlength{\tabcolsep}{3pt}
\begin{tabular}{lccccccc}
\toprule
 & CC\_Reg & CC\_Diet & CC\_Zero & Pep\_Reg & Pep\_Diet & Pep\_Max & Other \\
\midrule
\multicolumn{8}{l}{\textit{Panel A: Low income}} \\[3pt]
CC\_Reg & $-1.93\;(0.42)$ & $\phantom{-}0.16\;(0.11)$ & $\phantom{-}0.15\;(0.09)$ & $\phantom{-}0.74\;(0.14)$ & $\phantom{-}0.02\;(0.02)$ & $\phantom{-}0.27\;(0.10)$ & $\phantom{-}0.06\;(0.04)$ \\
CC\_Diet & $\phantom{-}0.27\;(0.18)$ & $-0.75\;(0.54)$ & $\phantom{-}0.09\;(0.17)$ & $\phantom{-}0.03\;(0.06)$ & $\phantom{-}0.02\;(0.02)$ & $\phantom{-}0.19\;(0.12)$ & $\phantom{-}0.06\;(0.07)$ \\
CC\_Zero & $\phantom{-}0.50\;(0.25)$ & $\phantom{-}0.17\;(0.29)$ & $-1.17\;(0.81)$ & $\phantom{-}0.06\;(0.10)$ & $\phantom{-}0.02\;(0.03)$ & $\phantom{-}0.25\;(0.19)$ & $\phantom{-}0.08\;(0.10)$ \\
Pep\_Reg & $\phantom{-}2.99\;(0.63)$ & $\phantom{-}0.07\;(0.09)$ & $\phantom{-}0.07\;(0.09)$ & $-2.85\;(0.65)$ & $\phantom{-}0.01\;(0.01)$ & $\phantom{-}0.14\;(0.07)$ & $\phantom{-}0.03\;(0.05)$ \\
Pep\_Diet & $\phantom{-}0.18\;(0.09)$ & $\phantom{-}0.07\;(0.06)$ & $\phantom{-}0.04\;(0.05)$ & $\phantom{-}0.01\;(0.01)$ & $-2.08\;(0.58)$ & $\phantom{-}1.78\;(0.51)$ & $\phantom{-}0.01\;(0.02)$ \\
Pep\_Max & $\phantom{-}0.49\;(0.17)$ & $\phantom{-}0.21\;(0.12)$ & $\phantom{-}0.14\;(0.12)$ & $\phantom{-}0.07\;(0.03)$ & $\phantom{-}0.44\;(0.14)$ & $-1.28\;(0.43)$ & $\phantom{-}0.05\;(0.05)$ \\
Other & $\phantom{-}0.37\;(0.20)$ & $\phantom{-}0.21\;(0.25)$ & $\phantom{-}0.15\;(0.20)$ & $\phantom{-}0.04\;(0.12)$ & $\phantom{-}0.01\;(0.02)$ & $\phantom{-}0.17\;(0.16)$ & $-0.29\;(0.23)$ \\
\midrule
\multicolumn{8}{l}{\textit{Panel B: Middle income}} \\[3pt]
CC\_Reg & $-1.70\;(0.67)$ & $\phantom{-}0.36\;(0.41)$ & $\phantom{-}0.19\;(0.14)$ & $\phantom{-}0.57\;(0.14)$ & $\phantom{-}0.05\;(0.03)$ & $\phantom{-}0.16\;(0.08)$ & $\phantom{-}0.02\;(0.03)$ \\
CC\_Diet & $\phantom{-}0.45\;(0.49)$ & $-1.36\;(0.47)$ & $\phantom{-}0.09\;(0.11)$ & $\phantom{-}0.07\;(0.03)$ & $\phantom{-}0.55\;(0.21)$ & $\phantom{-}0.14\;(0.08)$ & $\phantom{-}0.01\;(0.03)$ \\
CC\_Zero & $\phantom{-}0.56\;(0.30)$ & $\phantom{-}0.22\;(0.24)$ & $-1.12\;(0.57)$ & $\phantom{-}0.17\;(0.13)$ & $\phantom{-}0.03\;(0.03)$ & $\phantom{-}0.20\;(0.17)$ & $\phantom{-}0.01\;(0.02)$ \\
Pep\_Reg & $\phantom{-}2.24\;(0.54)$ & $\phantom{-}0.22\;(0.08)$ & $\phantom{-}0.21\;(0.23)$ & $-2.69\;(0.62)$ & $\phantom{-}0.07\;(0.04)$ & $\phantom{-}0.21\;(0.12)$ & $\phantom{-}0.02\;(0.01)$ \\
Pep\_Diet & $\phantom{-}0.27\;(0.12)$ & $\phantom{-}2.32\;(0.97)$ & $\phantom{-}0.06\;(0.07)$ & $\phantom{-}0.09\;(0.05)$ & $-2.60\;(0.91)$ & $\phantom{-}0.17\;(0.07)$ & $\phantom{-}0.01\;(0.01)$ \\
Pep\_Max & $\phantom{-}0.29\;(0.13)$ & $\phantom{-}0.20\;(0.12)$ & $\phantom{-}0.12\;(0.11)$ & $\phantom{-}0.10\;(0.05)$ & $\phantom{-}0.06\;(0.03)$ & $-0.59\;(0.39)$ & $-0.01\;(0.11)$ \\
Other & $\phantom{-}0.16\;(0.19)$ & $\phantom{-}0.07\;(0.18)$ & $\phantom{-}0.02\;(0.04)$ & $\phantom{-}0.05\;(0.04)$ & $\phantom{-}0.01\;(0.02)$ & $-0.05\;(0.57)$ & $-0.07\;(0.25)$ \\
\midrule
\multicolumn{8}{l}{\textit{Panel C: High income}} \\[3pt]
CC\_Reg & $-1.54\;(0.66)$ & $\phantom{-}0.20\;(0.27)$ & $\phantom{-}0.23\;(0.15)$ & $\phantom{-}0.48\;(0.14)$ & $\phantom{-}0.03\;(0.10)$ & $\phantom{-}0.22\;(0.15)$ & $\phantom{-}0.01\;(0.02)$ \\
CC\_Diet & $\phantom{-}0.23\;(0.32)$ & $-0.62\;(0.51)$ & $\phantom{-}0.08\;(0.12)$ & $\phantom{-}0.04\;(0.02)$ & $\phantom{-}0.09\;(0.11)$ & $\phantom{-}0.16\;(0.14)$ & $\phantom{-}0.01\;(0.01)$ \\
CC\_Zero & $\phantom{-}0.51\;(0.29)$ & $\phantom{-}0.16\;(0.20)$ & $-1.34\;(0.64)$ & $\phantom{-}0.06\;(0.04)$ & $\phantom{-}0.01\;(0.04)$ & $\phantom{-}0.59\;(0.36)$ & $\phantom{-}0.01\;(0.01)$ \\
Pep\_Reg & $\phantom{-}2.79\;(0.54)$ & $\phantom{-}0.18\;(0.10)$ & $\phantom{-}0.14\;(0.08)$ & $-2.80\;(0.51)$ & $\phantom{-}0.01\;(0.02)$ & $\phantom{-}0.14\;(0.10)$ & $\phantom{-}0.00\;(0.01)$ \\
Pep\_Diet & $\phantom{-}0.21\;(0.30)$ & $\phantom{-}0.55\;(0.43)$ & $\phantom{-}0.04\;(0.07)$ & $\phantom{-}0.01\;(0.02)$ & $-0.79\;(0.59)$ & $\phantom{-}0.09\;(0.10)$ & $-0.00\;(0.03)$ \\
Pep\_Max & $\phantom{-}0.36\;(0.22)$ & $\phantom{-}0.23\;(0.18)$ & $\phantom{-}0.45\;(0.28)$ & $\phantom{-}0.04\;(0.03)$ & $\phantom{-}0.02\;(0.04)$ & $-0.95\;(0.46)$ & $\phantom{-}0.02\;(0.03)$ \\
Other & $\phantom{-}0.07\;(0.13)$ & $\phantom{-}0.07\;(0.08)$ & $\phantom{-}0.03\;(0.05)$ & $\phantom{-}0.00\;(0.02)$ & $-0.00\;(0.05)$ & $\phantom{-}0.07\;(0.19)$ & $-0.08\;(0.14)$ \\
\bottomrule
\end{tabular}
\begin{flushleft}
\footnotesize
\textbf{Note:} Entry $(k,j)$ reports the point estimate with standard error in parentheses for the percentage change in product~$k$'s ergodic market share per 1\% permanent increase in product~$j$'s price. Standard errors are computed as the standard deviation across 100 draws from the asymptotic distribution of the estimated parameters.
\end{flushleft}
\end{table}

Figures~\ref{fig:cc_reg_type_low}--\ref{fig:cc_reg_type_high} decompose the long-run own-price elasticity of Coca-Cola Regular across latent types for each income group, where types are ranked from the least to the most price-elastic and the mixture probabilities $\pi_m$ are reported below each bar. The figures reveal substantial heterogeneity in price sensitivity across types. In the low-income group, Types~1--2 have relatively small and imprecisely estimated elasticities. The most prevalent type is Type~3 ($\pi_3=0.245$), with an elasticity of approximately $-2.44$, while Type~5 ($\pi_5=0.174$) is the most elastic, at approximately $-2.94$. In the middle-income group, Type~1 has a positive but imprecisely estimated elasticity whose confidence interval includes zero. The most prevalent type is Type~2 ($\pi_2=0.252$), with an elasticity of approximately $-0.59$, while Type~5 is the most elastic, at approximately $-3.07$. In the high-income group, the two least elastic types have confidence intervals that include zero. Type~5 ($\pi_5=0.289$) is both the most prevalent and the most elastic, with an elasticity of approximately $-2.53$. Types with elasticities near zero likely represent brand-loyal consumers whose Coca-Cola Regular purchases are relatively insensitive to price changes.

\begin{figure}[h!]
  \centering
  \begin{subfigure}[t]{0.45\textwidth}
    \centering
    \caption{Low income}
    \label{fig:cc_reg_type_low}
    \includegraphics[width=\textwidth]{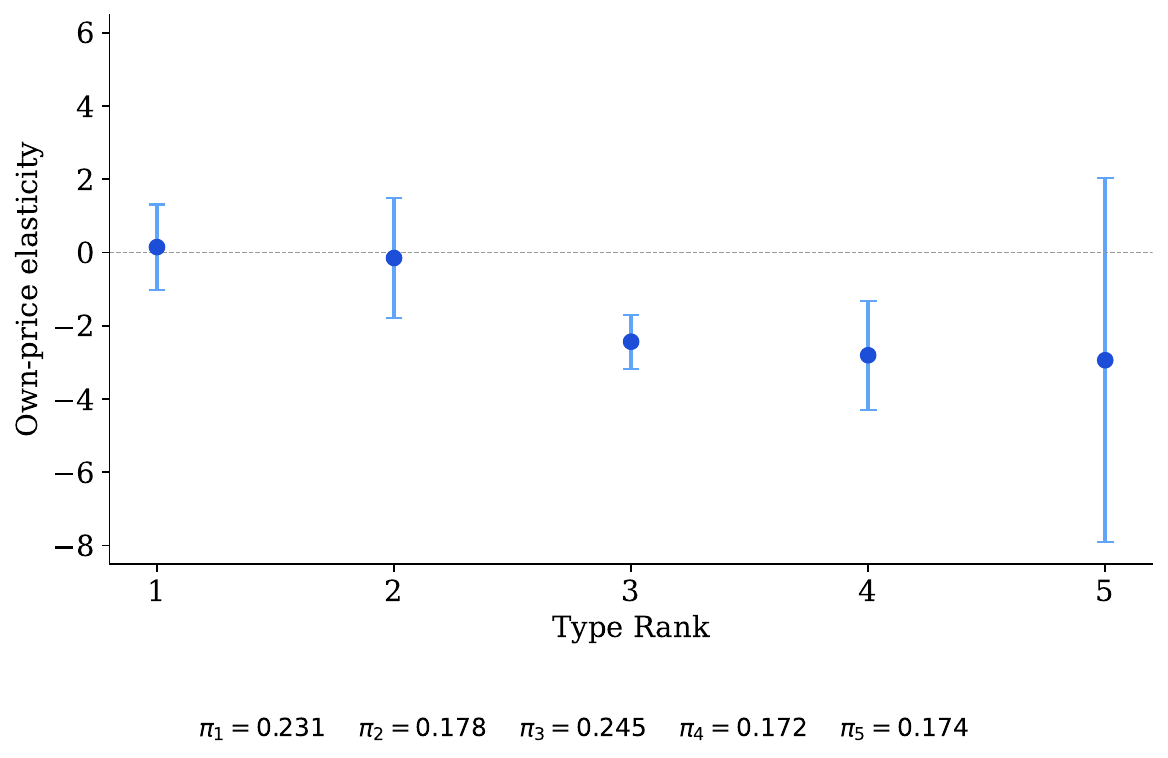}
  \end{subfigure}
  \hfill
  \begin{subfigure}[t]{0.45\textwidth}
    \centering
    \caption{Middle income}
    \label{fig:cc_reg_type_medium}
    \includegraphics[width=\textwidth]{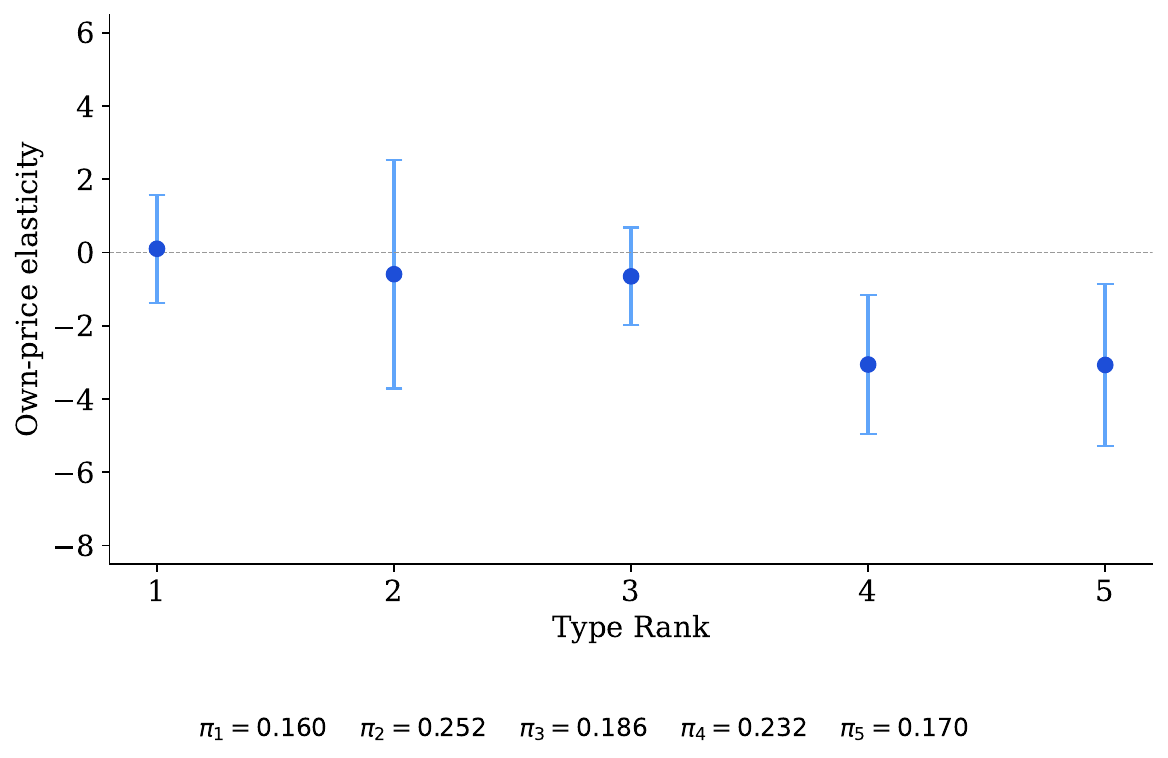}
  \end{subfigure}

  \begin{subfigure}[t]{0.45\textwidth}
    \centering
    \caption{High income}
    \label{fig:cc_reg_type_high}
    \includegraphics[width=\textwidth]{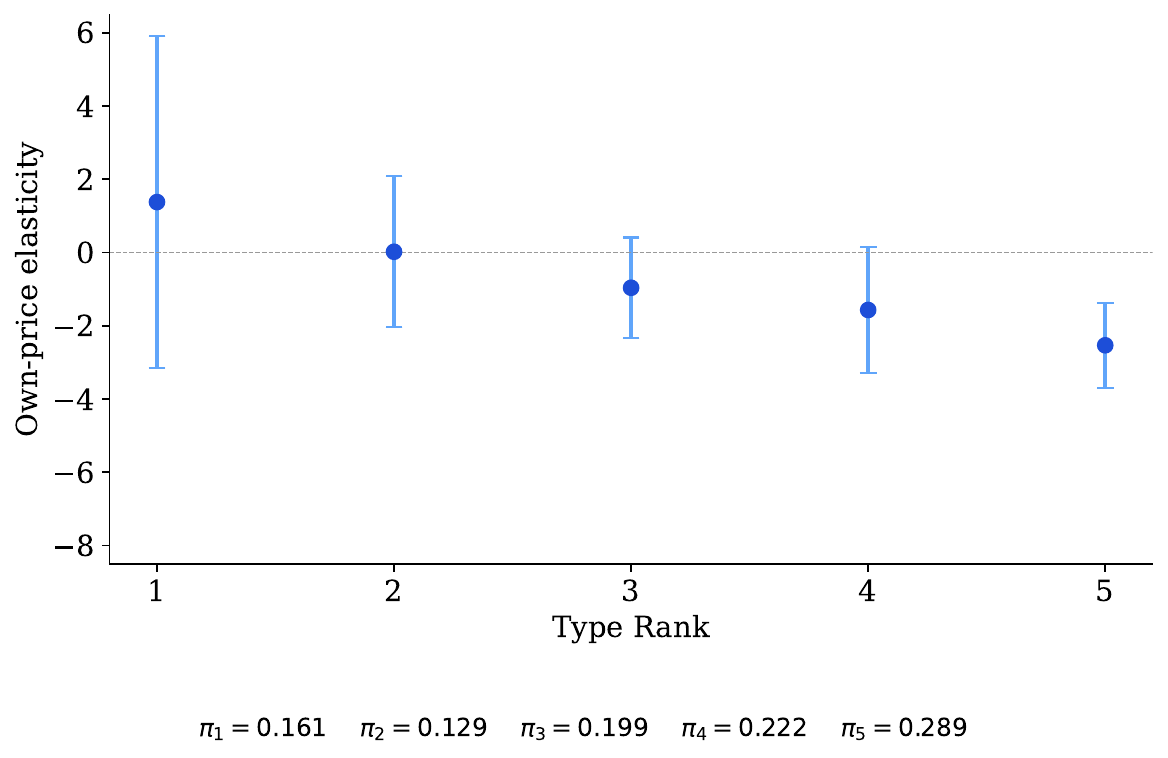}
  \end{subfigure}
  \caption{Long-run Coca-Cola Regular own-price elasticity by latent type, for each income group. Types are ranked by long-run Coca-Cola Regular own-price elasticity from least to most elastic, and the mixture probabilities $\pi_m$ are reported below each bar.}
  \label{fig:type_profile}
\end{figure}

\subsubsection{Soda Tax Policy Analysis}\label{sssec:welfare_analysis}

\noindent We consider a sugar-sweetened beverage tax of \pounds0.25 per litre applied to the two regular (full-sugar) products, Coca-Cola Regular and Pepsi Regular. This policy experiment is broadly comparable to the UK Soft Drinks Industry Levy introduced in 2018. We use one-time compensating variation (CV) to measure the consumers' welfare loss, defined as the amount of money needed to leave households indifferent between the baseline and taxed environments.

The left column of Figure~\ref{fig:cv_vs_M} plots the average CV as a function of~$M$ for each income group. CV rises sharply between $M=2$ and $M=3$ in the low- and high-income groups, from about \pounds3.3 to \pounds5.2 and from about \pounds1.0 to \pounds3.7, respectively. The increase is more gradual in the middle-income group. Beyond $M=3$, the point estimates remain well above the homogeneous benchmark but rise and fall as $M$ increases. The confidence intervals generally widen and overlap across many neighboring specifications. At the benchmark $M=5$, the average CV is about \pounds5.98 for the low-income group, \pounds3.47 for the middle-income group, and \pounds3.68 for the high-income group. These are approximately 10.4, 5.0, and 8.8 times the corresponding homogeneous ($M=1$) estimates. \citet{dubois2020well} report the average CV of the UK Soft Drinks Industry Levy to be about \pounds4.94, which is similar in size to our benchmark estimates.

The right column of Figure~\ref{fig:cv_vs_M} decomposes CV by latent type at the benchmark $M=5$, where types are ranked as in \Cref{fig:type_profile}. Welfare losses are highly concentrated: Type~3 in the low-income group ($\pi_3=0.245$), Type~4 in the middle-income group ($\pi_4=0.232$), and Type~5 in the high-income group ($\pi_5=0.289$) have CVs of approximately \pounds23.2, \pounds12.5, and \pounds11.8, respectively, while the remaining type-specific estimates are no larger than about \pounds1.8 in absolute value. These dominant types account for approximately 95\%, 84\%, and 93\% of aggregate CV, respectively, after weighting by their mixture probabilities. The dominant type is also the most prevalent in the low- and high-income groups, but not in the middle-income group. Type-specific CV depends jointly on product preferences, state dependence, substitution possibilities, and the price coefficient used to convert utility changes into monetary units. Overall, allowing for unobserved heterogeneity substantially increases the estimated welfare cost and reveals its concentration among a single type in each income group.

\begin{figure}[h!]
  \centering
  \begin{subfigure}[t]{0.48\textwidth}
    \centering
    \caption{Low income --- CV vs.\ $M$}
    \label{fig:cv_vs_M_low}
    \includegraphics[width=\textwidth]{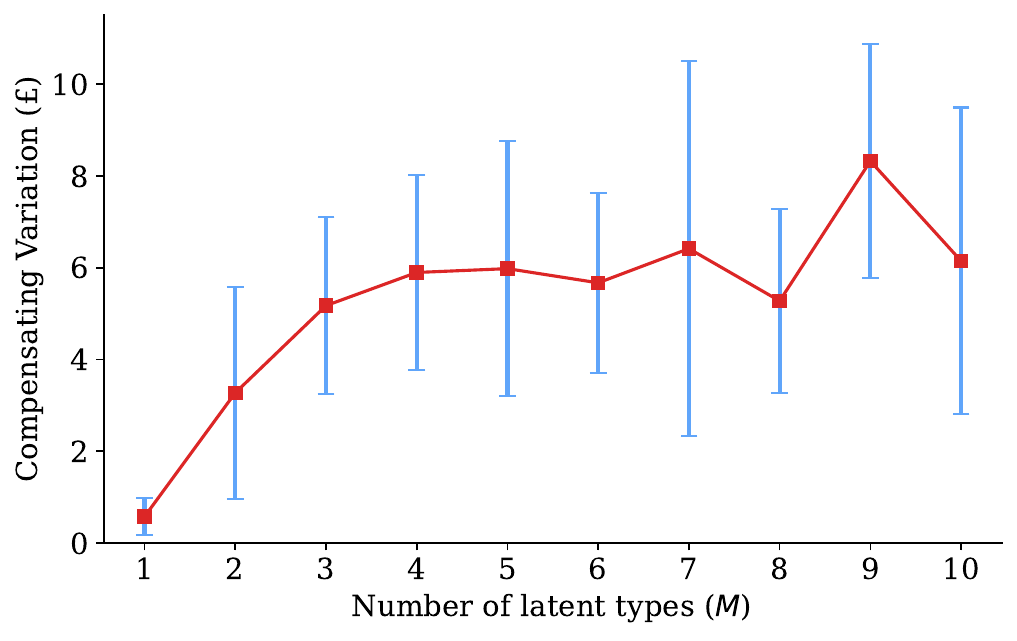}
  \end{subfigure}
  \hfill
  \begin{subfigure}[t]{0.48\textwidth}
    \centering
    \caption{Low income --- CV by type ($M=5$)}
    \label{fig:cv_by_type_low}
    \includegraphics[width=\textwidth]{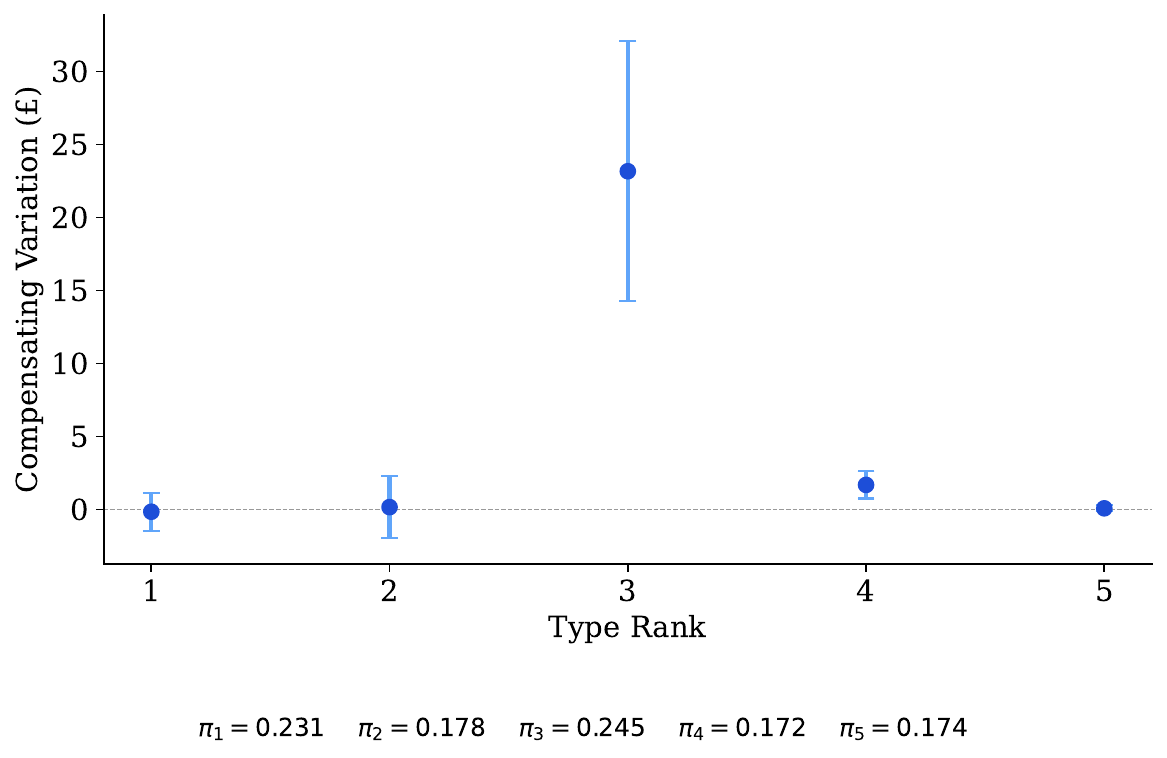}
  \end{subfigure}

  \vspace{0.5em}

  \begin{subfigure}[t]{0.48\textwidth}
    \centering
    \caption{Middle income --- CV vs.\ $M$}
    \label{fig:cv_vs_M_medium}
    \includegraphics[width=\textwidth]{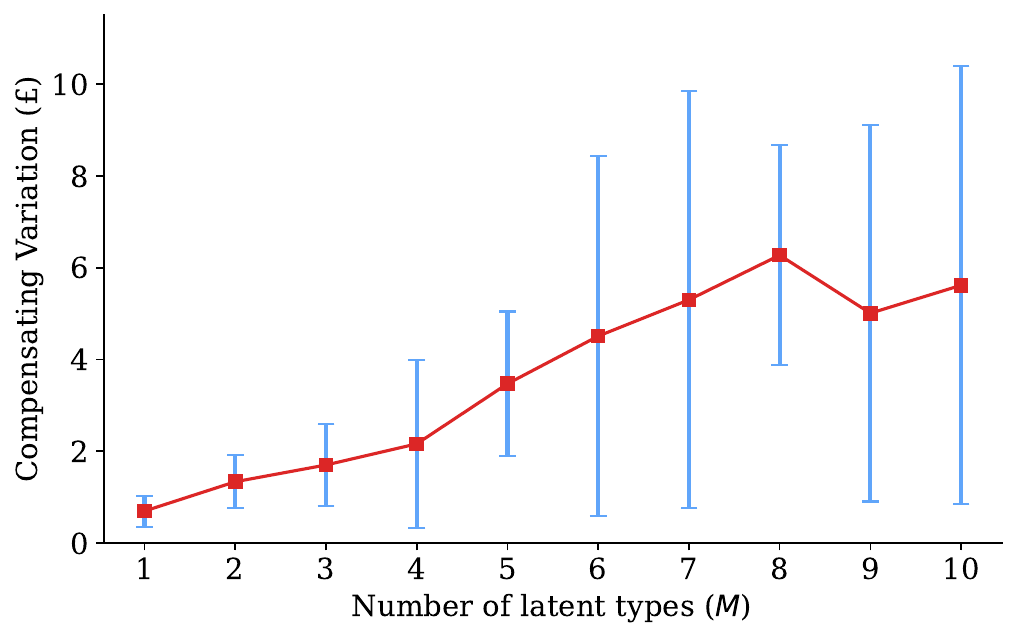}
  \end{subfigure}
  \hfill
  \begin{subfigure}[t]{0.48\textwidth}
    \centering
    \caption{Middle income --- CV by type ($M=5$)}
    \label{fig:cv_by_type_medium}
    \includegraphics[width=\textwidth]{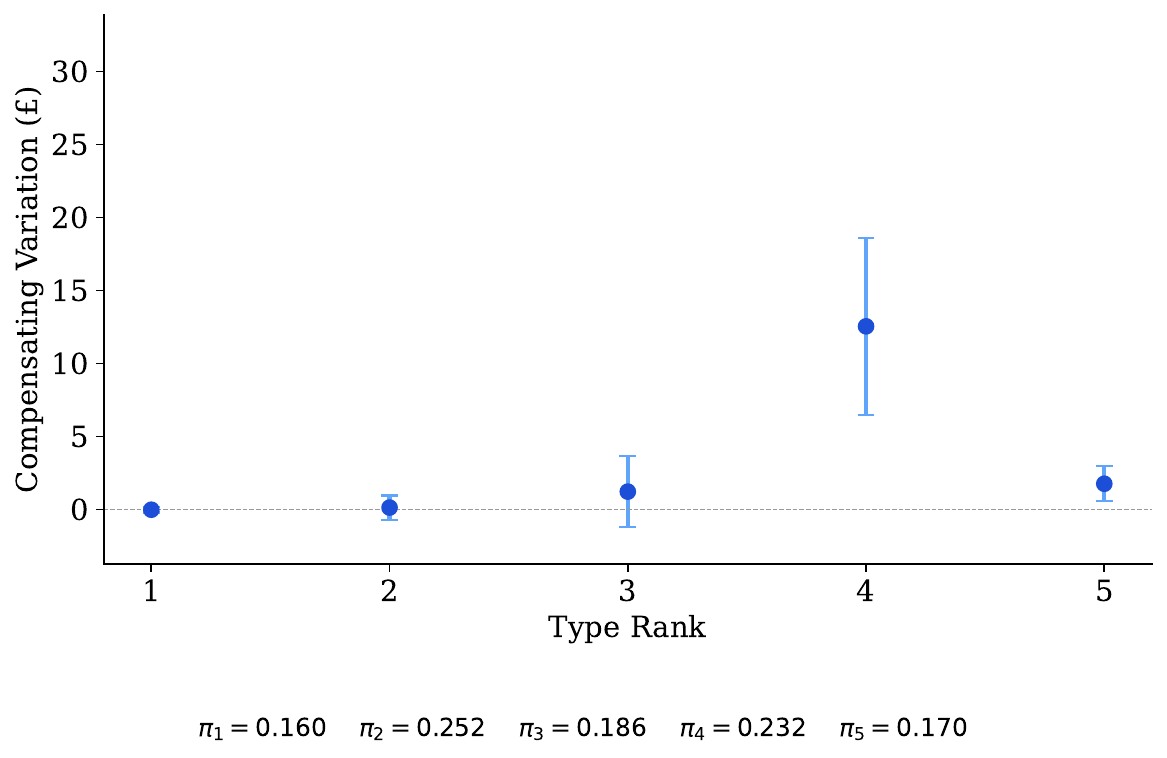}
  \end{subfigure}

  \vspace{0.5em}

  \begin{subfigure}[t]{0.48\textwidth}
    \centering
    \caption{High income --- CV vs.\ $M$}
    \label{fig:cv_vs_M_high}
    \includegraphics[width=\textwidth]{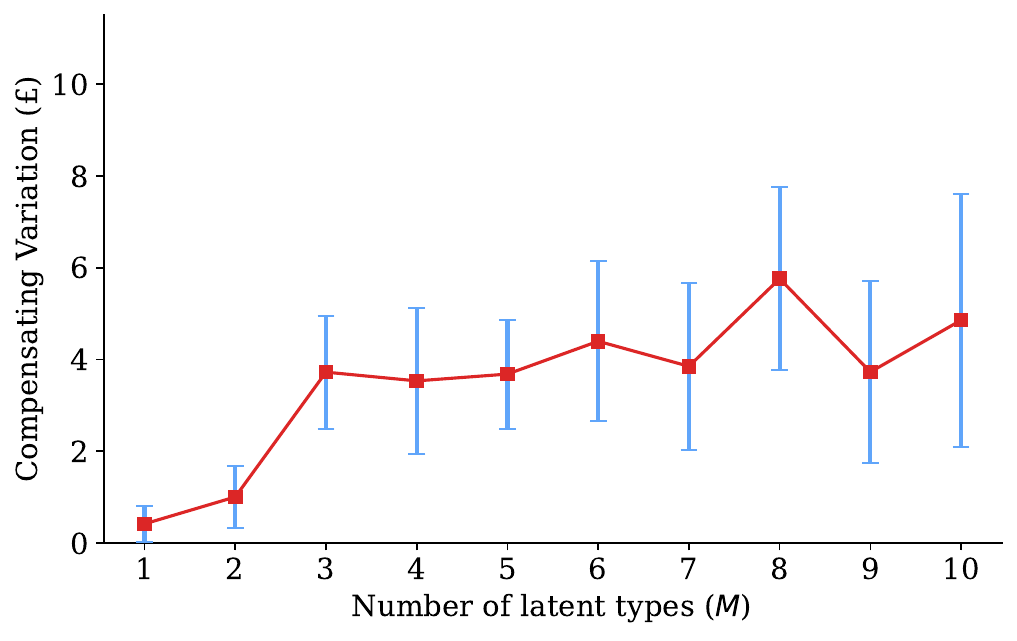}
  \end{subfigure}
  \hfill
  \begin{subfigure}[t]{0.48\textwidth}
    \centering
    \caption{High income --- CV by type ($M=5$)}
    \label{fig:cv_by_type_high}
    \includegraphics[width=\textwidth]{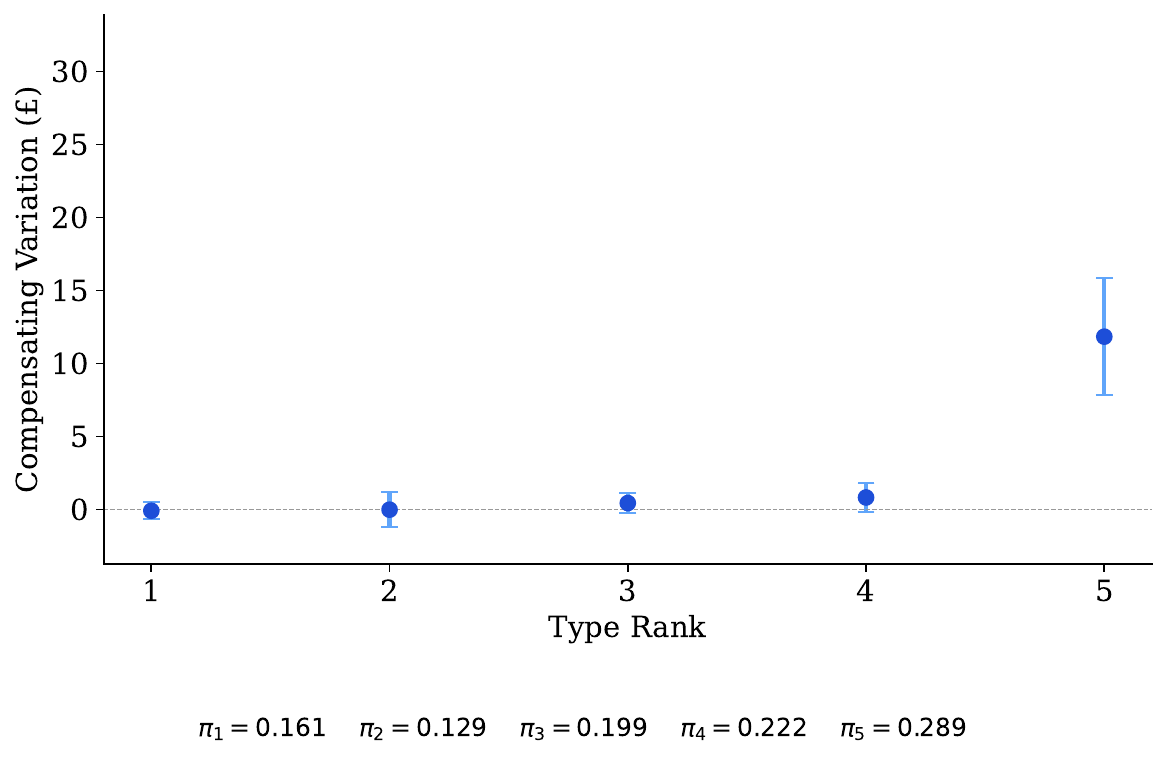}
  \end{subfigure}
  \caption{Compensating variation: by number of latent types~$M$ (left column) and by latent type at $M=5$ (right column), for each income group. In the right column, types are ranked by long-run Coca-Cola Regular own-price elasticity from least to most elastic (same ordering as \Cref{fig:type_profile}); mixture probabilities $\pi_m$ are reported below each bar.}
  \label{fig:cv_vs_M}
\end{figure}

Figures~\ref{fig:cv_children_low}--\ref{fig:cv_children_high} decompose the welfare cost by household size, grouping households into four categories by the number of children (0, 1, 2, and $\geq 3$). In the low-income group, CV rises from about \pounds4.8 for childless households to \pounds6.2 for one-child households, peaks at \pounds10.2 for two-child households, and is \pounds9.4 for households with three or more children. The middle-income group has its highest point estimate among one-child households, at about \pounds4.9, compared with \pounds3.1, \pounds3.5, and \pounds3.2 for households with zero, two, and three or more children, respectively. The high-income group is essentially flat across all categories, with CV between approximately \pounds3.6 and \pounds4.0. The 95\% confidence intervals overlap across categories in all three groups. The low-income point estimates show the clearest differences by household size, with a peak at two children.

\begin{figure}[h!]
  \centering
  \begin{subfigure}[t]{0.45\textwidth}
    \centering
    \caption{Low income}
    \label{fig:cv_children_low}
    \includegraphics[width=\textwidth]{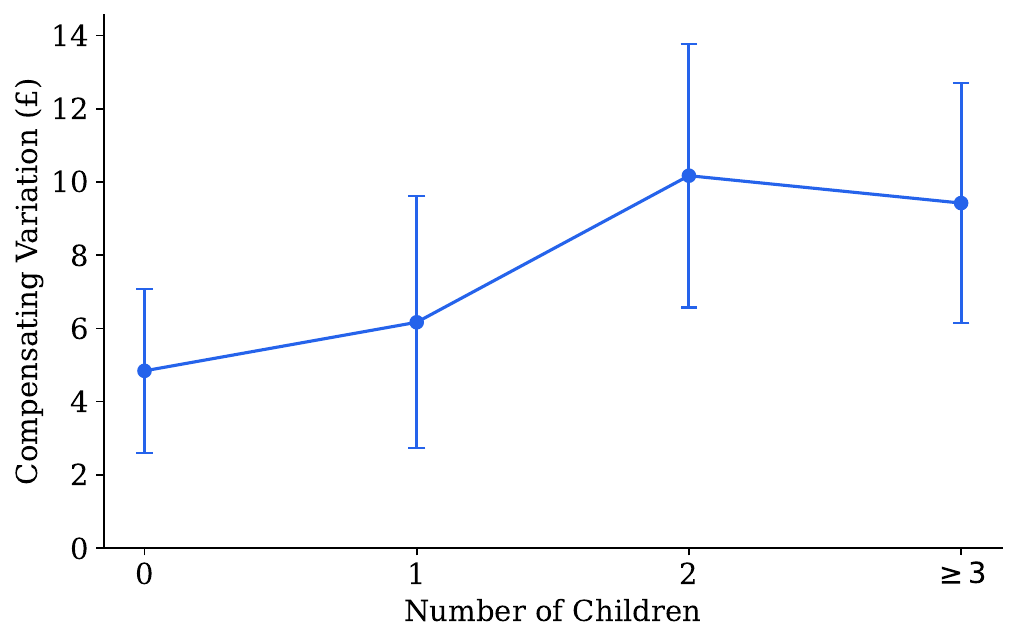}
  \end{subfigure}
  \hfill
  \begin{subfigure}[t]{0.45\textwidth}
    \centering
    \caption{Middle income}
    \label{fig:cv_children_medium}
    \includegraphics[width=\textwidth]{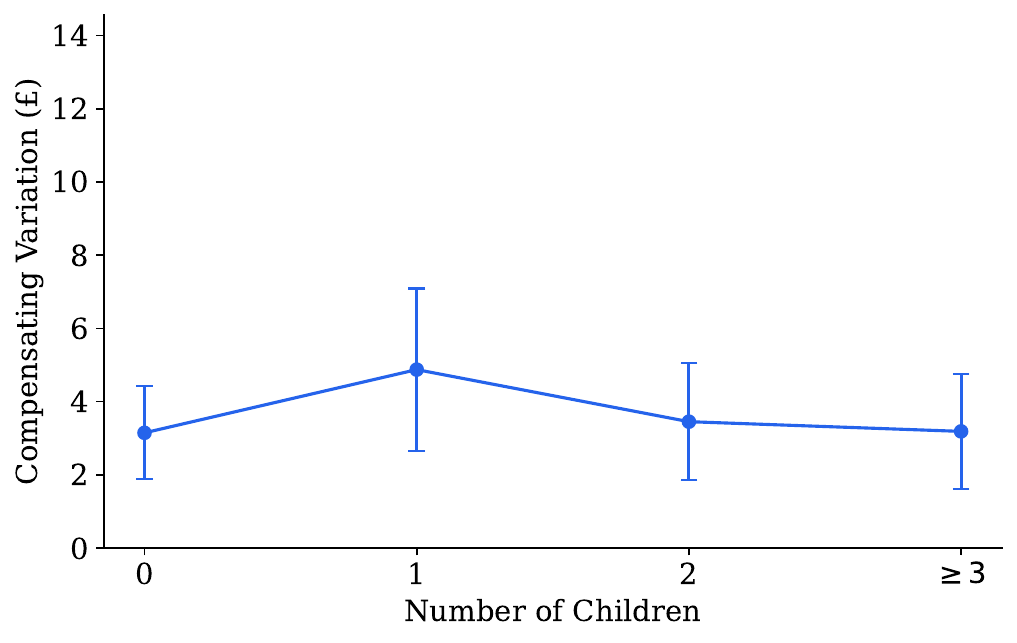}
  \end{subfigure}

  \begin{subfigure}[t]{0.45\textwidth}
    \centering
    \caption{High income}
    \label{fig:cv_children_high}
    \includegraphics[width=\textwidth]{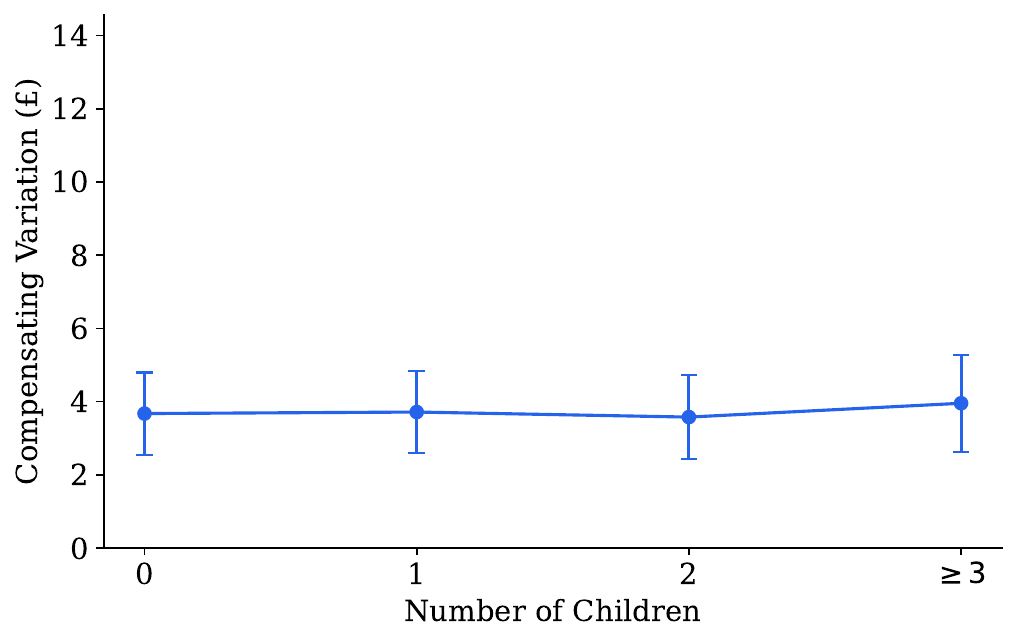}
  \end{subfigure}
  \caption{Compensating variation by number of children}
  \label{fig:cv_children}
\end{figure}

\FloatBarrier

\section{Conclusion} \label{sec: conclusion}

\noindent For dynamic discrete choice models with finite-mixture unobserved heterogeneity, this paper develops the EM-NPL($q$) estimator that employs a truncated inner fixed-point solver, establishes its consistency and asymptotic normality, and characterizes the local convergence of the algorithm used to compute it. For linear-in-parameters models, EM-NPL($q$) is numerically identical to EM-NPL for any $q \geq 1$. Simulations show that GMRES delivers large speed gains, while EPL can improve accuracy in dynamic games at a higher computational cost. The cola-demand application shows that ignoring unobserved heterogeneity substantially understates price elasticities and the welfare cost of a soda tax.

\setlength{\bibsep}{0pt}
\setcitestyle{authoryear,round}
\bibliographystyle{apalike}
\bibliography{reference}

@article{Higgins23joe,
	author = {Ayden Higgins and Koen Jochmans},
	doi = {https://doi.org/10.1016/j.jeconom.2023.04.006},
	issn = {0304-4076},
	journal = {Journal of Econometrics},
	number = {1},
	pages = {105462},
	title = {Identification of mixtures of dynamic discrete choices},
	url = {https://www.sciencedirect.com/science/article/pii/S0304407623001562},
	volume = {237},
	year = {2023}}

@article{fukasawa2025epl,
	author = {Takeshi Fukasawa},
	doi = {https://doi.org/10.1016/j.econlet.2025.112195},
	issn = {0165-1765},
	journal = {Economics Letters},
	pages = {112195},
	title = {{J}acobian-free Efficient Pseudo-Likelihood ({EPL}) algorithm},
	url = {https://www.sciencedirect.com/science/article/pii/S0165176525000321},
	volume = {247},
	year = {2025}}

@article{aguirregabiria2007sequential,
	author = {Aguirregabiria, Victor and Mira, Pedro},
	journal = {Econometrica},
	number = {1},
	pages = {1--53},
	publisher = {Wiley Online Library},
	title = {Sequential estimation of dynamic discrete games},
	volume = {75},
	year = {2007}}

@article{aguirregabiria2002swapping,
	author = {Aguirregabiria, Victor and Mira, Pedro},
	journal = {Econometrica},
	number = {4},
	pages = {1519--1543},
	publisher = {Wiley Online Library},
	title = {Swapping the nested fixed point algorithm: A class of estimators for discrete {M}arkov decision models},
	volume = {70},
	year = {2002}}

@article{kasahara2011sequential,
	author = {Kasahara, Hiroyuki and Shimotsu, Katsumi},
	journal = {Working Paper},
	title = {Sequential estimation of dynamic programming models},
	year = {2011}}

@article{kasahara2009nonparametric,
	author = {Kasahara, Hiroyuki and Shimotsu, Katsumi},
	journal = {Econometrica},
	number = {1},
	pages = {135--175},
	publisher = {Wiley Online Library},
	title = {Nonparametric identification of finite mixture models of dynamic discrete choices},
	volume = {77},
	year = {2009}}

@article{arcidiacono2011conditional,
	author = {Arcidiacono, Peter and Miller, Robert A},
	journal = {Econometrica},
	number = {6},
	pages = {1823--1867},
	publisher = {Wiley Online Library},
	title = {Conditional choice probability estimation of dynamic discrete choice models with unobserved heterogeneity},
	volume = {79},
	year = {2011}}

@article{dearing2025efficient,
	author = {Dearing, Adam and Blevins, Jason R},
	journal = {Review of Economic Studies},
	number = {2},
	pages = {981--1021},
	publisher = {Oxford University Press UK},
	title = {Efficient and convergent sequential pseudo-likelihood estimation of dynamic discrete games},
	volume = {92},
	year = {2025}}

@article{aguirregabiria2023solving,
	author = {Aguirregabiria, Victor and Magesan, Arvind},
	journal = {Working Paper},
	title = {Solving discrete choice dynamic programming models using {E}uler equations},
	year = {2023}}

@article{kasahara2008pseudo,
	author = {Kasahara, Hiroyuki and Shimotsu, Katsumi},
	journal = {Journal of Econometrics},
	number = {1},
	pages = {92--106},
	publisher = {Elsevier},
	title = {Pseudo-likelihood estimation and bootstrap inference for structural discrete {M}arkov decision models},
	volume = {146},
	year = {2008}}

@article{kasahara2012sequential,
	author = {Kasahara, Hiroyuki and Shimotsu, Katsumi},
	journal = {Econometrica},
	number = {5},
	pages = {2303--2319},
	publisher = {Wiley Online Library},
	title = {Sequential estimation of structural models with a fixed point constraint},
	volume = {80},
	year = {2012}}

@article{aguirregabiria2021imposing,
	author = {Aguirregabiria, Victor and Marcoux, Mathieu},
	journal = {Quantitative Economics},
	number = {4},
	pages = {1223--1271},
	publisher = {Wiley Online Library},
	title = {Imposing equilibrium restrictions in the estimation of dynamic discrete games},
	volume = {12},
	year = {2021}}

@article{bugni2021iterated,
	author = {Bugni, Federico A and Bunting, Jackson},
	journal = {The Review of Economic Studies},
	number = {3},
	pages = {1031--1073},
	publisher = {Oxford University Press},
	title = {On the iterated estimation of dynamic discrete choice games},
	volume = {88},
	year = {2021}}

@article{bunting2025faster,
	author = {Bunting, Jackson and Ura, Takuya},
	journal = {Journal of Econometrics},
	pages = {106004},
	publisher = {Elsevier},
	title = {Faster estimation of dynamic discrete choice models using index invertibility},
	volume = {250},
	year = {2025}}

@article{rust1987optimal,
	author = {Rust, John},
	journal = {Econometrica: Journal of the Econometric Society},
	pages = {999--1033},
	publisher = {JSTOR},
	title = {Optimal replacement of {GMC} bus engines: An empirical model of {H}arold {Z}urcher},
	year = {1987}}

@article{hotz1993conditional,
	author = {Hotz, V Joseph and Miller, Robert A},
	journal = {The Review of Economic Studies},
	number = {3},
	pages = {497--529},
	publisher = {Wiley-Blackwell},
	title = {Conditional choice probabilities and the estimation of dynamic models},
	volume = {60},
	year = {1993}}

@book{judd1998numerical,
	author = {Judd, Kenneth L},
	publisher = {MIT press},
	title = {Numerical Methods in Economics},
	year = {1998}}

@book{saad2003iterative,
	author = {Saad, Yousef},
	publisher = {SIAM},
	title = {Iterative methods for sparse linear systems},
	year = {2003}}

@article{su2012constrained,
	author = {Su, Che-Lin and Judd, Kenneth L},
	journal = {Econometrica},
	number = {5},
	pages = {2213--2230},
	publisher = {Wiley Online Library},
	title = {Constrained optimization approaches to estimation of structural models},
	volume = {80},
	year = {2012}}

@article{tauchen1986finite,
	author = {Tauchen, George},
	journal = {Economics letters},
	number = {2},
	pages = {177--181},
	publisher = {Elsevier},
	title = {Finite state {M}arkov-chain approximations to univariate and vector autoregressions},
	volume = {20},
	year = {1986}}

@article{chu1986generalization,
	author = {Chu, King-wah Eric},
	journal = {Numerische Mathematik},
	number = {6},
	pages = {685--691},
	publisher = {Springer-Verlag Berlin/Heidelberg},
	title = {Generalization of the {B}auer-{F}ike theorem},
	volume = {49},
	year = {1986}}

@book{Horn_Johnson_1991,
	author = {Horn, Roger A. and Johnson, Charles R.},
	publisher = {Cambridge University Press},
	title = {Topics in Matrix Analysis},
	year = {1991}}

@article{aguirregabiria2010dynamic,
	author = {Aguirregabiria, Victor and Mira, Pedro},
	doi = {10.1016/j.jeconom.2009.09.007},
	journal = {Journal of Econometrics},
	number = {1},
	pages = {38--67},
	title = {Dynamic Discrete Choice Structural Models: A Survey},
	volume = {156},
	year = {2010}}

@article{dube2012improving,
	author = {Dub{\'e}, Jean-Pierre and Fox, Jeremy T and Su, Che-Lin},
	journal = {Econometrica},
	number = {5},
	pages = {2231--2267},
	publisher = {Wiley Online Library},
	title = {Improving the numerical performance of static and dynamic aggregate discrete choice random coefficients demand estimation},
	volume = {80},
	year = {2012}}

@incollection{heckman1981heterogeneity,
	author = {Heckman, James J},
	booktitle = {Studies in Labor Markets},
	editor = {Rosen, Sherwin},
	pages = {91--139},
	publisher = {University of Chicago Press},
	title = {Heterogeneity and State Dependence},
	year = {1981}}

@article{kasahara2014non,
	author = {Kasahara, Hiroyuki and Shimotsu, Katsumi},
	journal = {Journal of the Royal Statistical Society: Series B},
	number = {1},
	pages = {97--111},
	publisher = {Wiley Online Library},
	title = {Non-Parametric Identification and Estimation of the Number of Components in Multivariate Mixtures},
	volume = {76},
	year = {2014}}

@article{hu2012nonparametric,
	author = {Hu, Yingyao and Shum, Matthew},
	journal = {Journal of Econometrics},
	number = {1},
	pages = {32--44},
	publisher = {Elsevier},
	title = {Nonparametric Identification of Dynamic Models with Unobserved State Variables},
	volume = {171},
	year = {2012}}

@article{adusumilli2025temporal,
	author = {Adusumilli, Karun and Eckardt, Dita},
	doi = {10.1093/restud/rdaf081},
	journal = {Review of Economic Studies},
	note = {Forthcoming},
	title = {Temporal-Difference Estimation of Dynamic Discrete Choice Models},
	year = {2025}}

@article{dempster1977maximum,
	author = {Dempster, Arthur P. and Laird, Nan M. and Rubin, Donald B.},
	journal = {Journal of the Royal Statistical Society: Series B (Methodological)},
	number = {1},
	pages = {1--38},
	title = {Maximum Likelihood from Incomplete Data via the {EM} Algorithm},
	volume = {39},
	year = {1977}}

@article{wu1983convergence,
	author = {Wu, C. F. Jeff},
	journal = {Annals of Statistics},
	number = {1},
	pages = {95--103},
	title = {On the Convergence Properties of the {EM} Algorithm},
	volume = {11},
	year = {1983}}

@article{kolda2009tensor,
	author = {Kolda, Tamara G. and Bader, Brett W.},
	journal = {SIAM Review},
	number = {3},
	pages = {455--500},
	title = {Tensor Decompositions and Applications},
	volume = {51},
	year = {2009}}

@article{wang2015impact,
	author = {Wang, Emily Yucai},
	journal = {The RAND Journal of Economics},
	number = {2},
	pages = {409--441},
	publisher = {Wiley Online Library},
	title = {The impact of soda taxes on consumer welfare: implications of storability and taste heterogeneity},
	volume = {46},
	year = {2015}}

@article{osborne2011consumer,
	author = {Osborne, Matthew},
	journal = {Quantitative Marketing and Economics},
	number = {1},
	pages = {25--70},
	publisher = {Springer},
	title = {Consumer learning, switching costs, and heterogeneity: A structural examination},
	volume = {9},
	year = {2011}}

@article{keane1997career,
	author = {Keane, Michael P and Wolpin, Kenneth I},
	journal = {Journal of political Economy},
	number = {3},
	pages = {473--522},
	publisher = {The University of Chicago Press},
	title = {The career decisions of young men},
	volume = {105},
	year = {1997}}

@article{francesconi2002joint,
	author = {Francesconi, Marco},
	journal = {Journal of labor Economics},
	number = {2},
	pages = {336--380},
	publisher = {The University of Chicago Press},
	title = {A joint dynamic model of fertility and work of married women},
	volume = {20},
	year = {2002}}

@article{gowrisankaran2012dynamics,
	author = {Gowrisankaran, Gautam and Rysman, Marc},
	journal = {Journal of political Economy},
	number = {6},
	pages = {1173--1219},
	publisher = {University of Chicago Press Chicago, IL},
	title = {Dynamics of consumer demand for new durable goods},
	volume = {120},
	year = {2012}}

@article{arcidiacono2011practical,
	author = {Arcidiacono, Peter and Ellickson, Paul B},
	journal = {Annu. Rev. Econ.},
	number = {1},
	pages = {363--394},
	publisher = {Annual Reviews},
	title = {Practical methods for estimation of dynamic discrete choice models},
	volume = {3},
	year = {2011}}

@article{dubois2020well,
	author = {Dubois, Pierre and Griffith, Rachel and O'Connell, Martin},
	journal = {American Economic Review},
	number = {11},
	pages = {3661--3704},
	publisher = {American Economic Association 2014 Broadway, Suite 305, Nashville, TN 37203},
	title = {How well targeted are soda taxes?},
	volume = {110},
	year = {2020}}

\clearpage

\appendix

\section{Proofs} \label{sec:proofs}

\subsection{Supporting Lemmas}

\noindent Throughout the appendix, $C$ denotes a generic finite positive constant whose value may change from one occurrence to the next. Unless otherwise noted, we suppress the type index $m$ and work type by type. We write $z=(\theta,P,\tilde{\theta})$ and $z_0=(\theta_0,P_0,\tilde{\theta}_0)$ with $\tilde{\theta}_0=\theta_0$, and use the shorthand $\Gamma^q(z,Y)\coloneqq\Gamma^q(\theta,Y,P,\tilde{\theta})$, where $Y$ serves as both the initial guess and the center, with $\Gamma^q_z$ and $\Gamma^q_Y$ the Jacobians with respect to $z$ and $Y$. Similarly, we write $Y_{\mathrm{fp}}(z,\tilde{Y})\coloneqq Y_{\mathrm{fp}}(\theta,P,\tilde{\theta},\tilde{Y})$, with $(Y_{\mathrm{fp}})_z$ and $(Y_{\mathrm{fp}})_{\tilde{Y}}$ the Jacobians with respect to $z$ and $\tilde{Y}$. The true value $Y_0$ satisfies $Y_0 = Y_{\mathrm{fp}}(z_0,Y_0)$, and hence $Y_0 = \Gamma^q(z_0,Y_0)$ for all $q$ by \Cref{assumption: same fixed point}. All derivatives of $\Gamma^q$ and $Y_{\mathrm{fp}}$ are evaluated at $(z_0,Y_0)$ unless otherwise noted. Define the exact-solver counterpart of $\Gamma^q$ as $\Gamma^{\mathrm{NPL}}(z,Y)\coloneqq Y_{\mathrm{fp}}(z,Y)$; consequently, $\Gamma^{\mathrm{NPL}}_z=(Y_{\mathrm{fp}})_z$ and $\Gamma^{\mathrm{NPL}}_Y=(Y_{\mathrm{fp}})_{\tilde{Y}}$. For any $q$-dependent quantity $X^q$, we write $X^{\mathrm{NPL}}$ for its exact-solver counterpart, obtained by replacing $\Gamma^q$ with $\Gamma^{\mathrm{NPL}}$; for example, $s_i^{\mathrm{NPL}}$ denotes the score $s_i^q$ with $\Gamma^q$ replaced by $\Gamma^{\mathrm{NPL}}$.

\begin{lemma} \label{lemma: outer product difference bound}
    Let $x_{1},x_{2},y_{1},y_{2}$ be random vectors, then
    \begin{equation*}
        \|E[x_{1}x_{2}^{\top} - y_{1}y_{2}^{\top}]\|_F \le \sqrt{E[\|x_1\|^2]}\sqrt{E[\|x_2-y_2\|^2]} + \sqrt{E[\|y_2\|^2]} \sqrt{E[\|x_1-y_1\|^2]}
    \end{equation*}
\end{lemma}
\begin{proof}
    By Jensen's inequality and the triangle inequality, $\|E[x_1 x_2^{\top} - y_1 y_2^{\top}]\|_F \le E[\|x_1(x_2-y_2)^{\top}\|_F] + E[\|(x_1-y_1)y_2^{\top}\|_F] = E[\|x_1\|\|x_2-y_2\|] + E[\|x_1-y_1\|\|y_2\|]$. Applying the Cauchy--Schwarz inequality to each term gives the result.
\end{proof}

\begin{lemma}[Approximation bound] \label{lemma: derivative approximation bound}
    Under Assumptions~\ref{assump: LS primitives}, \ref{assumption: same fixed point}, \ref{assump: LS regularity}, and \ref{assumption: approximation error upper bound}, at $(z_0,Y_0)$,
    \[
        \|\Gamma_z^q(z_0,Y_0)-(Y_{\mathrm{fp}})_z(z_0,Y_0)\|_F \leq C f(q),
        \qquad
        \|\Gamma_Y^q(z_0,Y_0)-(Y_{\mathrm{fp}})_{\tilde{Y}}(z_0,Y_0)\|_F \leq C f(q).
    \]
\end{lemma}
\begin{proof}
    From \Cref{assumption: approximation error upper bound}, for $(z,Y)$ in a neighborhood of $(z_0,Y_0)$,
    \[
        \|\Gamma^q(z,Y)-Y_{\mathrm{fp}}(z,Y)\| \leq \|Y_{\mathrm{fp}}(z,Y)-Y\| f(q).
    \]
    Here $\Gamma^q$ and $Y_{\mathrm{fp}}$ are continuously differentiable by \Cref{assump: LS smoothness} and \Cref{assump: fp jacobian}, and $\Gamma^q(z_0,Y_0)=Y_{\mathrm{fp}}(z_0,Y_0)=Y_0$ by \Cref{assumption: same fixed point}.

    For the first bound, fix a unit vector $v$ and set $z_h=z_0+hv$. Evaluating the displayed inequality at $(z_h,Y_0)$ and using $Y_{\mathrm{fp}}(z_0,Y_0)=Y_0$,
    \[
        \left\|\frac{(\Gamma^q(z_h,Y_0)-\Gamma^q(z_0,Y_0)) - (Y_{\mathrm{fp}}(z_h,Y_0)-Y_{\mathrm{fp}}(z_0,Y_0))}{h}\right\|
        \leq
        \left\|\frac{Y_{\mathrm{fp}}(z_h,Y_0)-Y_{\mathrm{fp}}(z_0,Y_0)}{h}\right\| f(q).
    \]
    Letting $h\to0$ gives $\|(\Gamma_z^q(z_0,Y_0)-(Y_{\mathrm{fp}})_z(z_0,Y_0))v\| \leq \|(Y_{\mathrm{fp}})_z(z_0,Y_0)v\|f(q)$. Taking the supremum over unit vectors $v$ and using $\|(Y_{\mathrm{fp}})_z(z_0,Y_0)\|\le C$ gives the first bound, with the Frobenius-norm statement following from equivalence of norms in finite dimensions.

    For the second bound, fix a unit vector $u$ and set $Y_h=Y_0+hu$. Evaluating the displayed inequality at $(z_0,Y_h)$ and using $\Gamma^q(z_0,Y_0)=Y_{\mathrm{fp}}(z_0,Y_0)=Y_0$,
    \begin{align*}
        & \left\|\frac{(\Gamma^q(z_0,Y_h)-\Gamma^q(z_0,Y_0)) - (Y_{\mathrm{fp}}(z_0,Y_h)-Y_{\mathrm{fp}}(z_0,Y_0))}{h}\right\| \\
        & \leq
        \left\|\frac{(Y_{\mathrm{fp}}(z_0,Y_h)-Y_{\mathrm{fp}}(z_0,Y_0)) - (Y_h-Y_0)}{h}\right\| f(q).
    \end{align*}
    Letting $h\to0$ gives $\|(\Gamma_Y^q(z_0,Y_0)-(Y_{\mathrm{fp}})_{\tilde{Y}}(z_0,Y_0))u\| \leq \|((Y_{\mathrm{fp}})_{\tilde{Y}}(z_0,Y_0)-I)u\|f(q)$. Taking the supremum over unit vectors $u$ and using $\|(Y_{\mathrm{fp}})_{\tilde{Y}}(z_0,Y_0)\|\le C$ gives the second bound.
\end{proof}

\begin{lemma} \label{lemma: gradient bound}
    Suppose Assumptions~\ref{assump: stationary initial}, \ref{assump: LS primitives}, \ref{assumption: same fixed point}, \ref{assump: LS regularity}, and \ref{assumption: approximation error upper bound} hold. Then, uniformly in $q$,
    \begin{lemmaitems}
        \item \label{lemma: gradient bound1} $\sup_{x,a,j} \| \nabla_{\theta} \log\Lambda_j(\theta_{0}, \Gamma^{q}(\theta_{0}, Y_{0}, P_{0}, {\theta}_{0}),P_{0})(a|x)\| \leq C$.
        \item \label{lemma: gradient bound2} $\sup_{x,a,j} \| \nabla_{\theta} \log\Lambda_j(\theta_{0}, Y_{\mathrm{fp}}(\theta_{0}, P_{0}, {\theta}_{0}, Y_{0}),P_{0})(a|x)\| \leq C$.
        \item \label{lemma: gradient bound5} $\sup_{x,a,j} \|\nabla_{Y} \log\Lambda_j(\theta_{0},Y_{0},P_{0})(a|x)\| \leq C $.
        \item \label{lemma: gradient bound6} $\left\| \nabla_{(\theta, P, \tilde{\theta}, \tilde{Y})} Y_{\mathrm{fp}}(\theta_{0}, P_{0}, {\theta}_{0}, Y_{0})\right\|_{F} \leq C $.
        \item \label{lemma: gradient bound9} $\left\| \nabla_{(\theta, Y, P, \tilde{\theta})} \Gamma^{q}(\theta_{0}, Y_{0}, P_{0}, {\theta}_{0})\right\|_{F} \leq C$.
        \item \label{lemma: gradient bound10} $\sup_{x} \|\nabla_P \log p_*(x;P_0)\| \leq C$.
    \end{lemmaitems}
\end{lemma}
\begin{proof}
At each $(a,x)$, the gradients $\nabla_\theta \log\Lambda_j(a|x)$ and $\nabla_Y \log\Lambda_j(a|x)$ are finite because $\Lambda_j$ is continuously differentiable and $\Lambda_j(a|x)>0$ by \Cref{assump: LS primitives}. Part (iii) then follows since $\mA\times\mX\times\{1,\ldots,J\}$ is finite. Part (iv) is the finite Jacobian of $Y_{\mathrm{fp}}$ from \Cref{assump: fp jacobian}. For part (v), \Cref{assump: LS smoothness} makes $\nabla_{(\theta, Y, P, \tilde{\theta})}\Gamma^q$ finite for each fixed $q$, while \Cref{lemma: derivative approximation bound} and part (iv) bound it for all large $q$; enlarging the constant covers the remaining finitely many $q$, so the bound holds uniformly in $q$. Parts (i) and (ii) follow from the chain rule together with parts (iii), (iv), and (v). For part (vi), $p_*(\cdot;P_0)$ is the unique stationary distribution of the finite-state Markov kernel $f_x^{P_0}$ by \Cref{assump: stationary initial existence}. Uniqueness makes the eigenvalue one of $f_x^{P_0}$ simple, so the implicit function theorem applied to the stationary equation, together with the continuous differentiability of the transition primitives from \Cref{assump: LS smoothness}, makes $p_*(\cdot;P)$ continuously differentiable in $P$ at $P_0$. Moreover, stationarity in \Cref{assump: LS iid} makes $p_*(\cdot;P_0)$ the conditional law of $x_{it}$ given the type, so the full-support condition there gives $p_*(x;P_0)>0$ for every $x\in\mX$. Hence the sensitivity of $\log p_*$ to $P$ is finite at every $x$, and the supremum over the finite set $\mX$ is bounded.
\end{proof}

\begin{lemma}[M-step sensitivity approximation: $\boldsymbol{\vartheta}$] \label{lemma: H approximation}
    Suppose Assumptions~\ref{assump: stationary initial}, \ref{assump: LS primitives}, \ref{assumption: same fixed point}, \ref{assump: LS regularity}, \ref{assumption: approximation error upper bound}, and \ref{assumption: mixture convergence regularity} hold. Then, for each $\boldsymbol{\eta}\in\{\boldsymbol{\theta},\bfY,\bfP,\boldsymbol{\pi}\}$,
    \[
        \|H_{\boldsymbol{\eta}}^q-H_{\boldsymbol{\eta}}^{\mathrm{NPL}}\|_F=O(f(q)).
    \]
\end{lemma}
\begin{proof}
    For each type $m$, let $\xi^m \coloneqq (\vartheta^m,Y^m,P^m,\theta^m)$ and define the type-specific criterion as
    \[
        r_i^{q}(\xi^m)
        \coloneqq \log \bigl( \mL_i^{q}(\xi^m) \bigr)=
        \log p_*^m(x_{i1};P^m) + \sum_{t=1}^{T}\sum_{j=1}^{J}
        \log\Lambda_j\bigl(\vartheta^m,\Gamma^q(\xi^m),P^m\bigr)(a_{jit}|x_{it}),
    \]
    which does not depend on $\boldsymbol{\pi}$ since $\boldsymbol{\pi}$ enters $\ell_E^q$ only through the posterior weight $w_{im}^q$. Let $\vartheta_0^m=\theta_0^m$ for each type $m$, and write $\xi_0^m \coloneqq (\vartheta_0^m,Y_0^m,P_0^m,\theta_0^m)$. Define the score for each type as $s_{\xi^m i}^{q}(\xi^m)\coloneqq\nabla_{\xi^m} r_i^{q}(\xi^m)$.
    
       Define $\rho_i^{q}(\vartheta^m;\boldsymbol{\zeta}) \coloneqq w_{im}^{q}(\boldsymbol{\zeta})\, r_i^{q}(\xi^m)$ so that $\nabla_{\boldsymbol{\vartheta}\boldsymbol{\eta}}^2\ell_E^q(\boldsymbol{\vartheta};\boldsymbol{\zeta}) = \sum_{m=1}^M \bE [\nabla_{\boldsymbol{\vartheta}\boldsymbol{\eta}}^2 \rho_i^{q}(\vartheta^m;\boldsymbol{\zeta})]$.  Observe that $\nabla_{\vartheta^m} w_{im}^q(\boldsymbol{\zeta})=0$ since $w_{im}^q(\boldsymbol{\zeta})$ does not depend on ${\vartheta^m}$. It follows that
    \[
        \nabla_{ \xi^m \vartheta^m}^2
        \rho_i^{q}(\vartheta_0^m;\boldsymbol{\zeta}_0)
        =
 	(\nabla_{\xi^m} w_{im}^q(\boldsymbol{\zeta}_0))s_{\vartheta^m i}^{q}(\xi_0^m) + w_{im}^q(\boldsymbol{\zeta}_0) \nabla_{\xi^m} s_{\vartheta^m i}^{q}(\xi_0^m).
    \]
        For the second term on the right-hand side, we can apply the Bartlett identity because $r_i^q$ is the correctly specified type-$m$ log-density at the fixed point. Consequently, $\bE\!\left[w_{im}^q(\boldsymbol{\zeta}_0) \nabla_{\xi^m} s_{\vartheta^m i}^{q}(\xi_0^m)\right]= -\bE\!\left[w_{im}^q(\boldsymbol{\zeta}_0) s_{\xi^m i}^q(s_{\vartheta^m i}^q)^{\top}\right]$, and hence
    \begin{equation} \label{eq: E-derivative}
            \bE\!\left[\nabla_{ \xi^m \vartheta^m}^2
        \rho_i^{q}(\vartheta_0^m;\boldsymbol{\zeta}_0)\right]
        =
 	\bE\!\left[(\nabla_{\xi^m} w_{im}^q(\boldsymbol{\zeta}_0))s_{\vartheta^m i}^{q}(\xi_0^m) \right] -\bE\!\left[w_{im}^q(\boldsymbol{\zeta}_0) s_{\xi^m i}^q(s_{\vartheta^m i}^q)^{\top}\right].
    \end{equation}
For the $\vartheta^m$ component of $\xi^m$, the first term on the right-hand side of \eqref{eq: E-derivative} vanishes because $\nabla_{\vartheta^m} w_{im}^q=0$, so the M-step Hessian is given by the Bartlett term alone.
The gradient of $w_{im}^q(\boldsymbol{\zeta}_0)$ does depend on $q$. Write $w_{im}^q(\boldsymbol{\zeta})$ as
    \[
        w_{im}^q(\boldsymbol{\zeta}) = \frac{\pi_m\exp\bigl(r_i^q(\theta^m,Y^m,P^m,\theta^m)\bigr)}{\sum_{\ell=1}^M \pi_{\ell} \exp\bigl(r_i^q(\theta^{\ell},Y^{\ell},P^{\ell},\theta^{\ell})\bigr)}.
    \]
    For $\eta\in\{\theta,Y,P\}$ and each type $\ell$, let $s_{\eta i}^{q,\ell} \coloneqq \nabla_\eta r_i^q(\theta^\ell,Y^\ell,P^\ell,\theta^\ell)$, where for $\eta=\theta$ the derivative is total over the first and fourth arguments of $r_i^q$. Then, for each type $\ell$,
    \[
        \nabla_{\eta^\ell} w_{im}^q(\boldsymbol{\zeta}) =  w_{im}^q(\boldsymbol{\zeta}) \bigl[ s_{\eta i}^{q,m} \mathbbm{1}\{\ell=m\} -  w_{i\ell}^q(\boldsymbol{\zeta})s_{\eta i}^{q,\ell} \bigr].
    \]
    Since $\Gamma^q(\xi_0^m) = \Gamma^{\mathrm{NPL}}(\xi_0^m) = Y_0^m$ by \Cref{assumption: same fixed point} and the definition of $\Gamma^{\mathrm{NPL}}$, $w_{im}^q(\boldsymbol{\zeta}_0) = w_{im}^{\mathrm{NPL}}(\boldsymbol{\zeta}_0)$ exactly. From \Cref{lemma: derivative approximation bound}, \Cref{lemma: gradient bound}, and the chain rule, we have
    \begin{equation}\label{eq: s_diff}
        \sup_{i,m}\|s_{\xi^m i}^{q}(\xi_0^m)-s_{\xi^m i}^{\mathrm{NPL}}(\xi_0^m)\|=O(f(q)),
    \end{equation}
    and since $\theta_0^\ell=\vartheta_0^\ell$ for every type $\ell$, the same bound holds with $s_{\eta i}^{q,\ell}$ evaluated at $(\theta_0^\ell,Y_0^\ell,P_0^\ell,\theta_0^\ell)$ in place of $s_{\xi^\ell i}^q(\xi_0^\ell)$. In conjunction with this, we obtain, for $\eta\in\{\theta,Y,P\}$,
    \begin{equation}\label{eq: w_diff}
        \sup_{i,m,\ell}\left\|\nabla _{\eta^\ell} w_{im}^q(\boldsymbol{\zeta}_0) - \nabla _{\eta^\ell} w_{im}^{\mathrm{NPL}}(\boldsymbol{\zeta}_0)\right\| = O(f(q)),
    \end{equation}
    where $\ell$ ranges over both the own type ($\ell=m$) and cross-type ($\ell\neq m$) contributions.

    We analyze the derivative of $\nabla_{\vartheta^m}\rho_i^{q}(\vartheta^m;\boldsymbol{\zeta})$ with respect to $\boldsymbol{\pi}$. Since $w_{im}^q(\boldsymbol{\zeta})$ does not depend on ${\vartheta^m}$ and $r_i^{q}(\xi^m)$ does not depend on $\boldsymbol{\pi}$, we obtain
            \begin{equation} \label{eq: E-derivative_2}
            \bE\!\left[\nabla_{\boldsymbol{\pi} \vartheta^m}^2
        \rho_i^{q}(\vartheta_0^m;\boldsymbol{\zeta}_0)\right]
        =
 	\bE\!\left[(\nabla_{\boldsymbol{\pi}} w_{im}^q(\boldsymbol{\zeta}_0)) s_{\vartheta^m i}^{q}(\xi_0^m) \right].
    \end{equation}
To avoid differentiating under the simplex constraint $\sum_{m=1}^M \pi_m=1$, reparameterize $\boldsymbol{\pi}$ via the multinomial-logit link $\pi_m = \exp(\gamma_m)/\sum_{\ell=1}^M\exp(\gamma_\ell)$, with $\gamma_M\coloneqq 0$ for identification, so that $\boldsymbol{\gamma}=(\gamma_1,\ldots,\gamma_{M-1})\in\bR^{M-1}$ is unconstrained. Substituting into $w_{im}^q(\boldsymbol{\zeta})$ gives
    \[
        w_{im}^q(\boldsymbol{\zeta}) = \frac{\exp\bigl(\gamma_m + r_i^q(\theta^m,Y^m,P^m,\theta^m)\bigr)}{\sum_{\ell=1}^M \exp\bigl(\gamma_\ell + r_i^q(\theta^\ell,Y^\ell,P^\ell,\theta^\ell)\bigr)}.
    \]
By the same computation as for $\nabla_{\eta^\ell} w_{im}^q$, for each $\ell=1,\ldots,M-1$,
    \[
        \nabla_{\gamma_\ell} w_{im}^q(\boldsymbol{\zeta}) = w_{im}^q(\boldsymbol{\zeta}) \bigl[\mathbbm{1}\{\ell=m\}-w_{i\ell}^q(\boldsymbol{\zeta})\bigr].
    \]
    Because $w_{im}^q(\boldsymbol{\zeta}_0) = w_{im}^{\mathrm{NPL}}(\boldsymbol{\zeta}_0)$ exactly for every type as shown above, $\nabla_{\gamma_\ell} w_{im}^q(\boldsymbol{\zeta}_0) = \nabla_{\gamma_\ell} w_{im}^{\mathrm{NPL}}(\boldsymbol{\zeta}_0)$ exactly. Therefore, $\nabla_{\boldsymbol{\pi}} w_{im}^q(\boldsymbol{\zeta}_0)=\nabla_{\boldsymbol{\pi}} w_{im}^{\mathrm{NPL}}(\boldsymbol{\zeta}_0)$ holds.

    \eqref{eq: E-derivative}, \eqref{eq: s_diff}, \eqref{eq: w_diff}, \eqref{eq: E-derivative_2}, \Cref{lemma: outer product difference bound}, and boundedness of posterior weights and scores imply, for $\boldsymbol{\eta}\in\{\boldsymbol{\vartheta},\boldsymbol{\theta},\bfY,\bfP,\boldsymbol{\pi}\}$,
    \[
        \left\|
        \nabla_{\boldsymbol{\vartheta}\boldsymbol{\eta}}^2\ell_E^q(\boldsymbol{\vartheta}_0;\boldsymbol{\zeta}_0)
        -
        \nabla_{\boldsymbol{\vartheta}\boldsymbol{\eta}}^2\ell_E^{\mathrm{NPL}}(\boldsymbol{\vartheta}_0;\boldsymbol{\zeta}_0)
        \right\|_F
        =
        O(f(q)).
    \]
Since $\nabla_{\boldsymbol{\vartheta}\boldsymbol{\vartheta}}^2\ell_E^{\mathrm{NPL}}$ and $\nabla_{\boldsymbol{\vartheta}\boldsymbol{\vartheta}}^2\ell_E^q$ are nonsingular from \Cref{assumption: mixture convergence regularity}, the matrix inverse perturbation bound gives
$\left\|
        \{\nabla_{\boldsymbol{\vartheta}\boldsymbol{\vartheta}}^2\ell_E^q (\boldsymbol{\vartheta}_0;\boldsymbol{\zeta}_0) \}^{-1}
        -
        \{\nabla_{\boldsymbol{\vartheta}\boldsymbol{\vartheta}}^2\ell_E^{\mathrm{NPL}} (\boldsymbol{\vartheta}_0;\boldsymbol{\zeta}_0) \}^{-1}
        \right\|_F
        =
        O(f(q))$, and the stated bound follows from the definition of $H_{\boldsymbol{\eta}}^q$ and $H_{\boldsymbol{\eta}}^{\mathrm{NPL}}$.
\end{proof}

\begin{lemma}[M-step sensitivity approximation: $\boldsymbol{\pi}$] \label{lemma: Pi approximation}
    Suppose Assumptions~\ref{assump: LS primitives}, \ref{assumption: same fixed point}, \ref{assump: LS regularity}, and \ref{assumption: approximation error upper bound} hold. Then, $\Pi_{\boldsymbol{\pi}}^q=\Pi_{\boldsymbol{\pi}}^{\mathrm{NPL}}$ and, for $\boldsymbol{\eta}\in\{\boldsymbol{\theta},\bfY,\bfP\}$,
    \[
        \|\Pi_{\boldsymbol{\eta}}^q-\Pi_{\boldsymbol{\eta}}^{\mathrm{NPL}}\|_F=O(f(q)).
    \]
\end{lemma}
\begin{proof}
    Recall $\Pi^q(\boldsymbol{\zeta}) = \bE[w_{i1}^q(\boldsymbol{\zeta}),\ldots,w_{iM}^q(\boldsymbol{\zeta})]^{\top}$. For $\Pi_{\boldsymbol{\pi}}^q$, since $\boldsymbol{\pi}_0 \in \mathrm{int}(\Delta^{M-1})$ from \Cref{assump: LS compactness}, the multinomial-logit link used in the proof of \Cref{lemma: H approximation} is continuously differentiable and invertible in a neighborhood of $\boldsymbol{\pi}_0$, with inverse Jacobian $\partial\boldsymbol{\gamma}/\partial\boldsymbol{\pi}$ that does not depend on $q$. That proof shows $\nabla_{\gamma_\ell}w_{im}^q(\boldsymbol{\zeta}_0)=\nabla_{\gamma_\ell}w_{im}^{\mathrm{NPL}}(\boldsymbol{\zeta}_0)$ exactly for every $(m,\ell)$. Multiplying by the same $q$-independent Jacobian gives $\nabla_{\pi_\ell}w_{im}^q(\boldsymbol{\zeta}_0)=\nabla_{\pi_\ell}w_{im}^{\mathrm{NPL}}(\boldsymbol{\zeta}_0)$, and hence $\Pi_{\boldsymbol{\pi}}^q=\Pi_{\boldsymbol{\pi}}^{\mathrm{NPL}}$.
    
    For $\Pi_{\boldsymbol{\eta}}^q$, recall that $\Pi_{\boldsymbol{\eta}}^q$ has $(m,\ell)$ block $\bE[\nabla_{\eta^\ell}w_{im}^q(\boldsymbol{\zeta}_0)]$. For $\eta\in\{\theta,Y,P\}$, \eqref{eq: w_diff} in the proof of \Cref{lemma: H approximation} gives $\sup_{i,m,\ell}\|\nabla_{\eta^\ell}w_{im}^q(\boldsymbol{\zeta}_0)-\nabla_{\eta^\ell}w_{im}^{\mathrm{NPL}}(\boldsymbol{\zeta}_0)\|=O(f(q))$. The stated result follows immediately.
\end{proof}

\begin{lemma}[Mixture convergence matrix difference] \label{lemma: M difference bound}
    Suppose Assumptions~\ref{assump: stationary initial}, \ref{assump: LS primitives}, \ref{assumption: same fixed point}, \ref{assump: LS regularity}, \ref{assumption: approximation error upper bound}, and \ref{assumption: mixture convergence regularity} hold. Then
    \begin{equation*}
        \| R^{q} - R^{\mathrm{NPL}} \|_{F} = O(f(q)).
    \end{equation*}
\end{lemma}
\begin{proof}
    The first row of $R^q - R^{\mathrm{NPL}}$ is $O(f(q))$ since $\Pi_{\boldsymbol{\pi}}^q = \Pi_{\boldsymbol{\pi}}^{\mathrm{NPL}}$ and $\|\Pi_{\boldsymbol{\eta}}^q - \Pi_{\boldsymbol{\eta}}^{\mathrm{NPL}}\|_F = O(f(q))$ for $\boldsymbol{\eta}\in\{\boldsymbol{\theta},\bfY,\bfP\}$ from \Cref{lemma: Pi approximation}. The second row of $R^q - R^{\mathrm{NPL}}$ is $O(f(q))$ since $\|H^q_{\boldsymbol{\eta}} - H^{\mathrm{NPL}}_{\boldsymbol{\eta}}\|_F = O(f(q))$ for $\boldsymbol{\eta}\in\{\boldsymbol{\theta},\bfY,\bfP,\boldsymbol{\pi}\}$ by \Cref{lemma: H approximation}. We bound the $(3,1)$ block of $R^q - R^{\mathrm{NPL}}$, $\boldsymbol{\Gamma}^q_{\boldsymbol{\theta}} H^q_{{\boldsymbol{\pi}}} - \boldsymbol{\Gamma}^{\mathrm{NPL}}_{\boldsymbol{\theta}} H^{\mathrm{NPL}}_{{\boldsymbol{\pi}}}$. Note that $\|\boldsymbol{\Gamma}^q_{\boldsymbol{\eta}} - \boldsymbol{\Gamma}^{\mathrm{NPL}}_{\boldsymbol{\eta}}\|_F = O(f(q))$ for $\boldsymbol{\eta} \in \{\boldsymbol{\theta}, \bfY, \bfP, \tilde{\boldsymbol{\theta}}\}$ from \Cref{lemma: derivative approximation bound}. Using the identity $\boldsymbol{\Gamma}^q_{\boldsymbol{\theta}}H^q_{{\boldsymbol{\pi}}} - \boldsymbol{\Gamma}^{\mathrm{NPL}}_{\boldsymbol{\theta}} H^{\mathrm{NPL}}_{{\boldsymbol{\pi}}} = \boldsymbol{\Gamma}^q_{\boldsymbol{\theta}} (H^q_{{\boldsymbol{\pi}}}-H^{\mathrm{NPL}}_{{\boldsymbol{\pi}}}) + (\boldsymbol{\Gamma}^q_{\boldsymbol{\theta}} -\boldsymbol{\Gamma}^{\mathrm{NPL}}_{\boldsymbol{\theta}})H^{\mathrm{NPL}}_{{\boldsymbol{\pi}}}$, the triangle inequality, and submultiplicativity of the Frobenius norm, we obtain $\|\boldsymbol{\Gamma}^q_{\boldsymbol{\theta}}H^q_{{\boldsymbol{\pi}}} - \boldsymbol{\Gamma}^{\mathrm{NPL}}_{\boldsymbol{\theta}}H^{\mathrm{NPL}}_{{\boldsymbol{\pi}}}\|_F = O(f(q))$. Applying a similar argument to the other blocks, we conclude that the third row of $R^q - R^{\mathrm{NPL}}$ is $O(f(q))$.
    
    We bound the $(4,1)$ block of $R^q - R^{\mathrm{NPL}}$, $\boldsymbol{\Lambda}^q_{\boldsymbol{\theta}} H^q_{{\boldsymbol{\pi}}} - \boldsymbol{\Lambda}^{\mathrm{NPL}}_{\boldsymbol{\theta}} H^{\mathrm{NPL}}_{{\boldsymbol{\pi}}}$. The chain rule gives $\boldsymbol{\Lambda}^q_{\boldsymbol{\eta}} - \boldsymbol{\Lambda}^{\mathrm{NPL}}_{\boldsymbol{\eta}} = \boldsymbol{\Lambda}_{\bfY}\bigl(\boldsymbol{\Gamma}^q_{\boldsymbol{\eta}} - \boldsymbol{\Gamma}^{\mathrm{NPL}}_{\boldsymbol{\eta}}\bigr)$ for $\boldsymbol{\eta}\in\{\boldsymbol{\theta},\bfY,\bfP,\tilde{\boldsymbol{\theta}}\}$, so $\|\boldsymbol{\Lambda}^q_{\boldsymbol{\eta}} - \boldsymbol{\Lambda}^{\mathrm{NPL}}_{\boldsymbol{\eta}}\|_F \leq \|\boldsymbol{\Lambda}_{\bfY}\|_F \|\boldsymbol{\Gamma}^q_{\boldsymbol{\eta}} - \boldsymbol{\Gamma}^{\mathrm{NPL}}_{\boldsymbol{\eta}}\|_F = O(f(q))$. Hence, using the same argument as above gives $\|\boldsymbol{\Lambda}^q_{\boldsymbol{\theta}} H^q_{{\boldsymbol{\pi}}} - \boldsymbol{\Lambda}^{\mathrm{NPL}}_{\boldsymbol{\theta}} H^{\mathrm{NPL}}_{{\boldsymbol{\pi}}}\|_F = O(f(q))$, and, similarly, the fourth row of $R^q - R^{\mathrm{NPL}}$ is $O(f(q))$. Taking the Frobenius norm of the full matrix difference gives $\|R^q - R^{\mathrm{NPL}}\|_F = O(f(q))$.
\end{proof}

\subsection{Proofs in \Cref{sec: estimator} and \Cref{sec: theory}}

\begin{proof}[\textbf{Proof of \Cref{lemma: algorithm limit is fixed point}}]
By hypothesis, $(\boldsymbol{\theta}^{(k)}, \boldsymbol{\pi}^{(k)}, \bfY^{(k)}, \bfP^{(k)}) \to (\hat{\boldsymbol{\theta}}, \hat{\boldsymbol{\pi}}, \hat{\bfY}, \hat{\bfP})$. Then, $\hat{\bfY} = \boldsymbol{\Gamma}^q(\hat{\boldsymbol{\theta}}, \hat{\bfY}, \hat{\bfP}, \hat{\boldsymbol{\theta}})$ and $\hat{\bfP} = \boldsymbol{\Lambda}(\hat{\boldsymbol{\theta}}, \hat{\bfY}, \hat{\bfP})$ hold since $\Gamma^q$ and $\Lambda$ are continuous from \Cref{assump: LS smoothness}.

For the first-order condition, we apply the classical EM stationarity argument of \citet{dempster1977maximum} and \citet{wu1983convergence} to the type-by-type M-step with nuisance parameters $(\bfY,\bfP)$. Let $w_{im}(\vartheta,\varpi) \coloneqq \varpi_m \mL_i^{q}(\vartheta^m;\hat{Y}^m,\hat{P}^m,\hat{\theta}^m) / \sum_{m'} \varpi_{m'} \mL_i^{q}(\vartheta^{m'};\hat{Y}^{m'},\hat{P}^{m'},\hat{\theta}^{m'})$ denote the posterior weight $(\vartheta,\varpi)$ induces. By \Cref{assump: LS positive ccp}, $w_{im}^{(k)} \to \hat{w}_{im} \coloneqq w_{im}(\hat{\boldsymbol{\theta}},\hat{\boldsymbol{\pi}})$.  For $(\vartheta,\varpi) \in \Theta^M \times \Delta^{M-1}$, define
    \[
        Q_N(\vartheta,\varpi) \coloneqq \frac{1}{N}\sum_i \sum_m \hat{w}_{im} \log\bigl(\varpi_m\, \mL_i^{q}(\vartheta^m;\hat{Y}^m,\hat{P}^m,\hat{\theta}^m)\bigr),
    \]
and define $H_N(\vartheta,\varpi) \coloneqq \frac{1}{N}\sum_i \sum_m \hat{w}_{im} \log w_{im}(\vartheta,\varpi)$. Because $\log\bigl(\sum_m \varpi_m \mL_i^{q}(\vartheta^m;\ldots)\bigr) = \sum_m \hat{w}_{im}\bigl[\log(\varpi_m \mL_i^{q}(\vartheta^m;\ldots)) - \log w_{im}(\vartheta,\varpi)\bigr]$ for any weights summing to one, we have the identity
    \begin{equation} \label{eq: Q minus H}
        \cQ_N^q(\vartheta,\varpi;\hat{\bfY},\hat{\bfP},\hat{\boldsymbol{\theta}}) = Q_N(\vartheta,\varpi) - H_N(\vartheta,\varpi), \qquad \forall (\vartheta,\varpi).
    \end{equation}
    We show that $(\hat{\boldsymbol{\theta}},\hat{\boldsymbol{\pi}})$ maximizes each of $Q_N$ and $H_N$ over $\Theta^M \times \Delta^{M-1}$.
For each type $m$, the M-step inequality \eqref{eq: M-step_obj} holds at every $k$; taking $k \to \infty$ and using continuity of $\mL_i^{q}$ and $w_{im}^{(k)} \to \hat{w}_{im}$ gives
    \[
        \sum_i \hat{w}_{im} \log \mL_i^{q}(\hat{\theta}^m;\hat{Y}^m,\hat{P}^m,\hat{\theta}^m) \geq \sum_i \hat{w}_{im} \log \mL_i^{q}(\vartheta^m;\hat{Y}^m,\hat{P}^m,\hat{\theta}^m), \qquad \forall \vartheta^m \in \Theta.
    \]
    Since $\varpi \mapsto \sum_m (N^{-1}\sum_i \hat{w}_{im}) \log \varpi_m$ is strictly concave on the simplex, the vector with components $N^{-1}\sum_i \hat{w}_{im}$ is its unique maximizer. Further, $\pi_m^{(k)} = N^{-1}\sum_i w_{im}^{(k)} \to N^{-1}\sum_i \hat{w}_{im} = \hat{\pi}_m$. Consequently,
    \begin{equation} \label{eq: Q_N argmax}
        (\hat{\boldsymbol{\theta}}, \hat{\boldsymbol{\pi}}) = \argmax_{(\vartheta,\varpi) \in \Theta^M \times \Delta^{M-1}} Q_N(\vartheta,\varpi).
    \end{equation}
Because $w_{im}(\hat{\boldsymbol{\theta}},\hat{\boldsymbol{\pi}}) = \hat{w}_{im}$, for any $(\vartheta,\varpi)$, Jensen's inequality gives
    \[
        H_N(\hat{\boldsymbol{\theta}},\hat{\boldsymbol{\pi}}) - H_N(\vartheta,\varpi) \geq 0.
    \]

Substituting $\varpi_M = 1 - \sum_{m=1}^{M-1}\varpi_m$ to eliminate the simplex constraint, $Q_N$ and $H_N$ are continuously differentiable in $(\boldsymbol{\vartheta}, \varpi_1, \ldots, \varpi_{M-1})$ from \Cref{assump: LS smoothness}. Since $(\hat{\boldsymbol{\theta}},\hat{\boldsymbol{\pi}})$ is an interior point of $\Theta^M \times \Delta^{M-1}$ by hypothesis, the first-order conditions give $\nabla Q_N(\hat{\boldsymbol{\theta}},\hat{\boldsymbol{\pi}}) = \nabla H_N(\hat{\boldsymbol{\theta}},\hat{\boldsymbol{\pi}}) = 0$ in this parametrization. By \eqref{eq: Q minus H}, $\nabla_{(\vartheta,\varpi)} \cQ_N^q(\hat{\boldsymbol{\theta}},\hat{\boldsymbol{\pi}};\hat{\bfY},\hat{\bfP},\hat{\boldsymbol{\theta}})= 0$.
\end{proof}

\begin{proof}[\textbf{Proof of \Cref{theorem: truncation invariance}}]
Under Assumptions \ref{assumption: theta separability} and \ref{assumption: same fixed point}, two facts hold. First, $\theta$-separability implies $Y_{\mathrm{fp}}(\theta, P, \tilde{\theta}, \tilde{Y})$ does not depend on $\theta$; we write $Y_{\mathrm{fp}}(P, \tilde{\theta}, \tilde{Y})$ for the fixed point. Second, since $G$ does not depend on the current $\theta$, neither do $\Gamma^q$ and $Y_{\mathrm{fp}}$, and by \Cref{assumption: same fixed point}, $Y = \Gamma^q(Y, P, \tilde{\theta})$ holds if and only if $Y = Y_{\mathrm{fp}}(P, \tilde{\theta}, Y)$.

We verify that all three conditions for $\mathcal{E}_N^q$ and $\mathcal{E}_N$ are equivalent. The $P$-condition does not involve $\Gamma^q$ or $Y_{\mathrm{fp}}$ and is therefore equivalent. From the second fact, the $Y$-conditions $Y = \Gamma^q(Y, P, \tilde{\theta})$ and $Y = Y_{\mathrm{fp}}(P, \tilde{\theta}, Y)$ are equivalent. For the argmax condition, consider any $(\boldsymbol{\theta}, \boldsymbol{\pi}, \bfY, \bfP)$ satisfying the $Y$-condition. Since neither $\Gamma^q$ nor $Y_{\mathrm{fp}}$ depends on the current $\theta$, for all $\vartheta$ and $q$ we have $\Gamma^q(\vartheta, Y, P, \tilde{\theta}) = \Gamma^q(Y, P, \tilde{\theta}) = Y$ and $Y_{\mathrm{fp}}(\vartheta, P, \tilde{\theta}, Y) = Y_{\mathrm{fp}}(P, \tilde{\theta}, Y) = Y$. Hence $\cQ_N^q$ and its exact-solver counterpart coincide as functions of $(\boldsymbol{\vartheta}, \boldsymbol{\varpi})$, and the argmax conditions are equivalent. Since all three conditions are equivalent, $\mathcal{E}_N^q = \mathcal{E}_N$ for all $q \geq 1$, and since the two pseudo-likelihoods coincide on this set, the estimators are numerically identical.
\end{proof}

\begin{proof}[\textbf{Proof of \Cref{theorem: LS mixture asymptotic}}]
    We prove consistency and asymptotic normality in sequence.

    \medskip
    \textit{Part~I: Consistency.}
    We extend the proof of Proposition~2 of \citet{aguirregabiria2007sequential} (hereafter AM07) to the mixture setting. AM07's proof proceeds in five steps: (1) $(\theta_0, P_0)$ uniquely maximizes the population pseudo-likelihood in the population fixed-point set; (2) every sample fixed point belongs to an arbitrarily small neighborhood of some population fixed point; (3) the sample NPL operator converges uniformly to the population NPL operator in a neighborhood of $P_0$; (4) a sample fixed point near $(\theta_0, P_0)$ exists with probability approaching 1, using the nonsingular Jacobian condition; and (5) the fixed point maximizing the sample pseudo-likelihood converges to $(\theta_0, P_0)$.

    The NPL operator in AM07 maps $P$ to $P$ with $\theta$ determined as a by-product. In our setting, $\boldsymbol{\theta}$ enters $\phi_q(\bfY, \bfP,\tilde{\boldsymbol{\theta}})$ through the $\tilde{\boldsymbol{\theta}}$ argument and cannot be eliminated from the state. We therefore define $\Phi_N$ as a mapping from $(\bfY, \bfP, \boldsymbol{\theta})$ to itself. Specifically, let $\boldsymbol{\alpha}_N(\bfY, \bfP, \boldsymbol{\theta}) \coloneqq \argmax_{\boldsymbol{\alpha}} \cQ_N^q(\boldsymbol{\alpha};\, \bfY, \bfP, \boldsymbol{\theta})$ denote the sample M-step and $\boldsymbol{\theta}_{\boldsymbol{\alpha}_N}$ its $\boldsymbol{\theta}$-component. Define
    \begin{equation*}
        \Phi_N(\bfY, \bfP, \boldsymbol{\theta}) \coloneqq \Bigl(\phi_q(\bfY, \bfP, \boldsymbol{\theta}),\; \boldsymbol{\theta}_{\boldsymbol{\alpha}_N(\bfY, \bfP, \boldsymbol{\theta})}\Bigr),
    \end{equation*}
    with population counterpart $\Phi_0$ defined analogously using $\cQ^q$. The EM-NPL($q$) fixed-point equation is $(\bfY, \bfP, \boldsymbol{\theta}) = \Phi_N(\bfY, \bfP, \boldsymbol{\theta})$. AM07's five steps extend to this higher-dimensional operator by the same arguments.

    Each step of AM07 extends as follows. Step~1 follows from \Cref{assump: nplq identification1}. Step~2 follows from uniform convergence of $\cQ_N^q$ to $\cQ^q$, which holds by Assumptions \ref{assump: LS smoothness} and \ref{assump: LS iid}. Step~3 holds because $\boldsymbol{\alpha}_N(\bfY, \bfP, \boldsymbol{\theta})$ converges uniformly to $\boldsymbol{\alpha}_0(\bfY, \bfP, \boldsymbol{\theta})$ by uniform convergence of $\cQ_N^q$ and the local uniqueness of the maximizer (\Cref{assump: LS local equiv}), so $\Phi_N \to_p \Phi_0$ uniformly in a neighborhood of $(\bfY_0, \bfP_0, \boldsymbol{\theta}_0)$. Step~4 follows from \Cref{assump: LS isolated fp}, which ensures the Jacobian of $\Phi_0 - I$ is nonsingular at $(\bfY_0, \bfP_0, \boldsymbol{\theta}_0)$, so the implicit function theorem guarantees a sample fixed point of $\Phi_N$ in a neighborhood of $(\bfY_0, \bfP_0, \boldsymbol{\theta}_0)$ for all large $N$. For Step~5, by Step~4 a sample fixed point of $\Phi_N$ near $(\bfY_0, \bfP_0, \boldsymbol{\theta}_0)$ exists with probability approaching~1. By Step~1 and uniform convergence of $\cQ_N^q$, this fixed point has a strictly higher sample pseudo-likelihood than any fixed point bounded away from $(\bfY_0, \bfP_0, \boldsymbol{\theta}_0)$ with probability approaching~1. Therefore the EM-NPL($q$) estimator, defined as the fixed point of $\Phi_N$ maximizing $\cQ_N^q$, converges in probability to $(\boldsymbol{\theta}_0, \boldsymbol{\pi}_0, \bfY_0, \bfP_0)$.

    \medskip
    \textit{Part~II: Asymptotic normality.} Let $(\hat{\boldsymbol{\theta}}, \hat{\boldsymbol{\pi}}, \hat{\bfY}, \hat{\bfP})$ denote the EM-NPL($q$) estimator, writing $\hat{\boldsymbol{\theta}} = \hat{\boldsymbol{\theta}}_{\mathrm{NPL}(q)}$, $\hat{\boldsymbol{\pi}} = \hat{\boldsymbol{\pi}}_{\mathrm{NPL}(q)}$, and $\hat{\boldsymbol{\alpha}} \coloneqq (\hat{\boldsymbol{\theta}}, \hat{\boldsymbol{\pi}})$, and suppressing the $\mathrm{NPL}(q)$ subscript for the remainder of this proof. By \Cref{def: em npl estimator}, the EM-NPL($q$) estimator satisfies
    \begin{align}
 &      \frac{1}{N} \sum_{i=1}^N \nabla_{\boldsymbol{\alpha}} \ell_i^q(\hat{\boldsymbol{\alpha}};\, \hat{\bfY}, \hat{\bfP}, \hat{\boldsymbol{\theta}}) = 0, \label{eq:foc} \\
 &    \hat{\bfY} = \boldsymbol{\Gamma}^q(\hat{\boldsymbol{\theta}}, \hat{\bfY}, \hat{\bfP}, \hat{{\boldsymbol{\theta}}}), \quad \hat{\bfP} = \boldsymbol{\Lambda}(\hat{\boldsymbol{\theta}}, \hat{\bfY}, \hat{\bfP}). \label{eq:fpc}
    \end{align}
    Define $\delta_{\boldsymbol{\alpha}} \coloneqq\sqrt{N}(\hat{\boldsymbol{\alpha}} - \boldsymbol{\alpha}_0)$, $\delta_{\boldsymbol{\theta}} \coloneqq\sqrt{N}(\hat{\boldsymbol{\theta}} - \boldsymbol{\theta}_0)$, $\delta_{\bfY} \coloneqq\sqrt{N}(\hat{\bfY} - \bfY_0)$, and $\delta_{\bfP} \coloneqq\sqrt{N}(\hat{\bfP} - \bfP_0)$. Note that  $\delta_{\boldsymbol{\theta}} = (\mathbf{I}_{M d_\theta },\, 0)\, \delta_{\boldsymbol{\alpha}}$.

    Expanding \eqref{eq:foc} around $(\boldsymbol{\alpha}_0, \bfY_0, \bfP_0, \boldsymbol{\theta}_0)$ and using the consistency of $(\hat{\boldsymbol{\alpha}}, \hat{\bfY}, \hat{\bfP}, \hat{\boldsymbol{\theta}})$ gives
    \begin{equation} \label{eq: MVE}
    \begin{split}
        \frac{1}{\sqrt{N}} \sum_{i=1}^N s_i^q &= \Omega^q_{\boldsymbol{\alpha}\boldsymbol{\alpha}}\, \delta_{\boldsymbol{\alpha}} + \Omega^q_{\boldsymbol{\alpha}\tilde{\boldsymbol{\theta}}}\, \delta_{\boldsymbol{\theta}} + \Omega^q_{\boldsymbol{\alpha} \bfY}\, \delta_{\bfY} + \Omega^q_{\boldsymbol{\alpha} \bfP}\, \delta_{\bfP} + o_p(1),
    \end{split}
    \end{equation}
    where $s_i^q \coloneqq\nabla_{\boldsymbol{\alpha}} \ell_i^q(\boldsymbol{\alpha}_0;\, \bfY_0, \bfP_0, \boldsymbol{\theta}_0)$ is the score at the true parameter value. Expanding \eqref{eq:fpc} around $(\boldsymbol{\theta}_0,\bfY_0, \bfP_0)$  and rearranging yields
    \begin{equation}\label{eq: Jacobian}
        \begin{pmatrix} I - \boldsymbol{\Gamma}^q_{\boldsymbol{Y}} & -\boldsymbol{\Gamma}^q_{\boldsymbol{P}} \\ -\boldsymbol{\Lambda}_{\boldsymbol{Y}} & I - \boldsymbol{\Lambda}_{\boldsymbol{P}} \end{pmatrix}
        \begin{pmatrix} \delta_{\bfY} \\ \delta_{\bfP} \end{pmatrix}
        =
        \begin{pmatrix} \boldsymbol{\Gamma}^q_{\boldsymbol{\theta}} + \boldsymbol{\Gamma}^q_{\tilde{\boldsymbol{\theta}}} \\ \boldsymbol{\Lambda}_{\boldsymbol{\theta}} \end{pmatrix} \delta_{\boldsymbol{\theta}} + o_p(1).
    \end{equation}
Differentiating $\bfY^q(\boldsymbol{\theta}) = \boldsymbol{\Gamma}^q(\boldsymbol{\theta}, \bfY^q(\boldsymbol{\theta}), \bfP^q(\boldsymbol{\theta}), \boldsymbol{\theta})$ and $\bfP^q(\boldsymbol{\theta}) = \boldsymbol{\Lambda}(\boldsymbol{\theta}, \bfY^q(\boldsymbol{\theta}), \bfP^q(\boldsymbol{\theta}))$ with respect to $\boldsymbol{\theta}$ at the true fixed point and rearranging yields
    \begin{equation} \label{eq: implicit function}
        \begin{pmatrix} I - \boldsymbol{\Gamma}^q_{\boldsymbol{Y}} & -\boldsymbol{\Gamma}^q_{\boldsymbol{P}} \\ -\boldsymbol{\Lambda}_{\boldsymbol{Y}} & I - \boldsymbol{\Lambda}_{\boldsymbol{P}} \end{pmatrix}
        \begin{pmatrix} \bfY^q_{\boldsymbol{\theta}} \\ \bfP^q_{\boldsymbol{\theta}} \end{pmatrix}
        =
        \begin{pmatrix} \boldsymbol{\Gamma}^q_{\boldsymbol{\theta}} + \boldsymbol{\Gamma}^q_{\tilde{\boldsymbol{\theta}}} \\ \boldsymbol{\Lambda}_{\boldsymbol{\theta}} \end{pmatrix},
    \end{equation}
    where the Jacobian matrices are block-diagonal with the $m$-th diagonal block given by the corresponding Jacobian matrix evaluated at $(\theta_0^m, Y_0^m, P_0^m, \theta_0^m)$. It follows from \eqref{eq: Jacobian}, \eqref{eq: implicit function}, and \Cref{assump: implicit function theorem} that
    \begin{equation*}
        \begin{pmatrix} \delta_{\bfY} \\ \delta_{\bfP} \end{pmatrix}
        =
        \begin{pmatrix} \bfY^q_{\boldsymbol{\theta}} \\ \bfP^q_{\boldsymbol{\theta}} \end{pmatrix}
         \delta_{\boldsymbol{\theta}} + o_p(1).
    \end{equation*}
    Substituting this into \eqref{eq: MVE} and using $\delta_{\boldsymbol{\theta}} = (\mathbf{I}_{M d_\theta},\, 0)\, \delta_{\boldsymbol{\alpha}}$ gives
    \begin{equation*}
    \begin{split}
        \frac{1}{\sqrt{N}} \sum_{i=1}^N s_i^q &= \Bigl[\Omega^q_{\boldsymbol{\alpha}\boldsymbol{\alpha}} + \bigl(\Omega^q_{\boldsymbol{\alpha}\tilde{\boldsymbol{\theta}}} + \Omega^q_{\boldsymbol{\alpha} \bfY}\, \bfY^q_{\boldsymbol{\theta}} + \Omega^q_{\boldsymbol{\alpha} \bfP}\, \bfP^q_{\boldsymbol{\theta}}\bigr) (\mathbf{I}_{M d_\theta},\, 0)\Bigr]\, \delta_{\boldsymbol{\alpha}} + o_p(1) \\
        &= \bigl(\Omega^q_{\boldsymbol{\alpha}\boldsymbol{\alpha}} + S^q\bigr)\, \delta_{\boldsymbol{\alpha}} + o_p(1).
    \end{split}
    \end{equation*}
    Rearranging gives $\sqrt{N}(\hat{\boldsymbol{\alpha}} - \boldsymbol{\alpha}_0) = (\Omega^q_{\boldsymbol{\alpha}\boldsymbol{\alpha}} + S^q)^{-1} \, N^{-1/2} \sum_{i=1}^N s_i^q + o_p(1)$.
\Cref{assumption: same fixed point} ensures the $q$-step likelihood is correctly specified at the truth, so the Bartlett identity and the positivity and smoothness conditions in \Cref{assump: LS primitives} give $\bE[s_i^q] = 0$ and $\mathrm{Var}(s_i^q) = \Omega^q_{\boldsymbol{\alpha}\boldsymbol{\alpha}}$, and the central limit theorem yields $N^{-1/2} \sum_{i=1}^N s_i^q \xrightarrow{d} \mathcal{N}(0,\, \Omega^q_{\boldsymbol{\alpha}\boldsymbol{\alpha}})$. The asymptotic distribution of $\sqrt{N}(\hat{\boldsymbol{\alpha}} - \boldsymbol{\alpha}_0)$ then follows immediately.
\end{proof}

\begin{proof}[\textbf{Proof of \Cref{proposition: variance approximation}}]
    Let $\Omega^{\mathrm{NPL}}_{\boldsymbol{\alpha}\boldsymbol{\alpha}}$, $\Omega^{\mathrm{NPL}}_{\boldsymbol{\alpha}\bfY}$, $\Omega^{\mathrm{NPL}}_{\boldsymbol{\alpha}\bfP}$, $\Omega^{\mathrm{NPL}}_{\boldsymbol{\alpha}\tilde{\boldsymbol{\theta}}}$, and $S^{\mathrm{NPL}}$ denote the counterparts of $\Omega^q_{\boldsymbol{\alpha}\boldsymbol{\alpha}}$, $\Omega^q_{\boldsymbol{\alpha}\bfY}$, $\Omega^q_{\boldsymbol{\alpha}\bfP}$, $\Omega^q_{\boldsymbol{\alpha}\tilde{\boldsymbol{\theta}}}$, and $S^q$ with $\Gamma^q$ replaced by $\Gamma^{\mathrm{NPL}}$. Write $\Sigma^q = A(q)\, B(q)\, A(q)^{\top}$ and $\Sigma_{\mathrm{NPL}} = A_{\mathrm{NPL}}\, B_{\mathrm{NPL}}\, A_{\mathrm{NPL}}^{\top}$, where $A(q) \coloneqq (\Omega^q_{\boldsymbol{\alpha}\boldsymbol{\alpha}} + S^q)^{-1}$, $B(q) \coloneqq \Omega^q_{\boldsymbol{\alpha}\boldsymbol{\alpha}}$, $A_{\mathrm{NPL}} \coloneqq (\Omega^{\mathrm{NPL}}_{\boldsymbol{\alpha}\boldsymbol{\alpha}} + S^{\mathrm{NPL}})^{-1}$, and $B_{\mathrm{NPL}} \coloneqq \Omega^{\mathrm{NPL}}_{\boldsymbol{\alpha}\boldsymbol{\alpha}}$.

    We establish the two preliminary bounds
    \begin{align*}
        \text{(a)}& \quad \|B(q) - B_{\mathrm{NPL}}\|_F = O(f(q)), \\
        \text{(b)}& \quad \|A(q) - A_{\mathrm{NPL}}\|_F = O(f(q)).
    \end{align*}

    \textit{Part~(a).} Since $\Gamma^q(z_0, Y_0) = \Gamma^{\mathrm{NPL}}(z_0, Y_0) = Y_0$ by \Cref{assumption: same fixed point}, the Bartlett identity gives $\Omega^q_{\boldsymbol{\alpha}\boldsymbol{\alpha}} = \bE[s_i^q (s_i^q)^{\top}]$ and $\Omega^{\mathrm{NPL}}_{\boldsymbol{\alpha}\boldsymbol{\alpha}} = \bE[s_i^{\mathrm{NPL}} (s_i^{\mathrm{NPL}})^{\top}]$, where $s_i^q \coloneqq\nabla_{\boldsymbol{\alpha}} \ell_i^q(\boldsymbol{\alpha}_0;\, \bfY_0, \bfP_0, \boldsymbol{\theta}_0)$ is the score at the true parameter value. By \Cref{lemma: derivative approximation bound}, \Cref{lemma: gradient bound}, and the chain rule, $\sup_i \|s_i^q - s_i^{\mathrm{NPL}}\| = O(f(q))$. Boundedness of the scores and \Cref{lemma: outer product difference bound} then give $\|B(q) - B_{\mathrm{NPL}}\|_F = O(f(q))$.

    \textit{Part~(b).} It suffices to show $\|(\Omega^q_{\boldsymbol{\alpha}\boldsymbol{\alpha}} + S^q) - (\Omega^{\mathrm{NPL}}_{\boldsymbol{\alpha}\boldsymbol{\alpha}} + S^{\mathrm{NPL}})\|_F = O(f(q))$. Part~(a) gives $\|\Omega^q_{\boldsymbol{\alpha}\boldsymbol{\alpha}} - \Omega^{\mathrm{NPL}}_{\boldsymbol{\alpha}\boldsymbol{\alpha}}\|_F = O(f(q))$. For the cross-derivative matrices, the same score-covariance argument with \Cref{lemma: outer product difference bound} gives
    \[
        \|\Omega^q_{\boldsymbol{\alpha}\tilde{\boldsymbol{\theta}}} - \Omega^{\mathrm{NPL}}_{\boldsymbol{\alpha}\tilde{\boldsymbol{\theta}}}\|_F + \|\Omega^q_{\boldsymbol{\alpha}\bfY} - \Omega^{\mathrm{NPL}}_{\boldsymbol{\alpha}\bfY}\|_F + \|\Omega^q_{\boldsymbol{\alpha}\bfP} - \Omega^{\mathrm{NPL}}_{\boldsymbol{\alpha}\bfP}\|_F = O(f(q)).
    \]
    For the implicit function derivatives, $(\bfY^q_{\boldsymbol{\theta}}, \bfP^q_{\boldsymbol{\theta}})^{\top}$ solves \eqref{eq: implicit function}, and $(\bfY^{\mathrm{NPL}}_{\boldsymbol{\theta}}, \bfP^{\mathrm{NPL}}_{\boldsymbol{\theta}})^{\top}$ solves the counterpart of \eqref{eq: implicit function} with $\Gamma^q$ replaced by $\Gamma^{\mathrm{NPL}}$. By \Cref{lemma: derivative approximation bound}, the coefficient matrix and the right-hand side of \eqref{eq: implicit function} differ from their NPL counterparts by $O(f(q))$. Since both coefficient matrices are invertible by \Cref{assump: implicit function theorem}, the matrix inverse perturbation bound gives $\|\bfY^q_{\boldsymbol{\theta}} - \bfY^{\mathrm{NPL}}_{\boldsymbol{\theta}}\|_F + \|\bfP^q_{\boldsymbol{\theta}} - \bfP^{\mathrm{NPL}}_{\boldsymbol{\theta}}\|_F = O(f(q))$. Combining these bounds gives $\|S^q - S^{\mathrm{NPL}}\|_F = O(f(q))$, hence $\|(\Omega^q_{\boldsymbol{\alpha}\boldsymbol{\alpha}} + S^q) - (\Omega^{\mathrm{NPL}}_{\boldsymbol{\alpha}\boldsymbol{\alpha}} + S^{\mathrm{NPL}})\|_F = O(f(q))$. Since $\Omega^{\mathrm{NPL}}_{\boldsymbol{\alpha}\boldsymbol{\alpha}} + S^{\mathrm{NPL}}$ is invertible by \Cref{assump: LS info matrix}, the matrix inverse perturbation bound gives $\|A(q) - A_{\mathrm{NPL}}\|_F = O(f(q))$, where $\|A(q)\|_F$ is uniformly bounded for sufficiently large $q$.

    Finally, decompose $\Sigma^q-\Sigma_{\mathrm{NPL}}$ as
    \begin{align*}
       & A(q)\,B(q)\,A(q)^{\top} - A_{\mathrm{NPL}}\,B_{\mathrm{NPL}}\,A_{\mathrm{NPL}}^{\top}\\
       & = (A(q)-A_{\mathrm{NPL}})\,B(q)\,A(q)^{\top}
        + A_{\mathrm{NPL}}\,(B(q)-B_{\mathrm{NPL}})\,A(q)^{\top}
        + A_{\mathrm{NPL}}\,B_{\mathrm{NPL}}\,(A(q)-A_{\mathrm{NPL}})^{\top}.
    \end{align*}
    Applying sub-multiplicativity of the Frobenius norm to each term and using the triangle inequality, parts~(a) and~(b), and the uniform boundedness of $A(q)$, $B(q)$, $A_{\mathrm{NPL}}$, $B_{\mathrm{NPL}}$ (which follows from \Cref{assump: LS primitives} and \Cref{assump: LS regularity}) gives $\|\Sigma^q - \Sigma_{\mathrm{NPL}}\|_F = O(f(q))$.
\end{proof}

\begin{proof}[\textbf{Proof of \Cref{proposition: mixture convergence matrix}}]
    We linearize the EM-NPL($q$) update around the EM-NPL($q$) estimator $\hat{\boldsymbol{\zeta}} = (\hat{\boldsymbol{\theta}}, \hat{\bfY}, \hat{\bfP},\hat{\boldsymbol{\pi}})$. The state at iteration $k-1$ is $\boldsymbol{\zeta}^{(k-1)}=(\boldsymbol{\theta}^{(k-1)}, \bfY^{(k-1)}, \bfP^{(k-1)}, \boldsymbol{\pi}^{(k-1)})$.

The E-step computes posterior weights $w_{im}^{(k)}$ from $\boldsymbol{\zeta}^{(k-1)}$; the M-step then updates $\pi_m^{(k)} = N^{-1} \sum_i w_{im}^{(k)}$. Write the stacked M-step map for $\boldsymbol{\pi}$ as $\Pi_N(\boldsymbol{\zeta})$. Then $\boldsymbol{\pi}^{(k)} = \Pi_N(\boldsymbol{\zeta}^{(k-1)})$ and $\hat{\boldsymbol{\pi}} = \Pi_N(\hat{\boldsymbol{\zeta}})$. Expanding $\Pi_N(\boldsymbol{\zeta}^{(k-1)})$ around $\hat{\boldsymbol{\zeta}}$ gives $\boldsymbol{\pi}^{(k)} - \hat{\boldsymbol{\pi}} = \Pi_{\boldsymbol{\zeta}}\Delta\boldsymbol{\zeta}^{(k-1)} + O_p(N^{-1/2}\|\delta^{(k-1)}\| + \|\delta^{(k-1)}\|^2)$, which is the first row of $R^q$. Here, the remainder order follows because the sample Jacobian of $\Pi_N(\boldsymbol{\zeta})$ at $\hat{\boldsymbol{\zeta}}$ differs from $\Pi_{\boldsymbol{\zeta}}$ by $O_p(N^{-1/2})$ from the $\sqrt{N}$-consistency of the EM-NPL($q$) estimator, smoothness of $\Pi_N(\boldsymbol{\zeta})$, and the $\sqrt{N}$-consistency of the sample average to its population mean. Henceforth, we write $r_N$ for a generic (conformably dimensioned) term of order $O_p(N^{-1/2}\|\delta^{(k-1)}\| + \|\delta^{(k-1)}\|^2)$, whose value may change from line to line.

Let $\ell_N^q(\boldsymbol{\vartheta};\boldsymbol{\zeta})$ denote the sample counterpart of the population M-step criterion $\ell_E^q(\boldsymbol{\vartheta};\boldsymbol{\zeta})$. Then $\boldsymbol{\theta}^{(k)}$ and $\hat{\boldsymbol{\theta}}$ satisfy the first-order conditions $\nabla_{\boldsymbol{\theta}}\ell_N^q(\boldsymbol{\theta}^{(k)};\boldsymbol{\zeta}^{(k-1)})=0$ and $\nabla_{\boldsymbol{\theta}}\ell_N^q(\hat{\boldsymbol{\theta}};\hat{\boldsymbol{\zeta}})=0$. Expanding the first-order condition $\nabla_{\boldsymbol{\theta}}\ell_N^q(\boldsymbol{\theta}^{(k)};\boldsymbol{\zeta}^{(k-1)})=0$ around $(\hat{\boldsymbol{\theta}},\hat{\boldsymbol{\zeta}})$, inverting the sample Hessian $\nabla_{\boldsymbol{\theta}\boldsymbol{\theta}}^2\ell_N^q$, which is nonsingular with probability approaching one by \Cref{assumption: mixture convergence regularity}(iii), and using the same sampling-plus-smoothness argument as for the M-step map for $\boldsymbol{\pi}$ gives the second row of $R^q$,
    \begin{equation*}
        \Delta\boldsymbol{\theta}^{(k)} = H_{\boldsymbol{\pi}}^q\, \Delta\boldsymbol{\pi}^{(k-1)} + H_{\boldsymbol{\theta}}^q\, \Delta\boldsymbol{\theta}^{(k-1)} + H_{\bfY}^q\, \Delta\bfY^{(k-1)} + H_{\bfP}^q\, \Delta\bfP^{(k-1)} + r_N.
    \end{equation*}

Now we analyze the $(\bfY, \bfP)$-update. The updates are $\bfY^{(k)} = \boldsymbol{\Gamma}^q(\boldsymbol{\theta}^{(k)}, \bfY^{(k-1)}, \bfP^{(k-1)}, \boldsymbol{\theta}^{(k-1)})$ and $\bfP^{(k)} = \boldsymbol{\Lambda}^q(\boldsymbol{\theta}^{(k)}, \bfY^{(k-1)}, \bfP^{(k-1)}, \boldsymbol{\theta}^{(k-1)})$. Expanding around the EM-NPL($q$) estimator, by the same argument as above, we obtain
    \begin{equation*}
        \begin{pmatrix} \Delta\bfY^{(k)} \\ \Delta\bfP^{(k)} \end{pmatrix}
        =
        \begin{pmatrix} \boldsymbol{\Gamma}_{\boldsymbol{\theta}}^q \\ \boldsymbol{\Lambda}^{q}_{\boldsymbol{\theta}} \end{pmatrix} \Delta\boldsymbol{\theta}^{(k)}
        +
        \begin{pmatrix} \boldsymbol{\Gamma}_{\tilde{\boldsymbol{\theta}}}^q \\ \boldsymbol{\Lambda}^{q}_{\boldsymbol{\tilde{\theta}}} \end{pmatrix} \Delta\boldsymbol{\theta}^{(k-1)}
        +
        \begin{pmatrix} \boldsymbol{\Gamma}_{\boldsymbol{Y}}^q & \boldsymbol{\Gamma}_{\boldsymbol{P}}^q \\ \boldsymbol{\Lambda}^{q}_{\boldsymbol{Y}} & \boldsymbol{\Lambda}^{q}_{\boldsymbol{P}} \end{pmatrix}
        \begin{pmatrix} \Delta\bfY^{(k-1)} \\ \Delta\bfP^{(k-1)} \end{pmatrix}
        + r_N.
    \end{equation*}
Substituting $\Delta\boldsymbol{\theta}^{(k)}$ from the M-step and collecting terms yields the third and fourth row of $R^q$.
 \end{proof}

\begin{proof}[\textbf{Proof of \Cref{theorem: npl(q) local convergence}}]
    By \Cref{lemma: M difference bound}, $\|R^q - R^{\mathrm{NPL}}\|_F = O(f(q)) \to 0$. The spectral radius is continuous in finite dimensions, so $\rho(R^q)\to\rho(R^{\mathrm{NPL}})$. Since $\rho(R^{\mathrm{NPL}})<1$, there exists $\bar q$ such that $\rho(R^q)<1$ for all $q\geq \bar q$.

    Fix any such $q$. By \Cref{proposition: mixture convergence matrix}, $\delta^{(k)} = R^q\delta^{(k-1)} + r_N$ with $\rho(R^q)<1$. Because $\rho(R^q)<1$, there exists an equivalent vector norm $\|\cdot\|_*$ such that the induced matrix norm satisfies $\|R^q\|_*<1$. Shrinking the neighborhood if necessary gives a contraction in this norm with probability approaching one. Hence the EM-NPL($q$) sequence converges locally for every sufficiently large $q$.
\end{proof}

\begin{proof}[\textbf{Proof of \Cref{theorem: number of iterations bound}}]
    By \Cref{lemma: M difference bound}, $\|R^{q} - R^{\mathrm{NPL}}\|_{F} = O(f(q)) \to 0$. Recall that $\nu$ is the size of the largest Jordan block of $R^{\mathrm{NPL}}$. The generalized Bauer--Fike theorem \citep[Theorems~3 and 3A]{chu1986generalization} bounds the spectral variation induced by any sufficiently small perturbation of $R^{\mathrm{NPL}}$. Applying it to the perturbation $R^q - R^{\mathrm{NPL}}$ implies the spectral perturbation bound
    \[
        \operatorname{dist}_H\{\mathrm{spec}(R^q),\mathrm{spec}(R^{\mathrm{NPL}})\}
        =
        O(\|R^q-R^{\mathrm{NPL}}\|_F^{1/\nu})
        =
        O(f(q)^{1/\nu}),
    \]
    where $\operatorname{dist}_H$ denotes Hausdorff distance between finite sets. Since $|\rho(A)-\rho(B)|\le\operatorname{dist}_H\{\mathrm{spec}(A),\mathrm{spec}(B)\}$ for any matrices $A,B$ from the definition of Hausdorff distance, there exists a sequence $g(q)=O(f(q)^{1/\nu})$ such that
    \[
        |\rho(R^q)-\rho(R^{\mathrm{NPL}})|\le g(q).
    \]

    Under \Cref{assumption: theta separability}, \Cref{theorem: truncation invariance} implies that the EM-NPL($q$) and EM-NPL algorithms share the same fixed point. Fix $c>0$ such that $\rho(R^{\mathrm{NPL}})+c<1$, and set $r \coloneqq \rho(R^{\mathrm{NPL}}) + c/2$, so $\rho(R^{\mathrm{NPL}}) < r < \rho(R^{\mathrm{NPL}})+c$. Since $g(q)\to0$ and both $c/2$ and $1-\rho(R^{\mathrm{NPL}})-c$ are strictly positive, there exists $\bar{q}$ such that $g(q) < \min\{c/2,\, 1-\rho(R^{\mathrm{NPL}})-c\}$ for all $q \ge \bar{q}$. For such $q$ we then have $\rho(R^q) \le \rho(R^{\mathrm{NPL}}) + g(q) < r \le \rho(R^q)+c$ and $\rho(R^q)+c < 1$; in particular, the circle $\Gamma \coloneqq \{z \in \bC : |z|=r\}$ encloses $\mathrm{spec}(R^q)$ for every such $q$, as well as $\mathrm{spec}(R^{\mathrm{NPL}})$. By the Cauchy integral representation of $A^k$ \citep[Theorem 6.2.28]{Horn_Johnson_1991}, for any $k \ge 0$ and $q \ge \bar{q}$,
    \[
        (R^q)^k = \frac{1}{2\pi i} \oint_{\Gamma} z^k (zI - R^q)^{-1} \, dz,
    \]
hence $\|(R^q)^k\|_F \le \frac{1}{2\pi} \oint_{\Gamma} |z|^k \|(zI - R^q)^{-1}\|_F |dz| \le r^{k+1} \max_{z \in \Gamma} \|(zI-R^q)^{-1}\|_F$.
    Since $R^q \to R^{\mathrm{NPL}}$ from \Cref{lemma: M difference bound} and $\rho(R^{\mathrm{NPL}}),\rho(R^q) < r$ for all $q \ge \bar{q}$, the resolvent $(zI-A)^{-1}$ is jointly continuous in $(z,A)$ on the compact set $\Gamma \times (\{R^q : q \ge \bar{q}\} \cup \{R^{\mathrm{NPL}}\})$, so $M \coloneqq \sup_{q \ge \bar{q}} \max_{z \in \Gamma} \|(zI-R^q)^{-1}\|_F$ is finite. Setting $C_c \coloneqq rM$, a constant not depending on $q$, and using $r \le \rho(R^q)+c$ from above, we obtain
    \[
        \|(R^q)^k\|_F
        \le
        C_c\{\rho(R^q)+c\}^k
        \qquad
        \text{for all } k\ge0 \text{ and all } q \ge \bar{q}.
    \]
    Therefore, the linearized recursion $\delta_q^{(k)}=(R^q)^k\delta^{(0)}$ satisfies
    \[
        \|\delta_q^{(k)}\|
        \le
        C_c\{\rho(R^q)+c\}^k\|\delta^{(0)}\|.
    \]
    If $\varepsilon<C_c\|\delta^{(0)}\|$, then the right-hand side is no larger than $\varepsilon$ whenever
    \[
        k
        \ge
        \left\lceil
        \frac{\log\varepsilon-\log C_c-\log\|\delta^{(0)}\|}
        {\log\{\rho(R^q)+c\}}
        \right\rceil
        =
        K_q(c).
    \]

    For $q \ge \bar{q}$, the choice of $\bar{q}$ gives $\rho(R^{\mathrm{NPL}})+g(q)+c<1$. Since $\rho(R^q)\le \rho(R^{\mathrm{NPL}})+g(q)$,
    \[
        \log\{\rho(R^q)+c\}
        \le
        \log\{\rho(R^{\mathrm{NPL}})+g(q)+c\}
        <0.
    \]
    It follows that $\frac{\log\varepsilon-\log C_c-\log\|\delta^{(0)}\|}{\log\{\rho(R^q)+c\}} \le \frac{\log\varepsilon-\log C_c-\log\|\delta^{(0)}\|}{\log\{\rho(R^{\mathrm{NPL}})+g(q)+c\}}$. Taking ceilings gives the displayed conservative upper bound for $K_q(c)$.
\end{proof}

\subsection{Additional Implementation Details} \label{sec: additional implementation}

\subsubsection{Initializing the EM-NPL($q$) Algorithm} \label{sec: initialization}

\noindent As with any EM algorithm, the quality of the initialization can affect convergence speed and the risk of converging to a local optimum. Sieve maximum likelihood estimation is one approach. For a sieve basis $B_K(x) = (b_1(x),\ldots,b_K(x))^{\top}$ with $K$ terms, define the type-specific sieve CCP as
\begin{equation*}
    P^m_{\mathrm{sieve}}(a|x) = \frac{\exp\bigl(B_K(x)^{\top}\alpha_a^m\bigr)}{\sum_{a' \in \mA} \exp\bigl(B_K(x)^{\top}\alpha_{a'}^m\bigr)}.
\end{equation*}
To initialize $\alpha_a^m$, one can estimate CCPs for each market separately and then cluster them into $M$ groups using $k$-means. Then, for each cluster $m$, $\alpha_a^m$ is initialized by multinomial logit using the data from firms in cluster $m$. Alternatively, $\alpha_a^m$ can be initialized randomly. One then runs the sieve EM algorithm to convergence to obtain the initial $(\bfP^{(0)}, \boldsymbol{\pi}^{(0)})$. To initialize $\boldsymbol{\theta}^{(0)}$, first estimate a pooled single-type model by maximum likelihood. For each type $m$, use the pooled estimate as the starting value and perform one weighted pseudo-likelihood M-step, using the posterior type probabilities from the sieve MLE as weights and $P_{\mathrm{sieve}}^m$ in the policy valuation mapping. This yields $\theta^{m,(0)}$. Given $(P^{m,(0)},\theta^{m,(0)})$, solve the selected fixed-point equation to full convergence for each type to obtain $Y^{m,(0)}$. If the mapping depends on the previous parameter iterate, set $\tilde{\theta}^{m,(0)}=\theta^{m,(0)}$.

\subsubsection{GMRES Algorithm} \label{sec: gmres algorithm}

\noindent Consider the linear system $AY=b$, where $A\in\mathbb{R}^{n\times n}$ is nonsingular and $b\in\mathbb{R}^n$. Starting from $Y^{(0)}\in\mathbb{R}^n$, standard GMRES applies the following Arnoldi recursion whenever the initial residual is nonzero.

\begin{itemize}
    \item \textbf{Input:} $A\in\mathbb{R}^{n\times n}$, $b\in\mathbb{R}^n$, initial guess $Y^{(0)}\in\mathbb{R}^n$, and inner iteration count $q$.
    \item \textbf{Step 0:} Compute $r_0=b-AY^{(0)}$ and $\rho_0=\|r_0\|_2$, and set $v_1=r_0/\rho_0$ when $\rho_0>0$.
    \item \textbf{Step 1 (Arnoldi process):} For $j=1,\ldots,q$:
    \begin{itemize}
        \item[(a)] Set $w=Av_j$.
        \item[(b)] For $i=1,\ldots,j$, set $h_{ij}=v_i^{\top}w$ and update $w\leftarrow w-h_{ij}v_i$.
        \item[(c)] Set $h_{j+1,j}=\|w\|_2$ and $v_{j+1}=w/h_{j+1,j}$ when $h_{j+1,j}>0$; otherwise terminate.
    \end{itemize}
    \item \textbf{Step 2 (Least-squares problem):} Let $\bar H_q\in\mathbb{R}^{(q+1)\times q}$ contain the coefficients $h_{ij}$ and let $e_1=(1,0,\ldots,0)^{\top}\in\mathbb{R}^{q+1}$. Compute
    \[
        y_q=\argmin_{y\in\mathbb{R}^q}\|\rho_0e_1-\bar H_qy\|_2.
    \]
    \item \textbf{Step 3 (Solution update):} Set
    \[
        Y^{(q)}=Y^{(0)}+V_qy_q,
        \quad
        V_q=[v_1,\ldots,v_q]\in\mathbb{R}^{n\times q}.
    \]
    \item \textbf{Output:} $\Gamma_{\mathrm{GMRES}}^q(\theta,Y^{(0)},P,\tilde{\theta})\coloneqq Y^{(q)}$.
\end{itemize}

For matrix right-hand sides $AX=B$, the implementation applies this vector algorithm to $\operatorname{vec}(X)$; the Euclidean inner product and norm then equal their Frobenius counterparts used in the code.

The Arnoldi normalization is undefined at $r_0=0$, so we smooth the numerical GMRES map locally. Let $\varepsilon_{\mathrm{in}}$ be its stopping tolerance, set $\tau=0.01\varepsilon_{\mathrm{in}}$, and let $s(u)$ equal zero for $u\leq0$, one for $u\geq1$, and $u^4(35-84u+70u^2-20u^3)$ otherwise. When $\rho_0\leq2\tau$, compute $q$ fallback iterations
\[
    Y_F^{(\ell+1)}=Y_F^{(\ell)}+D\{b-AY_F^{(\ell)}\},
    \quad
    Y_F^{(0)}=Y^{(0)},
\]
where $D=I$ gives Richardson iteration for policy valuation and $D=\omega A^{\top}$ gives Landweber iteration for EPL, with $0<\omega<2/\|A\|_2^2$, and return
\[
    Y_{\mathrm{smooth}}^{(q)}
    =
    \bigl[1-s\{(\rho_0-\tau)/\tau\}\bigr]Y_F^{(q)}
    +s\{(\rho_0-\tau)/\tau\}Y^{(q)}.
\]
Standard numerical GMRES is used when $\rho_0>2\tau$. At an exact fixed point, $\rho_0<\tau$ throughout some neighborhood. Hence, if $A$, $b$, $Y^{(0)}$, and $D$ are three times continuously differentiable locally, so is the smoothed map; it also leaves the fixed point unchanged.

\subsubsection{Kronecker Structure} \label{sec: kronecker}

\noindent Suppose $x_t=(x_{1t},\ldots,x_{Kt})$ has $K$ independently evolving components, each with $L$ grid points and transition matrix $F_k\in\mathbb{R}^{L\times L}$. The joint state space has $L^K$ points, and its transition matrix is
\begin{equation*}
    F = F_1 \otimes F_2 \otimes \cdots \otimes F_K.
\end{equation*}
The inner solvers compute $Fv$ without forming the full $L^K\times L^K$ matrix. For $v \in \bR^{L^K}$, let $\mV = \operatorname{mat}(v) \in \bR^{L \times \cdots \times L}$ denote the order-$K$ tensor whose $(x_1, \ldots, x_K)$ entry is the component of $v$ indexed by state $(x_1, \ldots, x_K)$, with $x_1$ varying slowest, consistent with the index ordering of $F_1\otimes\cdots\otimes F_K$. Let $\mV \times_k F_k$ denote the mode-$k$ product, which contracts the $k$-th index of $\mV$ with $F_k$ \citep{kolda2009tensor}. Then
\[
    Fv = \operatorname{vec}\!\left(
      \mV \times_1 F_1 \times_2 F_2 \cdots \times_K F_K \right),
\]
where $\operatorname{vec} = \operatorname{mat}^{-1}$, and these products are applied sequentially for $k = 1,\ldots,K$. The savings grow exponentially with the number of state dimensions. The cost per matrix-vector product falls from $\mathcal{O}(L^{2K})$ to $\mathcal{O}(K L^{K+1})$, and memory from $\mathcal{O}(L^{2K})$ to $\mathcal{O}(K L^{2})$.

\clearpage
\section{Additional Simulation Results} \label{sec: additional simulation results}

\begin{table}[!htbp]
\centering
\begin{threeparttable}
\caption{Mixture Model Results ($M=3$, $N=5000$, $\beta=0.80$)}
\small
\setlength{\tabcolsep}{4pt}
\label{tab:mixture_beta_0_80}
\begin{tabular}{ll|cccc|cccc}
\toprule
 & & \multicolumn{4}{c|}{Finite Dependence} & \multicolumn{4}{c}{Non-Finite Dependence} \\
\cmidrule(lr){3-6} \cmidrule(lr){7-10}
Method & $q$ & CT (s) & MSE & Avg Iter & Conv (\%) & CT (s) & MSE & Avg Iter & Conv (\%) \\
\midrule
PV\_GMRES & 4 & \phantom{0}2.84 (0.46) & 0.024 & \phantom{0}8.7 & 100.0 & \phantom{0}3.28 (0.28) & 0.034 & \phantom{0}6.0 & 100.0 \\
PV\_GMRES & 6 & \phantom{0}2.98 (0.49) & 0.024 & \phantom{0}8.6 & 100.0 & \phantom{0}3.34 (0.43) & 0.034 & \phantom{0}6.2 & 100.0 \\
PV\_GMRES & 8 & \phantom{0}3.15 (0.49) & 0.024 & \phantom{0}8.5 & 100.0 & \phantom{0}3.50 (0.49) & 0.034 & \phantom{0}6.2 & 100.0 \\
PV\_GMRES & 10 & \phantom{0}3.35 (0.57) & 0.024 & \phantom{0}8.5 & 100.0 & \phantom{0}3.64 (0.48) & 0.034 & \phantom{0}6.2 & 100.0 \\
\midrule
PV\_SA & 4 & \phantom{0}2.95 (0.47) & 0.024 & \phantom{0}8.7 & 100.0 & \phantom{0}3.37 (0.31) & 0.034 & \phantom{0}6.0 & 100.0 \\
PV\_SA & 6 & \phantom{0}3.14 (0.54) & 0.024 & \phantom{0}8.6 & 100.0 & \phantom{0}3.54 (0.60) & 0.034 & \phantom{0}6.4 & 100.0 \\
PV\_SA & 8 & \phantom{0}3.34 (0.59) & 0.024 & \phantom{0}8.6 & 100.0 & \phantom{0}3.64 (0.44) & 0.034 & \phantom{0}6.2 & 100.0 \\
PV\_SA & 10 & \phantom{0}3.54 (0.64) & 0.024 & \phantom{0}8.5 & 100.0 & \phantom{0}3.84 (0.71) & 0.034 & \phantom{0}6.3 & 100.0 \\
\midrule
BM\_SA & 4 & \phantom{0}5.51 (0.93) & 0.024 & \phantom{0}7.1 & 100.0 & \phantom{0}9.62 (1.24) & 0.034 & \phantom{0}6.9 & \phantom{0}99.8 \\
BM\_SA & 6 & \phantom{0}7.32 (1.28) & 0.024 & \phantom{0}7.2 & 100.0 & \phantom{0}9.83 (0.96) & 0.034 & \phantom{0}5.4 & 100.0 \\
BM\_SA & 8 & \phantom{0}9.11 (1.65) & 0.024 & \phantom{0}7.2 & 100.0 & 12.00 (1.70) & 0.034 & \phantom{0}5.7 & 100.0 \\
BM\_SA & 10 & 10.90 (2.01) & 0.024 & \phantom{0}7.2 & 100.0 & 14.62 (2.00) & 0.034 & \phantom{0}5.7 & 100.0 \\
\midrule
BM\_AD & 4 & \phantom{0}7.84 (1.59) & 0.024 & \phantom{0}7.4 & 100.0 & 11.69 (0.87) & 0.034 & \phantom{0}6.1 & 100.0 \\
BM\_AD & 6 & 10.44 (1.75) & 0.024 & \phantom{0}7.3 & 100.0 & 13.66 (1.03) & 0.034 & \phantom{0}5.5 & 100.0 \\
BM\_AD & 8 & 13.84 (2.06) & 0.024 & \phantom{0}6.7 & 100.0 & 18.55 (1.38) & 0.034 & \phantom{0}5.5 & 100.0 \\
BM\_AD & 10 & 17.01 (2.44) & 0.024 & \phantom{0}5.8 & 100.0 & 23.23 (1.72) & 0.034 & \phantom{0}5.2 & 100.0 \\
\midrule
BM\_NT & 4 & 17.36 (2.90) & 0.024 & \phantom{0}7.1 & 100.0 & 28.40 (2.74) & 0.034 & \phantom{0}5.7 & 100.0 \\
BM\_NT & 6 & 17.41 (2.91) & 0.024 & \phantom{0}7.1 & 100.0 & 28.54 (2.78) & 0.034 & \phantom{0}5.7 & 100.0 \\
BM\_NT & 8 & 17.39 (2.90) & 0.024 & \phantom{0}7.1 & 100.0 & 28.51 (2.76) & 0.034 & \phantom{0}5.7 & 100.0 \\
BM\_NT & 10 & 17.38 (2.88) & 0.024 & \phantom{0}7.1 & 100.0 & 28.50 (2.73) & 0.034 & \phantom{0}5.7 & 100.0 \\
\midrule
EE\_SA & 4 & \phantom{0}3.98 (0.73) & 0.024 & \phantom{0}6.9 & 100.0 & --- & --- & --- & --- \\
EE\_SA & 6 & \phantom{0}5.39 (1.01) & 0.024 & \phantom{0}7.2 & 100.0 & --- & --- & --- & --- \\
EE\_SA & 8 & \phantom{0}6.71 (1.26) & 0.024 & \phantom{0}7.3 & 100.0 & --- & --- & --- & --- \\
EE\_SA & 10 & \phantom{0}7.98 (1.49) & 0.024 & \phantom{0}7.2 & 100.0 & --- & --- & --- & --- \\
\midrule
PV\_GMRES & $\infty$ & \phantom{0}3.62 (0.59) & 0.024 & \phantom{0}8.5 & 100.0 & \phantom{0}4.01 (0.55) & 0.034 & \phantom{0}6.2 & 100.0 \\
BM\_NT & $\infty$ & 23.65 (4.38) & 0.024 & \phantom{0}7.2 & 100.0 & 38.69 (5.11) & 0.034 & \phantom{0}5.7 & 100.0 \\
EE\_SA & $\infty$ & 12.48 (2.41) & 0.024 & \phantom{0}7.2 & 100.0 & --- & --- & --- & --- \\
\bottomrule
\end{tabular}
\vspace{0.5em}
{\footnotesize \textit{Notes:} Results are averaged across 500 simulations ($M=3$ types, $N=5000$, $\beta=0.80$). $q$ denotes the number of inner iterations. CT reports the average computational time in seconds with standard deviation in parentheses. MSE is $\sum_{m=1}^{M}(\|\hat{\theta}_m - \theta_{m,0}\|^2 + \|\hat{\pi}_m - \pi_{m,0}\|^2)$ averaged across simulations. EM convergence is declared when $\max\{\|P^{(k)} - P^{(k-1)}\|_\infty,\, \|\theta^{(k)} - \theta^{(k-1)}\|_\infty,\, \|Y^{(k)} - Y^{(k-1)}\|_\infty / (1 + \|Y^{(k-1)}\|_\infty)\} < 10^{-3}$. The rows with $q=\infty$ show the results when the fixed-point equation is solved to full convergence.}
\end{threeparttable}
\end{table}

\begin{table}[!htbp]
\centering
\begin{threeparttable}
\caption{Mixture Model Results ($M=3$, $N=5000$, $\beta=0.9999$)}
\small
\setlength{\tabcolsep}{4pt}
\label{tab:mixture_beta_0_9999}
\begin{tabular}{ll|cccc|cccc}
\toprule
 & & \multicolumn{4}{c|}{Finite Dependence} & \multicolumn{4}{c}{Non-Finite Dependence} \\
\cmidrule(lr){3-6} \cmidrule(lr){7-10}
Method & $q$ & CT (s) & MSE & Avg Iter & Conv (\%) & CT (s) & MSE & Avg Iter & Conv (\%) \\
\midrule
PV\_GMRES & 4 & \phantom{0}3.39 (0.66) & 0.024 & \phantom{0}8.0 & 100.0 & 10.21 (8.84) & 0.061 & \phantom{0}6.4 & 100.0 \\
PV\_GMRES & 6 & \phantom{0}4.84 (1.03) & 0.024 & 14.6 & 100.0 & 11.69 (8.89) & 0.061 & 12.2 & 100.0 \\
PV\_GMRES & 8 & \phantom{0}3.65 (0.68) & 0.024 & \phantom{0}7.8 & 100.0 & 10.89 (8.83) & 0.061 & \phantom{0}7.8 & 100.0 \\
PV\_GMRES & 10 & \phantom{0}3.80 (0.77) & 0.024 & \phantom{0}7.6 & 100.0 & 10.21 (8.88) & 0.061 & \phantom{0}5.4 & 100.0 \\
\midrule
PV\_SA & 4 & \phantom{0}3.36 (0.65) & 0.024 & \phantom{0}7.5 & 100.0 & 10.03 (8.85) & 0.061 & \phantom{0}5.3 & 100.0 \\
PV\_SA & 6 & \phantom{0}3.54 (0.69) & 0.024 & \phantom{0}7.6 & 100.0 & 10.21 (8.89) & 0.061 & \phantom{0}5.8 & 100.0 \\
PV\_SA & 8 & \phantom{0}3.74 (0.74) & 0.024 & \phantom{0}7.6 & 100.0 & 10.16 (8.81) & 0.061 & \phantom{0}5.4 & 100.0 \\
PV\_SA & 10 & \phantom{0}3.91 (0.74) & 0.024 & \phantom{0}7.6 & 100.0 & 10.27 (8.86) & 0.061 & \phantom{0}5.4 & 100.0 \\
\midrule
BM\_SA & 4 & 33.52 (4.68) & 0.024 & 100.0 & \phantom{00}0.0 & 44.21 (9.15) & 0.061 & 100.0 & \phantom{00}0.0 \\
BM\_SA & 6 & 44.39 (6.32) & 0.024 & 100.0 & \phantom{00}0.0 & 53.64 (9.47) & 0.061 & 100.0 & \phantom{00}0.0 \\
BM\_SA & 8 & 55.12 (7.92) & 0.024 & 100.0 & \phantom{00}0.0 & 62.79 (9.79) & 0.061 & 100.0 & \phantom{00}0.0 \\
BM\_SA & 10 & 66.49 (9.74) & 0.024 & 100.0 & \phantom{00}0.0 & 73.70 (10.29) & 0.061 & 100.0 & \phantom{00}0.0 \\
\midrule
BM\_AD & 4 & 191.09 (34.45) & 0.029 & 100.0 & \phantom{00}0.0 & 358.86 (46.46) & 0.278 & 100.0 & \phantom{00}0.0 \\
BM\_AD & 6 & 233.46 (84.66) & 0.025 & 71.1 & \phantom{0}75.2 & 461.05 (66.05) & 0.079 & 96.3 & \phantom{0}24.2 \\
BM\_AD & 8 & 39.74 (10.52) & 0.024 & 11.3 & 100.0 & 137.31 (36.71) & 0.061 & 22.8 & 100.0 \\
BM\_AD & 10 & 28.35 (6.36) & 0.024 & \phantom{0}7.4 & 100.0 & 65.74 (16.61) & 0.061 & \phantom{0}8.8 & 100.0 \\
\midrule
BM\_NT & 4 & 26.22 (3.78) & 0.024 & \phantom{0}7.3 & 100.0 & 53.99 (10.92) & 0.061 & \phantom{0}5.4 & 100.0 \\
BM\_NT & 6 & 26.29 (3.78) & 0.024 & \phantom{0}7.3 & 100.0 & 54.16 (10.96) & 0.061 & \phantom{0}5.4 & 100.0 \\
BM\_NT & 8 & 26.35 (3.88) & 0.024 & \phantom{0}7.3 & 100.0 & 54.00 (10.84) & 0.061 & \phantom{0}5.4 & 100.0 \\
BM\_NT & 10 & 26.33 (3.87) & 0.024 & \phantom{0}7.3 & 100.0 & 53.98 (10.95) & 0.061 & \phantom{0}5.4 & 100.0 \\
\midrule
EE\_SA & 4 & \phantom{0}4.73 (0.87) & 0.024 & \phantom{0}7.0 & 100.0 & --- & --- & --- & --- \\
EE\_SA & 6 & \phantom{0}6.22 (1.08) & 0.024 & \phantom{0}7.3 & 100.0 & --- & --- & --- & --- \\
EE\_SA & 8 & \phantom{0}7.64 (1.22) & 0.024 & \phantom{0}7.3 & 100.0 & --- & --- & --- & --- \\
EE\_SA & 10 & \phantom{0}9.05 (1.47) & 0.024 & \phantom{0}7.4 & 100.0 & --- & --- & --- & --- \\
\midrule
PV\_GMRES & $\infty$ & \phantom{0}4.58 (0.84) & 0.024 & \phantom{0}7.6 & 100.0 & 10.98 (8.87) & 0.061 & \phantom{0}5.4 & 100.0 \\
BM\_NT & $\infty$ & 38.01 (6.47) & 0.024 & \phantom{0}7.3 & 100.0 & 76.33 (13.45) & 0.061 & \phantom{0}5.4 & 100.0 \\
EE\_SA & $\infty$ & 16.23 (2.60) & 0.024 & \phantom{0}7.4 & 100.0 & --- & --- & --- & --- \\
\bottomrule
\end{tabular}
\vspace{0.5em}
{\footnotesize \textit{Notes:} Results are averaged across 500 simulations ($M=3$ types, $N=5000$, $\beta=0.9999$). $q$ denotes the number of inner iterations. CT reports the average computational time in seconds with standard deviation in parentheses. MSE is $\sum_{m=1}^{M}(\|\hat{\theta}_m - \theta_{m,0}\|^2 + \|\hat{\pi}_m - \pi_{m,0}\|^2)$ averaged across simulations. EM convergence is declared when $\max\{\|P^{(k)} - P^{(k-1)}\|_\infty,\, \|\theta^{(k)} - \theta^{(k-1)}\|_\infty,\, \|Y^{(k)} - Y^{(k-1)}\|_\infty / (1 + \|Y^{(k-1)}\|_\infty)\} < 10^{-3}$. The rows with $q=\infty$ show the results when the fixed-point equation is solved to full convergence.}
\end{threeparttable}
\end{table}

\end{document}